\documentclass{amsart}

\usepackage{amssymb,amsmath,amsthm,bbm}
\usepackage{mathrsfs}         
\usepackage{graphicx}
\usepackage{hyperref}
\usepackage{upgreek}
\usepackage{enumerate}

\usepackage{a4wide}

\theoremstyle{plain}
\newtheorem{thm}{Theorem}
\newtheorem{prop}[thm]{Proposition}
\newtheorem{lemma}[thm]{Lemma}
\newtheorem{cor}[thm]{Corollary}
\theoremstyle{definition}
\newtheorem{definition}[thm]{Definition}
\newtheorem{remark}[thm]{Remark}

\newcommand{\tn}[1]{\ensuremath{\mathbb{T}^{#1}}}
\newcommand{\rn}[1]{\ensuremath{\mathbb{R}^{#1}}}
\newcommand{\nn}[1]{\ensuremath{\mathbb{N}^{#1}}}
\newcommand{\zn}[1]{\ensuremath{\mathbb{Z}^{#1}}}
\newcommand{\sn}[1]{\ensuremath{\mathbb{S}^{#1}}}

\newcommand{\ldr}[1]{\langle #1\rangle}
\newcommand{\g}{\gamma}
\renewcommand{\b}{\beta}
\newcommand{\e}{\epsilon}
\newcommand{\bK}{\bar{K}}
\newcommand{\bge}{\bar{g}}
\newcommand{\bp}{\bar{p}}
\newcommand{\bR}{\bar{R}}
\newcommand{\br}{\bar{r}}
\newcommand{\bvr}{\bar{\varrho}}
\newcommand{\vare}{\varepsilon}
\newcommand{\bPhi}{\bar{\Phi}}
\newcommand{\Phia}{\Phi_1}
\newcommand{\Phib}{\Phi_0}
\newcommand{\bmu}{\bar{\mu}}
\newcommand{\bS}{\bar{S}}
\newcommand{\bk}{\bar{k}}
\newcommand{\bM}{\bar{M}}
\newcommand{\bx}{\bar{x}}
\newcommand{\bX}{\bar{X}}
\newcommand{\bY}{\bar{Y}}
\newcommand{\bc}{\bar{c}}
\newcommand{\bd}{\bar{d}}
\newcommand{\refer}{\mathrm{ref}}
\newcommand{\rorel}{\mathrm{rel}}
\newcommand{\rood}{\mathrm{od}}
\newcommand{\rodiv}{\mathrm{div}}
\newcommand{\roin}{\mathrm{in}}
\newcommand{\roout}{\mathrm{out}}
\newcommand{\rond}{\mathrm{nd}}
\newcommand{\roK}{\mathrm{K}}
\newcommand{\Spe}{\mathrm{Sp}}
\newcommand{\bh}{\bar{h}}
\newcommand{\bg}{\bar{g}}
\newcommand{\bD}{\bar{D}}
\newcommand{\bnabla}{\overline{\nabla}}
\newcommand{\chK}{\check{K}}
\newcommand{\hG}{\hat{G}}
\newcommand{\hU}{\hat{U}}
\newcommand{\hN}{\hat{N}}
\newcommand{\hg}{\hat{g}}
\newcommand{\hh}{\hat{h}}
\newcommand{\chh}{\check{h}}
\newcommand{\chb}{\check{b}}
\newcommand{\chmu}{\check{\mu}}
\newcommand{\hml}{\hat{\mathcal{L}}}
\newcommand{\ml}{\mathcal{L}}
\newcommand{\mA}{\mathcal{A}}
\newcommand{\me}{\mathcal{E}}
\newcommand{\mK}{\mathcal{K}}
\newcommand{\mP}{\mathcal{P}}
\newcommand{\mW}{\mathcal{W}}
\newcommand{\cweight}{\mathfrak{u}}
\newcommand{\bmR}{\bar{\mathcal{R}}}
\newcommand{\bmS}{\bar{\mathcal{S}}}
\newcommand{\bmN}{\bar{\mathcal{N}}}
\newcommand{\msK}{\mathscr{K}}
\newcommand{\msO}{\mathscr{O}}
\newcommand{\msX}{\mathscr{X}}
\newcommand{\bmsX}{\bar{\mathscr{X}}}
\newcommand{\hmsX}{\hat{\mathscr{X}}}
\newcommand{\msY}{\mathscr{Y}}
\newcommand{\bmsY}{\bar{\mathscr{Y}}}
\newcommand{\hmsY}{\hat{\mathscr{Y}}}

\newcommand{\bsfx}{\bar{\mathsf{x}}}
\newcommand{\tr}{\mathrm{tr}}
\renewcommand{\d}{\partial}
\newcommand{\bfI}{\mathbf{I}}
\newcommand{\bfJ}{\mathbf{J}}
\newcommand{\bfK}{\mathbf{K}}
\newcommand{\bfL}{\mathbf{L}}

\begin{document}

\author{Hans Ringstr\"{o}m}
\title{Initial data on big bang singularities}
\begin{abstract}
  We introduce a geometric notion of initial data on the singularity for the Einstein-scalar field equations and demonstrate that
  previous notions of data on the singularity constitute special cases. The definition thus gives a unified geometric perspective on
  existing results. The formulation also leads to a natural geometric initial value formulation of Einstein's equations with initial data on the
  singularity. However, in order for any notion of initial data on the singularity to parametrise convergent solution, the crucial question is: do
  convergent solutions necessarily induce such data? There are several notions of data on the singularity for the Einstein-scalar field
  equations and corresponding existence results. The existence results demonstrate that the corresponding notions are sufficient to ensure
  convergent behaviour. However, none of the existing notions of data on the singularity are necessary. A central part of the article is
  therefore to demonstrate the necessity of our requirements. 
\end{abstract}

\maketitle

\section{Introduction}

In this article, we are interested in the \textit{Einstein-scalar field equations} with a \textit{cosmological constant} $\Lambda$; i.e., the equations
\begin{subequations}
  \begin{align}
    G+\Lambda g = & T,\\
    \Box_g\phi = & 0
  \end{align}
\end{subequations}
for a Lorentz metric $g$ and a smooth function $\phi$ (the scalar field) on a manifold $M$. Here
\[
G:=\mathrm{Ric}-\frac{1}{2}Sg
\]
is the \textit{Einstein tensor}, $\mathrm{Ric}$ is the \textit{Ricci tensor} of $g$ and $S$ is the \textit{scalar curvature} of $g$. Moreover,
$\Lambda\in\rn{}$, $\Box_g$ is the wave operator associated with $g$ and
\[
T=d\phi\otimes d\phi-\tfrac{1}{2}|d\phi|_g^2\cdot g.
\]
We are also interested in \textit{Einstein's vacuum equations}; i.e. the equation $G=0$ (when we speak of Einstein's vacuum equations we, in other
words, assume a vanishing cosmological constant). 

In order to explain the origin of the notion of initial data on the singularity used in, e.g., \cite{aarendall,damouretal}, it is natural to return
to the ideas of
Belinski\v{\i}, Khalatnikov and Lifschitz (BKL for short). One fundamental idea in the BKL proposal concerning the nature of cosmological singularities
is that kinetic terms (roughly speaking time derivative terms) should dominate over spatial curvature terms (roughly speaking spatial derivative terms).
In \cite{ELS}, the authors introduce a notion of velocity dominated singularity by comparing, asymptotically, an actual solution with a solution to
a truncated set of equations, obtained by working in geodesic normal coordinates and dropping spatial derivative terms. The solutions to the truncated
set of equations can, in some sense, be thought of as initial data on the singularity. These ideas were refined in \cite{JimandVince}, in which the
authors introduce the notion of a solution being \textit{asymptotically velocity term dominated near the singularity} (AVTDS); see \cite[p.~88]{JimandVince}.
Again, a solution is AVTDS if it asymptotes to a solution of truncated system of equations, the \textit{velocity
  term dominated} (VTD) system; see \cite[(3) and (4), pp.~87--88]{JimandVince}. In this sense, solutions to the VTD system constitute initial data on
the singularity. However, as pointed out in \cite{JimandVince}, the VTD system has no covariant formulation. In this sense, the notion of initial data
on the singularity is, in this case, not geometric in nature. The exact norm to be used to compare an actual solution to a solution to the VTD system
is also left unspecified in the definition of AVTDS solutions. Next, in \cite{aarendall}, the authors write down the velocity dominated system 
in the case of the Einstein-scalar field equations and the Einstein-stiff fluid equations in $3+1$-dimensions using a Gaussian foliation; see
\cite[(10)--(14), p.~483]{aarendall}. Given real analytic solutions to the velocity dominated system (satisfying certain conditions), the authors
then prove that there are real analytic solutions to the actual equations
asymptoting to the prescribed solutions to the velocity dominated system. This is a very important result. On the other hand, it is not quite clear that
the notion of initial data on the singularity is geometric. More importantly, it is not clear that convergent solutions necessarily induce the type of
data introduced in \cite{aarendall}. In fact, the conclusions of \cite{fal} (see also \cite{klinger}) imply that the conditions of \cite{aarendall} are
not necessary (note that the
vacuum equations considered in \cite{fal} are of course a special case of the Einstein-scalar field equations considered in \cite{aarendall}). In
\cite{fal}, the authors specify initial data on a singularity with $\tn{3}$-topology for the Einstein vacuum equations and prove that there are
corresponding solutions to the Einstein vacuum equations converging to the data at the singularity. In this case it is not so clear that the formulation
is geometric in nature. Viewing the result of \cite{fal} as a result concerning the Einstein-scalar field equations, it is also clear from \cite{aarendall}
that the conditions are not necessary. However, even considering the result \cite{fal} as a result concerning the vacuum equations and ignoring the
restriction to $\tn{3}$-spatial topology, it is not clear why the conditions should be necessary. In addition to \cite{aarendall,damouretal,klinger,fal}, there
are several results in situations with symmetry; see, e.g., \cite{kar,iak,ren,iam,sta,ABIF,aeta}. In this case, the results \cite{aarendall,damouretal,klinger,fal}
illustrate that the relevant notions of initial data on the singularity are not necessary, since \cite{aarendall,damouretal,klinger,fal} apply
in the absence of symmetries. Moreover, the data on the singularity in the symmetric settings typically correspond to the asymptotics of solutions to
a symmetry reduced version of the equations. It is therefore not clear that the corresponding notions of data on the singularity are geometric.

In spite of the progress represented by the above results, it is clear that none of the corresponding notions of data on the singularty fulfill the
demands one would like to make:
\begin{enumerate}[(i)]
\item The notion of initial data on the singularity should be geometric in nature, and there should be a corresponding geometric initial value problem.
\item The notions of data used in the existing results should all be special cases of this notion.  
\item Convergent solutions should necessarily induce initial data on the singularity.
\end{enumerate}
Since Einstein's equations are geometric (they only determine the isometry class of the solution), the notion of data on the singularity should clearly
be geometric. Moreover, data on the singularity should give rise to unique corresponding maximal globally hyperbolic developments, and isometric
data on the singularity should give rise to isometric developments. Second, we clearly want the notion to encompass all previous notions. Finally,
it is only if one can prove the necessity of the notion of initial data introduced that one has a hope of parametrising convergent solutions by data on
the singularity. Moreover, if one is not able to prove necessity, there could be a large number of other notions of data on the singularity that
could be equally viable.

That initial data on the singularity, as defined here, give rise to a unique corresponding maximal globally hyperbolic development and that isometric
data on the singularity give rise to isometric developments is verified in \cite{grisales} for the Einstein-non-linear scalar field equations in
$3+1$-dimension. Needless to say, it would be of interest to verify this statement more generally. Nevertheless, in this article, we focus on the
remaining requirements in the above list. Combining the existence result in \cite{fal} with the geometric perspective developed in\cite{RinGeo}, we
formulate a notion of initial data on the singularity in the case of the $3+1$-dimensional vacuum equations in
Definition~\ref{def:vacuumidos}. This definition is then generalised to the setting of the $n+1$-dimensional Einstein-scalar field equations with a
cosmological constant; see Definition~\ref{def:ndsfidonbbs}. That the notion of data on the singularity used in \cite{fal} is a special case of 
Definition~\ref{def:vacuumidos} is verified in Proposition~\ref{prop:falaregeoidos}. That it is locally possible to go in the other direction is
verified in Proposition~\ref{prop:geometrictofal}. In Subsection~\ref{ssection:stablequiescentregime} we verify that the VTD solutions introduced in
\cite{aarendall} can be considered to be initial data on the singularity in the sense of Definition~\ref{def:ndsfidonbbs} in the non-degenerate setting
on which we focus here.

Even though the above definitions and observations are of interest, the main focus of this article is on the question of necessity. To address this
question, we first have to decide on a minimal requirement for convergence. 

\textbf{Minimal requirements, convergence.} To make sense of limits, we first we need a foliation of the causal past of a Cauchy hypersurface. Next,
since the limit should correspond to a cosmological singularity, it is natural to assume that the volume density, see
(\ref{eq:varphidef}) below, converges to zero uniformly or that the mean curvature diverges to infinity uniformly (i.e., that the foliation is crushing) 
in the direction of the singularity. Turning to the quantities that should converge, it is useful to recall the BKL proposal; i.e., that time
derivatives dominate over spatial derivatives. In other words, we expect the second fundamental form (roughly speaking the normal derivative of the
induced metric on the leaves of the foliation) and the normal derivative of the scalar field to constitute the leading order terms in describing the
asymptotics in the direction of the singularity. And if the asymptotics are supposed to be convergent, it is natural to expect the leading order terms
to converge (after an appropriate rescaling, if necessary). More specifically, let $\bK$ be the Weingarten map of the leaves of the foliation (i.e., the
second fundamental form with one index raised) and let $\theta$ denote the mean curvature of the leaves of the foliation. Since we are interested in
situations in which $\theta=\mathrm{tr}\bK$ diverges to infinity, it is clear that we need to expansion normalise $\bK$. 
A natural object to consider is then the expansion normalised Weingarten map, defined by $\mK:=\bK/\theta$. Concerning
this object, there is a clear dichotomy. There are solutions for which the eigenvalues of $\mK$ oscillate in the direction of the singularity;
see, e.g., \cite{cbu}. In fact, the asymptotics for the eigenvalues can even exhibit chaotic dynamics; see \cite{du}. This is a possibility we wish to
exclude here. Moreover, to the best of our knowledge, in all the results deducing convergent asymptotics, in particular in all the results concerning
data on the singularity, $\mK$ converges. It is therefore natural to require $\mK$ to converge to
a limit, say $\msK$. Similarly, it is natural to assume that an expansion normalised version of the normal derivative
of the scalar field converges; if $U$ denotes the future directed unit vector field normal to the leaves of the foliation and $\hU:=\theta^{-1}U$,
then we require $\hU\phi$ to converge to, say, $\Phi_1$. Again, to the best of our knowledge, in all the results deducing convergent asymptotics, in
particular in all the results concerning data on the singularity, this is satisfied.

It is clear that there is a big gap between these minimal requirements for convergence and the notions of initial data on the singularity introduced in
Definition~\ref{def:ndsfidonbbs} and \cite{aarendall,damouretal,fal}. 

\textbf{Additional requirements, data on the singularity.} Clearly, $\msK$ and $\Phi_1$ do not contain enough information to constitute initial data
on the singularity. In Definition~\ref{def:ndsfidonbbs}, there are two additional pieces of information, namely $\bh$ and $\Phi_0$. Here $\bh$
constitutes the limit of anisotropic rescalings of the metrics induced on the leaves of the foliation. Moreover, $\Phi_0$ is a lower order
term in the asymptotic expansion for $\phi$. Similar information constitute part of the initial data on the singularity in \cite{aarendall,damouretal,fal};
${}^{0}g_{ab}$ and ${}^0\phi$ in the case of \cite{aarendall,damouretal}, see \cite[Subsection~2.2, pp.~482--484]{aarendall} and
\cite[Subsection~2.2, pp.~1060-1061]{damouretal}, and $c_{ij}$ in the case of \cite{fal}, see \cite[(1.4) and (1.5), p.~1185]{fal}. Next,
the existence and uniqueness results in \cite{aarendall,damouretal,fal} require specific anisotropic convergence rates for the components of the expansion
normalised Weingarten map; see \cite[(41d), p.~492]{aarendall}, \cite[(2.33), p.~1065]{damouretal} and \cite[(1.11), p.~1188]{fal}. In addition, the data
on the singularity are required to satisfy constraint equations; see \cite[(10), p.~483]{aarendall},
\cite[(2.18), (2.19), p.~1062]{damouretal}, conditions (2) and (5) of \cite[Theorem~1.1, pp.~1185--1186]{fal} and 
condition (2) of Definition~\ref{def:ndsfidonbbs}. Finally, there are conditions on the eigenvalues, say $\{p_A\}$, of $\msK$ and the corresponding
eigenvector fields. One possibility is to impose the condition that the eigenvalues satisfy $1+p_A-p_B-p_C>0$ for all $A,B,C$ with $B\neq C$, see
Definition~\ref{def:ndsfidonbbs} and \cite[Theorem~10.1, p.~1104]{damouretal}; in $3+1$-dimensions, this requirement translates to $p_A>0$ for all $A$,
which is the condition appearing in \cite[Theorem~1, p.~484]{aarendall}. The second possibility is to demand that certain structure coefficients of the
eigenvector fields of $\msK$ vanish if the conditions on the eigenvalues are not satisfied; see condition (4) of Definition~\ref{def:ndsfidonbbs} and
note that a corresponding condition is implicitly part of the requirements in \cite{fal}.

If $\hU\phi$ converges at an appropriate rate, one could imagine integrating this estimate in order to obtain a lower order term in the asymptotic
expansion for $\phi$, say $\Phi_0$. The picture concerning the geometry is somewhat more complicated. Why should the specific anisotropic rescalings
of the metrics induced on the leaves of the foliation (see Definition~\ref{def:developmentsfLambda} and \cite{aarendall,damouretal,fal}) converge? Similarly,
why should the expansion normalised Weingarten map satisfy the stated anisotropic convergence rates? These are technical conditions going quite far
beyond the minimal requirements of convergence. Turning to the constraint equations on the singularity, one of them follows naturally from the
ideas of BKL, namely $\tr\msK^2+\Phi_1^2=1$. In fact, this equation is obtained by taking the limit of expansion normalised versions of the
Hamiltonian constraint equations induced on the leaves of the foliation. In order to justify this statement, note that the
Hamiltonian constraint can be written as in (\ref{eq:Hamiltonianconstraint}). Again, going back to the ideas of BKL, the spatial derivatives
should be less important than the time derivatives. In particular, the expansion normalised spatial scalar curvature $\bS/\theta^2$ (where $\bS$
denotes the scalar curvature of the leaves of the foliation) should tend to zero and the expansion normalised version of the spatial derivatives
of the scalar field should converge to zero. Finally, since we expect the mean curvature to diverge uniformly, we expect $\Lambda/\theta^2$
to converge to zero. Combining these observations with (\ref{eq:Hamiltonianconstraint}) and taking the limit leads to the conclusion that
$\tr\msK^2+\Phi_1^2=1$. This is one of the constraint equations for the data on the singularity; see Definition~\ref{def:ndsfidonbbs}. It is less
obvious why it should be necessary to demand that the second constraint equation, namely $\rodiv_{\bh}\msK=\Phi_1d\Phi_0$, hold. Finally, it is
far from clear why the conditions on the eigenvalues of $\msK$ and the structure coefficients of the eigenvector fields of $\msK$ should be necessary.

\textbf{Necessity.} In Subsection~\ref{ssection:contodataonsing}, we address the question of necessity. In order to give a rough idea of some of the
results, assume the foliation to be of the form $\bM\times I$, where $I=(0,t_+)$ and $t=0$ represents the singularity. We assume the lapse function to
be strictly positive and the shift vector field to vanish. Next, let $\bge_{\refer}$ be a reference Riemannian metric on $\bM$ and $\bge$ denote the metrics
induced on the leaves of the foliation. Then the \textit{volume density} $\varphi$ is defined by the relation
\begin{equation}\label{eq:varphidef}
  \varphi\mu_{\bge_{\refer}}=\mu_{\bge},
\end{equation}
where $\mu_{\bge_{\refer}}$ and $\mu_{\bge}$ denote the volume forms associated with $\bge_{\refer}$ and $\bge$ respectively; here we view $\bge$ as a family
of Riemannian metrics on $\bM$. The \textit{logarithmic volume density} is defined by $\varrho:=\ln\varphi$. We assume that $\varrho$ diverges uniformly
to $-\infty$ in the direction of the singularity. In order
to be able to make conclusions, it is not sufficient to assume convergence, we also need a rate. Since we also wish to have assumptions that do not
involve a limit, we prefer to impose decay on $\hml_{U}\mK$, the expansion normalised normal derivative of $\mK$. The most basic bound we impose is,
in fact, that
\[
|\hml_{U}\mK|_{\bge_{\refer}}\leq C_0 \ldr{\varrho}^{a_0}e^{\e\varrho},
\]
where $\ldr{\varrho}:=(\varrho^2+1)^{1/2}$; $C_0$ and $a_0$ are constants; and $\e:\bM\rightarrow (0,1)$ has a strictly positive lower bound.
From this bound, it can be deduced that $\mK$ converges to a limit $\msK$, and we here assume the eigenvalues of $\msK$ to be distinct. Under these,
quite weak, assumptions, it can be deduced that the relevant anisotropic rescaling of the metric converges to a limit in $C^0$; see
(\ref{eq:gbABgenasform}) and (\ref{eq:bABasymptotics}) below. Moreover, the improved anisotropic convergence rates for the different components
of $\mK$ and $\hml_{U}\mK$ can be demonstrated; see (\ref{eq:mKmsYAmsXBgenest}) and (\ref{eq:hmlUfixedframeopt}). In order to obtain conclusions
for spatial derivatives, it is necessary to impose conditions on the relative spatial derivatives of $\hN:=\theta N$, where $N$ denotes the
lapse function; see (\ref{eq:bDllnhNconditions}) below (note that our requirement allows polynomial growth of the relative spatial derivatives,
whereas the relative spatial derivatives are assumed to decay exponentially in \cite{fal}; see \cite[(1.11), p.~1188]{fal}). In fact, combining
the above assumptions with (\ref{eq:CkexpdechmlUmK}) and
(\ref{eq:bDllnhNconditions}) yields the same conclusions concerning the anisotropic rescalings of the metric and anisotropic convergence
rates as before, but for higher $C^k$-norms. Similarly, imposing decay bounds on $\hU^2\phi$, see
(\ref{eq:bDlhUsqphiestimate}), yields the desired asymptotics for the scalar field. Finally, given the above assumptions; that
(\ref{eq:CkexpdechmlUmK}), (\ref{eq:bDlhUsqphiestimate}) and (\ref{eq:bDllnhNconditions}) hold with $k=2$; and that $\tr\msK^2+\Phi_1^2=1$,
we deduce that all the other conditions of Definition~\ref{def:ndsfidonbbs} follow as consequences; see Theorem~\ref{thm:reprod}. In this sense,
convergent solutions necessarily induce data on the singularity in the sense of Definition~\ref{def:ndsfidonbbs}. Moreover, for this reason,
it is natural to expect that data on the singularity parametrise convergent solutions and that Definition~\ref{def:ndsfidonbbs} is the natural
notion of data on the singularity (there should not be yet another notion which is not a special case of Definition~\ref{def:ndsfidonbbs}). 

\textbf{The oscillatory setting.} 
One context in which the notion of initial data on the singularity (and, in particular, the necessity of that notion) is potentially important,
is in the case of oscillatory and spatially inhomogeneous big bang singularities. Considering the existing arguments in the spatially homogeneous
setting, a crucial first step in the study of oscillatory behaviour is to understand how solutions approach the Kasner circle along a stable
manifold, and then depart via an unstable manifold. In order to carry out a similar analysis in the spatially inhomogeneous setting, it is of
central importance to first identify the stable manifold. Due to the above observations, it is natural to expect that this stable manifold should
correspond to a subset of vacuum data on the singularity introduced in Definition~\ref{def:vacuumidos}; see
Subsection~\ref{ssection:oscibbsing} below for a more detailed discussion of this aspect.  

\textbf{Outline.} The outline of the introduction is the following. In Subsection~\ref{ssection:fal}, we recall the results of \cite{fal}.
In Subsection~\ref{ssection:geometricformfv}, we provide a geometric notion of initial data on the singularity in the $3+1$-dimensional
vacuum setting. We also prove that this notion is locally equivalent to the one introduced in \cite{fal}. In
Subsection~\ref{ssection:geo developments}, we introduce a corresponding geometric notion of development. In Subsection~\ref{ssection:generalisations},
we then generalise the notion of data on the singularity and developments to higher
dimensions and scalar field matter. In Subsection~\ref{ssection:stablequiescentregime}, we verify that the notion of VTD solution used in
\cite{aarendall} as a substitute for initial data on the singularity can be considered to be a special case of the data on the singularity introduced
here. In Subsection~\ref{ssection:imprasympt}, we turn to the question of improving the asymptotics. The requirements in our definition
of development are weaker than the requirements in \cite{aarendall,fal}. For that reason, it is of interest to analyse if one can, starting with our
requirements, derive the improvements corresponding to \cite{aarendall,fal}. This is also related to the question of necessity discussed above. This
naturally leads us to the topic of necessity, which is addressed in Subsection~\ref{ssection:contodataonsing}. Finally, in
Subsection~\ref{ssection:concloutl}, we provide conclusions and an outlook.

\subsection{Quiescent solutions to Einstein's vacuum equations without symmetries}\label{ssection:fal}
The notion of initial data on a singularity we propose here arises naturally when considering the recent results of Fournodavlos and Luk,
\cite{fal}, with the geometric perspective developed in \cite{RinGeo} in mind. For this reason, it is natural to begin by a discussion of the
results of \cite{fal}. In this article, the authors specify the following initial data on the singularity
\begin{definition}\label{def:falid}
  \textit{Fournodavlos and Luk initial data on the singularity} consist of smooth functions $c_{ij},p_{i}:\tn{3}\rightarrow\rn{}$, $i,j=1,2,3$
  satisfying the following conditions: $c_{ij}=c_{ji}$; $p_{1}<p_{2}<p_{3}$; $c_{11},c_{22},c_{33}>0$;
  \begin{equation}\label{eq:Kasnerrelationsfal}
    \textstyle{\sum}_{i=1}^{3}p_{i}=\textstyle{\sum}_{i=1}^{3}p_{i}^{2}=1;
  \end{equation}
  and 
  \begin{equation}\label{eq:falciicond}
    \sum_{l=1}^{3}\left[\frac{\d_{i}c_{ll}}{c_{ll}}(p_{l}-p_{i})+2\d_{l}\kappa_{i}^{\phantom{i}l}
      +\mathbbm{1}_{\{l>i\}}\frac{\d_{l}(c_{11}c_{22}c_{33})}{c_{11}c_{22}c_{33}}\kappa_{i}^{\phantom{i}l}\right]=0
  \end{equation}
  for $i=1,2,3$. Here $\kappa_{i}^{\phantom{i}i}=-p_{i}$ (no summation); $\kappa_{i}^{\phantom{i}l}=0$ if $l<i$;
  \begin{align}
    \kappa_{1}^{\phantom{1}2} := & (p_{1}-p_{2})\frac{c_{12}}{c_{22}},\ \ \
    \kappa_{2}^{\phantom{2}3} := (p_{2}-p_{3})\frac{c_{23}}{c_{33}},\ \ \
    \kappa_{1}^{\phantom{1}3} := (p_{2}-p_{1})\frac{c_{12}c_{23}}{c_{22}c_{33}}+(p_{1}-p_{3})\frac{c_{13}}{c_{33}};\label{eq:kappacijrelfal}
  \end{align}
  $\mathbbm{1}_{\{l>i\}}=1$ if $l>i$; and $\mathbbm{1}_{\{l>i\}}=0$ if $l\leq i$.
\end{definition}

Given data satisfying these conditions, there is, due to \cite[Theorem~1.1, p.~1185--1186]{fal}, a $C^{2}$ solution to the Einstein vacuum equations
$\mathrm{Ric}=0$ of the form
\begin{equation}\label{eq:gfalasformofmetric}
  g=-dt\otimes dt+\textstyle{\sum}_{i,j=1}^{3}a_{ij}t^{2p_{\max\{i,j\}}}dx^{i}\otimes dx^{j},
\end{equation}
where $(t,x^{1},x^{2},x^{3})\in (0,T]\times\tn{3}$ for some $T>0$ and $a_{ij}:(0,T]\times\tn{3}\rightarrow\rn{}$ are smooth functions (symmetric in
$i$ and $j$) which extend to continuous functions $:[0,T]\times\tn{3}\rightarrow\rn{}$. Moreover, the $a_{ij}$ obey
\begin{equation}\label{eq:aijlimcijfal}
  \lim_{t\rightarrow 0+}a_{ij}(t,x)=c_{ij}(x).
\end{equation}
However, in the course of the argument, the authors also obtain the following control of the Weingarten map:
\begin{equation}\label{eq:kconvfal}
  \textstyle{\sum}_{r=0}^{1}\sum_{|\alpha|\leq 2-r}t^{r+1}|\d_{t}^{r}\d^{\alpha}_{x}(k^{\phantom{j}i}_{j}-t^{-1}\kappa^{\phantom{j}i}_{j})|
  =O(\min\{t^{\varepsilon},t^{\varepsilon-2p_{i}+2p_{j}}\}),
\end{equation}
where $\varepsilon>0$ and $k$ is the second fundamental form of the leaves of the foliation (using the conventions of \cite{fal}; here we use the
opposite conventions). Clearly, the formulation of this result is tied quite closely to the specific situation under consideration: $\tn{3}$ spatial
topology; Gaussian foliation etc. On the other hand, the results of \cite{RinWave,RinGeo} indicate that quiescent behaviour in the
case of $3+1$-dimensional vacuum solutions to Einstein's equations should be tied to properties of the expansion normalised Weingarten map.
It is therefore of interest to try to combine the formulation in \cite{fal} with the more geometric perspective taken in \cite{RinGeo}. 

\subsection{A geometric formulation of initial data on the singularity}\label{ssection:geometricformfv}
In order to relate the results of \cite{fal} with those of \cite{RinWave,RinGeo},
note that our conventions concerning the second fundamental form are opposite to those of \cite{fal}. Here we use $\bk$ to denote
the second fundamental form of the leaves of the foliation, and $\bk=-k$, where $k$ is object appearing in (\ref{eq:kconvfal}). Moreover,
we denote the Weingarten map by $\bK$; i.e., $\bK$ is the $(1,1)$-tensor field satisfying $\bK^{\phantom{j}i}_{j}=\bk^{\phantom{j}i}_{j}$. If $\theta$ denotes
the mean curvature of the leaves of the foliation, the \textit{expansion normalised Weingarten map} is defined by $\mK:=\bK/\theta$, assuming
$\theta\neq 0$. Note here that taking the trace of (\ref{eq:kconvfal}), keeping the definition of $\kappa$ and the conventions in mind, yields
$t\theta-1=O(t^{\varepsilon})$ in $C^{2}$. Combining this observation with (\ref{eq:kconvfal}) again yields the conclusion that $\mK$ converges to a
limit, say $\msK$, in $C^{2}$, and the rate of convergence is $t^{\varepsilon}$. In fact, $\msK^{\phantom{j}i}_{j}=-\kappa^{\phantom{j}i}_{j}$. At this stage,
it is of interest to reformulate the conditions of Definition~\ref{def:falid} in terms of $\msK$. Note, to this end, that (\ref{eq:Kasnerrelationsfal})
can be reformulated to $\tr\msK=\tr\msK^{2}=1$. The fact that the $p_{i}$ are distinct can be reformulated as saying that the eigenvalues of $\msK$ are
distinct. Next, note that (\ref{eq:gfalasformofmetric}) and (\ref{eq:aijlimcijfal}) are closely tied to the coordinate vector fields on $\tn{3}$. From
a geometric perspective, it is more natural to express these equalities using an eigenframe associated with $\msK$. However, such a frame
$\{\bmsX_{A}\}$ with associated co-frame $\{\bmsY^{A}\}$ is given by
\begin{subequations}\label{eq:bmsXbmsYdefintro}
  \begin{align}%\label{eq:preferredcoordinatesintro}
    \bmsX_{3} = & \d_{3},\ \ \
    \bmsX_{2}=\d_{2}+\bmsX_{2}^{3}\d_{3},\ \ \
    \bmsX_{1}=\d_{1}+\bmsX_{1}^{2}\d_{2}+\bmsX_{1}^{3}\d_{3},\label{eq:bmsXdefintro}\\
    \bmsY^{1} = & dx^{1},\ \ \ \bmsY^{2}=dx^{2}-\bmsX_{1}^{2}dx^{1},\ \ \
    \bmsY^{3}=dx^{3}-\bmsX_{2}^{3}dx^{2}+(\bmsX_{1}^{2}\bmsX_{2}^{3}-\bmsX_{1}^{3})dx^{1},\label{eq:bmsYdefintro}
  \end{align}
\end{subequations}
where
\begin{equation}\label{eq:bmsXABdef}
  \bmsX_{1}^{2} = \frac{1}{p_{1}-p_{2}}\msK_{1}^{2},\ \ \
  \bmsX_{2}^{3} = \frac{1}{p_{2}-p_{3}}\msK_{2}^{3},\ \ \
  \bmsX_{1}^{3} = \frac{1}{p_{1}-p_{3}}\left(\msK_{1}^{3}+\frac{1}{p_{1}-p_{2}}\msK_{1}^{2}\msK_{2}^{3}\right).
\end{equation}
With this notation, $\msK\bmsX_{A}=p_{A}\bmsX_{A}$ (no summation). Moreover, (\ref{eq:gfalasformofmetric}) can be written
\begin{equation}\label{eq:ggeometricfal}
  g=-dt\otimes dt+\textstyle{\sum}_{A,B=1}^{3}b_{AB}t^{2p_{\max\{A,B\}}}\bmsY^{A}\otimes\bmsY^{B}.
\end{equation}
Combining the above terminology with (\ref{eq:kappacijrelfal}) and (\ref{eq:aijlimcijfal}), it can also be verified that
\[
\lim_{t\rightarrow 0+}b_{AB}=c_{AA}\delta_{AB}
\]
(no summation on $A$), so that
\begin{equation}\label{eq:bhdefintro}
  \lim_{t\rightarrow 0+}\textstyle{\sum}_{A,B=1}^{3}b_{AB}\bmsY^{A}\otimes\bmsY^{B}
  =\textstyle{\sum}_{A=1}^{3}c_{AA}\bmsY^{A}\otimes\bmsY^{A}=:\bh.
\end{equation}
Since $c_{AA}>0$, the limit, i.e. $\bh$, is a Riemannian metric on $\tn{3}$. At this stage, we thus have two natural structures on
$\tn{3}$: the $(1,1)$-tensor field $\msK$ and the Riemannian metric $\bh$. It would of course be desirable if all the conditions
could be expressed in terms of these quantities. From the above observations it is clear that the $p_{i}$ can be read off as the
eigenvalues of $\msK$; that (\ref{eq:Kasnerrelationsfal}) can be reformulated to $\tr\msK=\tr\msK^{2}=1$; and that $\msK$ is symmetric
with respect to $\bh$, cf. (\ref{eq:bhdefintro}). Next, on the basis of \cite{RinGeo}, the structure coefficients $\g_{BC}^{A}$ associated
with the frame $\{\bmsX_{A}\}$ can be expected to be important. They are defined by
\begin{equation}\label{eq:gABCdef}
  [\bmsX_{B},\bmsX_{C}]=\g^{A}_{BC}\bmsX_{A}.
\end{equation}
Due to \cite{RinGeo}, we would expect $\g^{1}_{23}=0$ to be a necessary condition in order to obtain the quiescent behaviour of \cite{fal} (note
that this condition is independent of the normalisation of the eigenvector fields and that it is crucially dependent on the assumption
that $p_{1}<p_{2}<p_{3}$). However, the frame introduced in (\ref{eq:bmsXbmsYdefintro}) satisfies this condition. In fact, this condition
is automatically built into the definition of the initial data. Finally, we need to express (\ref{eq:falciicond}) geometrically. Since this equation
is the limit of the momentum constraint, cf. \cite[Remark~1.4, p.~1186]{fal}, it is of interest to relate it to $\rodiv_{\bh}\msK$. In fact, it turns out
that (\ref{eq:falciicond}) is equivalent to the condition that $\rodiv_{\bh}\msK=0$; cf. Lemma~\ref{lemma:momcondivmsKz} below. The above observations
suggest the following definition of initial data on a big bang singularity for Einstein's vacuum equations in $3+1$-dimensions.

\begin{definition}\label{def:vacuumidos}
  Let $(\bM,\bh)$ be a smooth $3$-dimensional Riemannian manifold and $\msK$ be a smooth $(1,1)$-tensor field on $\bM$. Then $(\bM,\bh,\msK)$
  are \textit{non-degenerate quiescent vacuum initial data on the singularity} if the following holds:
  \begin{enumerate}
  \item $\tr\msK=1$ and $\msK$ is symmetric with respect to $\bh$.
  \item $\tr\msK^{2}=1$ and $\mathrm{div}_{\bh}\msK=0$.
  \item The eigenvalues of $\msK$ are distinct.
  \item $\msO^{1}_{23}\equiv 0$ on $\bM$, where (if $p_{1}<p_{2}<p_{3}$ are the eigenvalues of $\msK$, $\msX_{A}$ is an eigenvector field corresponding
    to $p_{A}$, normalised so that $|\msX_{A}|_{\bh}=1$, and $\g_{BC}^{A}$ is defined by $[\msX_{B},\msX_{C}]=\g^{A}_{BC}\msX_{A}$) the quantity $\msO^{A}_{BC}$
    is defined by $\msO^{A}_{BC}:=(\g^{A}_{BC})^{2}$.
  \end{enumerate}
\end{definition}
\begin{remark}
  That $\msK$ is symmetric with respect to $\bh$ means that if $p\in\bM$ and $\xi,\zeta\in T_{p}\bM$, then $\bh(\msK\xi,\zeta)=\bh(\xi,\msK\zeta)$.
\end{remark}
\begin{remark}
  Due to the normalisation of the eigenvector fields, they are uniquely defined up to a sign. For this reason, the functions $\msO^{A}_{BC}$ are
  uniquely defined smooth functions on $\bM$.
\end{remark}
Due to the observations made prior to the statement of the definition, the following holds.
\begin{prop}\label{prop:falaregeoidos}
  Given Fournodavlos and Luk initial data on the singularity in the sense of Definition~\ref{def:falid}, define $\msK$ by
  $\msK^{\phantom{j}i}_{j}=-\kappa^{\phantom{j}i}_{j}$ and $\bh$ by (\ref{eq:bhdefintro}). Then $(\tn{3},\bh,\msK)$ are
  non-degenerate quiescent vacuum initial data on the singularity. 
\end{prop}
In short, initial data in the sense of Definition~\ref{def:falid} are a special case of initial data in the sense of
Definition~\ref{def:vacuumidos}. It is therefore of interest to find out if it is possible to go in the other direction. On the level of
local coordinates, it turns out to be possible. 
\begin{prop}\label{prop:geometrictofal}
  Let $(\bM,\bh,\msK)$ be non-degenerate quiescent vacuum initial data on the singularity in the sense of Definition~\ref{def:vacuumidos}
  and let $p\in\bM$. Let $p_{A}$, $A=1,2,3$, be the eigenvalues of $\msK$, and assume them to be ordered so that $p_{1}<p_{2}<p_{3}$. Then there
  are local coordinates $(V,\bsfx)$ such that $p\in V$ and such that $\msK$ has a frame $\{\bmsX_{A}\}$ of eigenvector fields on $V$ given by 
  (\ref{eq:bmsXdefintro}) and (\ref{eq:bmsXABdef}), where the sub- and superscripts refer to the coordinates $(V,\bsfx)$. Define
  $c_{AA}:=\bh(\bmsX_{A},\bmsX_{A})$ (no summation). Define
  \begin{equation}\label{eq:cotetcdef}
    c_{12}:=\frac{c_{22}}{p_{2}-p_{1}}\msK_{1}^{\phantom{1}2},\ \
    c_{23}:=\frac{c_{33}}{p_{3}-p_{2}}\msK_{2}^{\phantom{2}3},\ \
    c_{13}:=\frac{c_{33}}{p_{3}-p_{1}}\left(\msK_{1}^{\phantom{1}3}+\frac{1}{p_{3}-p_{2}}\msK_{1}^{\phantom{1}2}\msK_{2}^{\phantom{2}3}\right),
  \end{equation}
  where the components of $\msK$ are calculated with respect to the local coordinates $(V,\bsfx)$, and define the remaining components $c_{ij}$
  by requiring $c_{ij}=c_{ji}$. Then the functions $p_{i}$ and $c_{ij}$ satisfy
  all the conditions of Definition~\ref{def:falid}, but on the set $V$ instead of $\tn{3}$, and where the sub-, superscripts and derivatives refer
  to the coordinate frame associated with $(V,\bsfx)$.
\end{prop}
\begin{proof}
  The proof of the statement is to be found at the end of Subsection~\ref{ssection:specasdata}. 
\end{proof}
This observation indicates that Definitions~\ref{def:falid} and \ref{def:vacuumidos} are locally the same. On the other hand, it is conceivable
that a more general definition of data on the singularity is possible and would, non-trivially, include both Definitions~\ref{def:falid} and
\ref{def:vacuumidos}. We return to this topic in Subsection~\ref{ssection:contodataonsing} below.

One potential deficiency in the definitions
is the non-degeneracy condition. In some situations, it is to be expected that this condition can be relaxed; cf. Remark~\ref{remark:thedegeneratecase}
below. On the other hand, one of the main applications of the definitions is to think of the data on the singularity as parametrising
the stable manifold associated with the Kasner circle. One of the main potential uses of this perspective is to the understanding of
oscillatory spatially inhomogeneous spacetimes; cf. Subsection~\ref{ssection:oscibbsing} below. In that context, the degenerate case is not of
interest; cf. Subsection~\ref{ssection:oscibbsing} below.

\subsection{Developments of non-degenerate quiescent vacuum initial data on the singularity}\label{ssection:geo developments}
Once one has specified initial data on the singularity,
the question arises how these data are related to a corresponding development. In the case of \cite{fal}, we have already given the relation; cf.
(\ref{eq:gfalasformofmetric})--(\ref{eq:kconvfal}). Moreover, (\ref{eq:ggeometricfal}), (\ref{eq:bhdefintro}) and the fact that $\mK\rightarrow\msK$
suggest how to relate the development with the data more generally. In the setting of Definition~\ref{def:vacuumidos}, we use the following terminology.

\begin{definition}\label{def:developmentvacuum}
  Let $(\bM,\bh,\msK)$ be non-degenerate quiescent vacuum initial data on the singularity in the sense of Definition~\ref{def:vacuumidos}. A
  \textit{locally Gaussian development of the initial data} is then a time oriented Lorentz manifold
  $(M,g)$, solving Einstein's vacuum equations, such that the following holds. There is a $0<t_{+}\in\rn{}$ and a diffeomorphism $\Psi$ from
  $\bM\times (0,t_{+})$ to an open subset of $(M,g)$ such that 
  \begin{equation}\label{eq:rinansatzintroPsi}
    \Psi^{*}g=-dt\otimes dt+\textstyle{\sum}_{A,B}b_{AB}t^{2p_{\max\{A,B\}}}\msY^{A}\otimes \msY^{B},
  \end{equation}
  where $\{\msX_{A}\}$ is a local basis of eigenvector fields of $\msK$, $\{\msY^{A}\}$ is the dual basis, $\msK\msX_{A}=p_{A}\msX_{A}$ (no summation)
  and $p_{1}<p_{2}<p_{3}$. Here $\Psi(\bM_{t})$, where $\bM_{t}:=\bM\times \{t\}$, is required to have strictly positive mean curvature in $(M,g)$ for
  $t\in (0,t_{+})$. Let $\theta$ be the mean curvature of the leaves of the foliation; $\bK$ be the Weingarten map of the leaves of the foliation;
  $\mK:=\bK/\theta$ be the expansion normalised Weingarten map; and $\chh$ be defined by
  \begin{equation}\label{eq:chhdef}
    \chh:=\textstyle{\sum}_{A,B}b_{AB}\msY^{A}\otimes \msY^{B};
  \end{equation}
  note that $\chh$ is globally well defined independent of the choice of local frame. Then the following correspondence between the solution
  and the asymptotic data is required to hold. There is an $\vare>0$ and for every $0\leq l\in\nn{}$, a constant $C_{l}$ such that
  \begin{subequations}\label{eq:mKchhconvrate}
    \begin{align}
      \|\mK(\cdot,t)-\msK\|_{C^{l}(\bM)} \leq & C_{l}t^{\vare},\label{eq:mKlim}\\
      \|\chh(\cdot,t)-\bh\|_{C^{l}(\bM)} \leq & C_{l}t^{\vare}\label{eq:limitbABtobhchh}
    \end{align}
  \end{subequations}
  on $M_{+}:=\bM\times (0,t_{+})$. 
  If, in addition to the above, $(M,g)$ is globally hyperbolic and the hypersurfaces $\Psi(\bM_{t})$ are Cauchy hypersurfaces in $(M,g)$ for
  $t\in (0,t_{+})$, then $(M,g)$ is called a \textit{locally Gaussian globally hyperbolic development of the initial data}.
\end{definition}
\begin{remark}
  The convergence conditions (\ref{eq:mKchhconvrate}) could both be weakened and strengthened; one could change the rate, change the norm, only
  insist on control of a finite number of derivatives etc. However, we here stick to this particular choice and encourage the reader to consider
  other possibilities. 
\end{remark}
\begin{remark}
  Given initial data as in Definition~\ref{def:vacuumidos}, we locally obtain initial data in the sense of Definition~\ref{def:falid}; cf.
  Proposition~\ref{prop:geometrictofal}. Assuming it
  is possible to appeal to a localised version of \cite{fal}, it is therefore of interest to know if a development in the sense of \cite{fal} yields
  a development in the sense of Definition~\ref{def:developmentvacuum} (at least locally). That the answer to this question is yes is demonstrated
  in Proposition~\ref{prop:dergeoas} below. 
\end{remark}

Comparing this definition with \cite{fal}, it is of interest to note that the authors of \cite{fal} require control of the normal derivative of the
second fundamental form in order to conclude uniqueness; cf. \cite[Theorem~1.7, p.~1187--1188]{fal}, in particular \cite[(1.10), p.~1187]{fal} (this estimate
coindincides with (\ref{eq:kconvfal}) above). They also require control of the asymptotic behaviour of $\theta$. It is not immediately obvious how
to estimate $\theta$ by appealing to (\ref{eq:mKchhconvrate}). However, it turns out to be possible to use these estimates
to deduce that $t\theta-1$ converges to zero as $t^{\vare}$ in $C^{l}$ etc. We discuss this topic in greater detail in
Subsection~\ref{ssection:imprasympt} below. 

\subsection{Generalisations}\label{ssection:generalisations}

So far, we have focused on the $3+1$-dimensional vacuum setting. However, it is of interest to generalise Definitions~\ref{def:vacuumidos} and
\ref{def:developmentvacuum}. For this reason, we now define data on the singularity in the case of higher dimensions, for scalar fields, and
in the presence of a cosmological constant. 

\begin{definition}\label{def:ndsfidonbbs}
  Let $3\leq n\in\nn{}$, $\Lambda\in\rn{}$, $(\bM,\bh)$ be a smooth $n$-dimensional Riemannian manifold, $\msK$ be a smooth $(1,1)$-tensor field on
  $\bM$ and $\Phia$ and $\Phib$ be smooth functions on $\bM$. Then $(\bM,\bh,\msK,\Phia,\Phib)$ are
  \textit{non-degenerate quiescent initial data on the singularity for the Einstein-scalar field equations with a cosmological constant $\Lambda$} if
  \begin{enumerate}
  \item $\tr\msK=1$ and $\msK$ is symmetric with respect to $\bh$.
  \item $\tr\msK^{2}+\Phia^{2}=1$ and $\mathrm{div}_{\bh}\msK=\Phia d\Phib$.
  \item The eigenvalues of $\msK$ are distinct.
  \item $\msO^{A}_{BC}$ vanishes in a neighbourhood of $\bx\in\bM$ if $1+p_{A}(\bx)-p_{B}(\bx)-p_{C}(\bx)\leq 0$.    
  \end{enumerate}
\end{definition}
\begin{remark}\label{remark:msOdefsf}
  Here the $\msO^{A}_{BC}$ are defined as follows. Let $p_{1}<\cdots<p_{n}$ be the eigenvalues of $\msK$, let $\msX_{A}$ be an eigenvector field
  corresponding to $p_{A}$ such that $|\msX_{A}|_{\bh}=1$, and define $\g^{A}_{BC}$ by $[\msX_{B},\msX_{C}]=\g^{A}_{BC}\msX_{A}$. Then
  $\msO^{A}_{BC}:=(\g^{A}_{BC})^{2}$. 
\end{remark}
\begin{remark}\label{remark:commcond}
  Since $\tr\msK^{2}+\Phia^{2}=1$, the sum of the $p_{A}^{2}$ is less than or equal to $1$. In case $n\geq 3$, it is thus clear that $|p_{A}|<1$ for
  all $A\in\{1,\dots,d\}$. In particular, $1+p_{A}-p_{B}-p_{C}>0$ if $A=B$ or if $A=C$. In the last criterion of the definition, we can therefore assume
  that $A\notin\{B,C\}$. Note also that $\msO^{A}_{BC}=0$ if $B=C$. To conclude, $A$, $B$ and $C$ can be assumed to be distinct in the last criterion of
  the definition. In case $n=3$ and $A$, $B$ and $C$ are distinct, $1+p_{A}-p_{B}-p_{C}=2p_{A}$. 
\end{remark}
\begin{remark}
  Condition 4 should be compared with the requirement $1+p_{1}-p_{n-1}-p_{n}>0$ appearing in \cite[Theorem~10.1, p.~1104]{damouretal} (assuming
  $p_{1}<\cdots<p_{n}$); cf. also \cite{Henneauxetal} for the origin of this condition. 
\end{remark}
\begin{remark}\label{remark:dataonsingneqthree}
  In case $n=3$, the fourth condition is only a restriction if $p_{1}(\bx)\leq 0$. Moreover, it then corresponds to the requirement that
  $\msO^{1}_{23}=0$ in a neighbourhood of $\bx$. The reason for this is that if $p_{1}>0$, then $p_{i}>0$, $i=1,2,3$, and the fourth condition is void.
  If $p_{1}\leq 0$, then $p_{2}>0$, since the sum of the $p_{i}$ equal $1$ and $p_{3}<1$. Thus the $A$ appearing in the fourth condition has to be $1$
  and $\{B,C\}$ has to equal $\{2,3\}$. This observation should be compared with condition~4 of Definition~\ref{def:vacuumidos}. 
\end{remark}
\begin{remark}
  For $A$, $B$ and $C$ distinct, the condition that $\msO^{A}_{BC}$ vanish is independent of the choice of normalisation of the eigenvector fields
  $\msX_{A}$. 
\end{remark}

Next, we generalise Definition~\ref{def:developmentvacuum}. 

\begin{definition}\label{def:developmentsfLambda}
  Let $(\bM,\bh,\msK,\Phia,\Phib)$ be non-degenerate quiescent initial data on the singularity for the Einstein-scalar field equations with a
  cosmological constant $\Lambda$, in the sense of Definition~\ref{def:ndsfidonbbs}. A \textit{locally Gaussian development of the initial data} is
  then a smooth time oriented Lorentz manifold $(M,g)$ and a smooth function $\phi$, solving the Einstein scalar field equations with a cosmological
  constant $\Lambda$, such that the
  following holds. There is a $0<t_{+}\in\rn{}$ and a diffeomorphism $\Psi$ from $\bM\times (0,t_{+})$ to an open subset of $(M,g)$ such that
  $\Psi^{*}g$ can be written as in (\ref{eq:rinansatzintroPsi}), where $\{\msX_{A}\}$ and $p_{A}$ are given in Remark~\ref{remark:msOdefsf}
  and $\{\msY^{A}\}$ is the dual basis to $\{\msX_{A}\}$. Moreover, there is a $t_{0}\in (0,t_{+})$ and a $0<\theta_{0}\in\rn{}$ such that $\Psi(\bM_{t})$,
  where $\bM_{t}:=\bM\times \{t\}$, has mean curvature $\theta$ bounded from below by $\theta_{0}$ for $t\leq t_{0}$. If $\Lambda>0$,
  $\theta_{0}$ is additionally required to satisfy $\theta_{0}>[2\Lambda/(n-1)]^{1/2}$. Let $\mK$ be the expansion normalised Weingarten map of the
  leaves of the foliation and $\chh$ be defined by (\ref{eq:chhdef}). Then the following correspondence between the solution and the asymptotic data
  is required to hold. There is an $\vare>0$ and for every $l\in\nn{}$, a constant $C_{l}$ such that
  \begin{subequations}\label{eq:mKchhconvratesf}
    \begin{align}
      \|\mK(\cdot,t)-\msK\|_{C^{l}(\bM)} \leq & C_{l}t^{\vare},\label{eq:mKlimsf}\\
      \|\chh(\cdot,t)-\bh\|_{C^{l}(\bM)} \leq & C_{l}t^{\vare},\label{eq:limitbABtobhchhsf}\\
      \|(\hU \phi)(\cdot,t)-\Phia\|_{C^{l}(\bM)} \leq & C_{l}t^{\vare},\label{eq:Phialim}\\
      \|(\phi-\Phia\varrho)(\cdot,t)-\Phib\|_{C^{l+1}(\bM)} \leq & C_{l}t^{\vare}\label{eq:Phiblim}
    \end{align}
  \end{subequations}
  for all $t\leq t_{0}$. Here $\hU:=\theta^{-1}\d_{t}$. Moreover, $\varrho$ denotes the logarithmic volume density, defined by the condition
  $e^{\varrho}\mu_{\bh}=\mu_{\bg}$, where
  $\mu_{\bh}$ is the volume form on $\bM$ induced by $\bh$, $\mu_{\bg}$ is the volume form on $\bM$ induced by the metric $\Psi^{*}\bg$ on the
  leaves $\bM_{t}$. If, in addition to the above, $(M,g)$ is globally hyperbolic and the hypersurfaces $\Psi(\bM_{t})$ are Cauchy hypersurfaces
  in $(M,g)$ for $t\in (0,t_{+})$, then $(M,g)$ is called a \textit{locally Gaussian globally hyperbolic development of the initial data}.
\end{definition}

It is sometimes of interest to consider foliations that are not Gaussian. Note, to this end, that given a development as in
Definition~\ref{def:developmentsfLambda}, the metric can, in a neighbourhood of the singularity, be given a representation expressed in terms of
the mean curvature instead of $t$; cf. Theorem~\ref{thm:improvingasymptoticsgeneralG}, in particular (\ref{eq:lnthetalnttthetamoneG}), below.
However, it might also be of interest to allow a lapse function different from $1$. We therefore give a generalised definition of development,
though we still require the shift vector field to vanish. 

\begin{definition}\label{def:developmentsfLambdacrushing}
  Let $(\bM,\bh,\msK,\Phia,\Phib)$ be non-degenerate quiescent initial data on the singularity for the Einstein-scalar field equations with a
  cosmological constant $\Lambda$, in the sense of Definition~\ref{def:ndsfidonbbs}. A \textit{local crushing development of the initial data with
  vanishing shift vector field} is then a smooth time oriented Lorentz manifold $(M,g)$ and a smooth function $\phi$, solving the Einstein scalar
  field equations with a cosmological constant $\Lambda$, such that the following holds. There is a $0<t_{+}\in\rn{}$ and a diffeomorphism $\Psi$
  from $\bM\times (0,t_{+})$ to an open subset of $(M,g)$ such that 
  \begin{equation}\label{eq:rinansatzintrocrushing}
    \Psi^{*}g=-N^{2}dt\otimes dt+\textstyle{\sum}_{A,B}b_{AB}\theta^{-2p_{\max\{A,B\}}}\msY^{A}\otimes \msY^{B},
  \end{equation}
  where $N$ is a strictly positive function (the \textit{lapse function}), $\{\msX_{A}\}$ and $p_{A}$ are given in Remark~\ref{remark:msOdefsf}
  and $\{\msY^{A}\}$ is the dual basis to $\{\msX_{A}\}$. Here $\Psi(\bM_{t})$, where $\bM_{t}:=\bM\times \{t\}$, is required to have strictly positive
  mean curvature $\theta$ in $(M,g)$ for $t\in (0,t_{+})$. Moreover, $\theta$ is required to diverge uniformly to $\infty$ as $t\rightarrow 0+$. 
  Let $\mK$ be the expansion normalised Weingarten map of the leaves of the foliation and $\chh$ be defined by (\ref{eq:chhdef}). Then the following
  correspondence between the solution and the asymptotic data is required to hold. Then the following correspondence between the solution and the
  asymptotic data is required to hold. There is an $\vare>0$ and for every $l\in\nn{}$, a constant $C_{l}$ such that
  \begin{subequations}\label{eq:mKchhconvratesfN}
    \begin{align}
      |\bD^{l}(\mK-\msK)|_{\bh} \leq & C_{l}\theta^{-\vare},\label{eq:mKlimsfN}\\
      |\bD^{l}(\chh-\bh)|_{\bh} \leq & C_{l}\theta^{-\vare},\label{eq:limitbABtobhchhsfN}\\
      |\bD^{l}(\hU \phi-\Phia)|_{\bh} \leq & C_{l}\theta^{-\vare},\label{eq:PhialimN}\\
      |\bD^{l}(\phi-\Phia\varrho-\Phib)|_{\bh} \leq & C_{l}\theta^{-\vare}\label{eq:PhiblimN}
    \end{align}
  \end{subequations}
  on $\bM\times (0,t_{0}]$ for some $0<t_{0}<t_{+}$. Here $\hU:=\theta^{-1}\d_{t}$. Moreover, $\bD$ denotes the Levi-Civita
  connection associated with $\bh$ and
  $\varrho$ denotes the logarithmic volume density, defined by the condition $e^{\varrho}\mu_{\bh}=\mu_{\bg}$, where $\mu_{\bh}$ is the volume form on
  $\bM$ induced by $\bh$, $\mu_{\bg}$ is the volume form on $\bM$ induced by the metric $\Psi^{*}\bg$ on the leaves $\bM_{t}$. 
  If, in addition to the above, $(M,g)$ is globally hyperbolic and the hypersurfaces $\Psi(\bM_{t})$ are Cauchy hypersurfaces in $(M,g)$ for
  $t\in (0,t_{+})$, then $(M,g)$ is called a \textit{locally crushing globally hyperbolic development with vanishing shift vector field of the
  initial data}.
\end{definition}
\begin{remark}
  If $S$ is a covariant $k$-tensor field on $\bM$, e.g., then the notation $|S|_{\bh}$ means
  \[
  |S|_{\bh}:=(\bh^{i_1j_1}\cdots\bh^{i_kj_k}S_{i_1\cdots i_k}S_{j_1\cdots j_k})^{1/2}.
  \]
\end{remark}

\subsection{The stable quiescent regime}\label{ssection:stablequiescentregime}

There are results guaranteeing the existence of developments corresponding to non-degenerate quiescent initial data on the singularity for the
Einstein-scalar field equations. In the vacuum setting, one example of this is given by \cite[Theorem~1.1, p.~1185--1186]{fal}, mentioned in connection with
Definition~\ref{def:falid}. Moreover, in case condition 4 of Definition~\ref{def:ndsfidonbbs} is void, the existence of developments is guaranteed,
in the real analytic setting, by \cite{aarendall,damouretal}. In order to see this, it is of interest to relate Definition~\ref{def:ndsfidonbbs}
with the perspective of \cite{aarendall,damouretal}. The main theorems of \cite{aarendall,damouretal} (i.e.,
\cite[Theorems~1 and 2, pp.~484--485]{aarendall} and \cite[Theorems~10.1 and 10.2, p.~1104]{damouretal}) schematically state the following: given real
analytic solutions to a so-called velocity dominated version of the equations (or to the Kasner-like Einstein-matter equations), there is a unique real
analytic solution to the actual equations whose asymptotics are determined by the velocity
dominated solution. In other words, in \cite{aarendall,damouretal}, the velocity dominated solutions should be thought of as initial data on the big bang
singularity. In the case of the Einstein-scalar field equations, solutions to the the velocity dominated system are given by
$({}^{0}g,{}^{0}k,{}^{0}\phi)$, where ${}^{0}g$ is a family of real analytic Riemannian metrics on a manifold $\bM$, ${}^{0}k$ is a family of real
analytic symmetric covariant $2$-tensor fields on $\bM$ and ${}^{0}\phi$ is a family of real analytic functions on $\bM$. The families are parametrised
by $t\in (0,\infty)$ and depend real analytically on $t$. Needless to say, one could discuss velocity dominated solutions in other regularity classes.
The equations that should be satisfied are listed in \cite[Subsection~2.2, pp.~482--483]{aarendall}; see \cite[(10)--(13), p.~438]{aarendall}.
One advantage of the perspective taken in \cite{aarendall,damouretal} is that the condition of non-degeneracy of $\msK$, cf.
Definition~\ref{def:ndsfidonbbs}, does not appear; cf. also Remark~\ref{remark:thedegeneratecase} below for a further discussion of the issue of
non-degeneracy. The following two propositions relate the two perspectives on initial data on the singularity.

\begin{prop}\label{prop:idtovds}
  Let $(\bM,\bh,\msK,\Phia,\Phib)$ be non-degenerate quiescent initial data on the singularity for the Einstein-scalar field equations with a
  vanishing cosmological constant. Denote the eigenvalues of $\msK$ by $p_{A}$, $A=1,\dots,n$, and order them so that $p_{1}<\cdots<p_{n}$. Let $\msX_{A}$,
  $A=1,\dots,n$, be an eigenvector field of $\msK$ corresponding to $p_{A}$, normalised so that $|\msX_{A}|_{\bh}=1$, and let $\{\msY^{A}\}$ be the basis
  dual to $\{\msX_{A}\}$. Then
  \begin{subequations}\label{eq:AVTDsolution}
    \begin{align}
      {}^{0}g := & \textstyle{\sum}_{A}t^{2p_{A}}\msY^{A}\otimes\msY^{A},\label{eq:zgdef}\\
      {}^{0}k := & -\textstyle{\sum}_{A}p_{A}t^{2p_{A}-1}\msY^{A}\otimes\msY^{A},\label{eq:zkdef}\\
      {}^{0}\phi := & \Phia\ln t+\Phib\label{eq:zphidef}
    \end{align}
  \end{subequations}
  is a velocity dominated solution in the sense of \cite[Subsection~2.2, pp.~482--483]{aarendall}. 
\end{prop}
\begin{remark}\label{remark:eightpi}
  The conventions of \cite{aarendall} differ from ours in that a factor of $8\pi$ is included in Einstein's equations in \cite{aarendall}. This means
  that all the terms involving the scalar field appearing on the right hand sides of \cite[(10a), (10b) and (11b), p.~483]{aarendall} should be
  divided by $8\pi$ in order to make the formulae consistent with our conventions.
\end{remark}
\begin{remark}
  Note that while the frame $\{\msX_{A}\}$ and co-frame $\{\msY^{A}\}$ may only be local, the $\msY^{A}$ are well defined up to a sign, so that
  the right hand sides of (\ref{eq:zgdef}) and (\ref{eq:zkdef}) are globally well defined. 
\end{remark}
\begin{proof}
  The proof is to be found in Section~\ref{section:existenceofdev} below.
\end{proof}
In the following proposition, we go in the opposite direction.

\begin{prop}\label{prop:vdstoid}
  Let $({}^{0}g,{}^{0}k,{}^{0}\phi)$ be a velocity dominated solution in the sense of \cite[Subsection~2.2, pp.~482--483]{aarendall}, where the underlying
  manifold is $\bM$ with $\mathrm{dim}\bM\geq 3$. Adjusting the time interval appropriately, it can then be assumed that $\tr{}^{0}k=-t^{-1}$. Moreover,
  if ${}^{0}K$ is the $(1,1)$-tensor field with components ${}^{0}K^{a}_{\phantom{a}b}={}^{0}k^{a}_{\phantom{a}b}$, where the indices are raised with ${}^{0}g$,
  then $-t\cdot{}^{0}K$ is independent of $t$. Define $\msK:=-t\cdot{}^{0}K$. Then the eigenvalues of $\msK$ are real. Assume them to be distinct, label
  them $p_{A}$, $A=1,\dots,n$, and order them so that $p_{1}<\cdots<p_{n}$. Let $\bmsX_{A}$, $A=1,\dots,n$, be a local eigenvector field of $\msK$
  corresponding to $p_{A}$, and let $\{\bmsY^{A}\}$ be the basis dual to $\{\bmsX_{A}\}$. Then there are functions $\alpha_{A}$, $A=1,\dots,n$, defined
  on the subset of $\bM$ on which $\{\bmsX_{A}\}$ is defined, which are strictly positive and of the same regularity as ${}^{0}g$ and ${}^{0}k$, such that
  \begin{equation}\label{eq:gzformulavd}
    {}^{0}g=\textstyle{\sum}_{A}\alpha_{A}^{2}t^{2p_{A}}\bmsY^{A}\otimes\bmsY^{A}.
  \end{equation}
  Define the Riemannian metric $\bh$ on $\bM$ by
  \[
  \bh:=\textstyle{\sum}_{A}\alpha_{A}^{2}\bmsY^{A}\otimes\bmsY^{A};
  \]
  the right hand side is independent of the choice of local frame $\{\bmsX_{A}\}$ and therefore the left hand side is well defined on $\bM$. 
  Finally, $\partial_{t}(t\d_{t}{}^{0}\phi)=0$, so that there are functions $\Phia$ and $\Phib$, of the same regularity as ${}^{0}\phi$ and
  $\d_{t}{}^{0}\phi$, such that ${}^{0}\phi=\Phia\ln t+\Phib$. Then $(\bM,\bh,\msK,\Phia,\Phib)$ are non-degenerate quiescent initial data
  on the singularity for the Einstein-scalar field equations with a vanishing cosmological constant, assuming that $1+p_{1}-p_{n-1}-p_{n}>0$.
\end{prop}
\begin{proof}
  The proof is to be found in Section~\ref{section:existenceofdev} below.
\end{proof}
Finally, we appeal to \cite{aarendall,damouretal} in order to deduce the existence of developments in case condition 4 of
Definition~\ref{def:ndsfidonbbs} is void. 

\begin{thm}\label{thm:aardetal}
  Let $(\bM,\bh,\msK,\Phia,\Phib)$ be non-degenerate quiescent initial data on the singularity for the Einstein-scalar field equations with a
  vanishing cosmological constant. Assume the manifold and the data to be real analytic. Assume, moreover, that if the eigenvalues of $\msK$,
  say $p_{A}$, $A=1,\dots,n$, are ordered so that $p_{1}<\cdots<p_{n}$, then $1+p_{1}-p_{n-1}-p_{n}>0$. Then there is an open neighbourhood $V$ of
  $\bM\times \{0\}$ in $\bM\times\rn{}$ and a real analytic solution $(M,g,\phi)$ to the Einstein-scalar field equations, where
  $M:=V\cap [\bM\times (0,\infty)]$, such that the following holds. On $M$, the metric $g$ is given by the right hand side of
  (\ref{eq:rinansatzintroPsi}). In this expression, $\msX_{A}$, $A=1,\dots,n$, is an eigenvector field of $\msK$ corresponding to $p_{A}$, normalised
  so that $|\msX_{A}|_{\bh}=1$, and $\{\msY^{A}\}$ is the basis dual to $\{\msX_{A}\}$. Moreover, if $K$ is a compact subset of $\bM$, there are constants
  $\varepsilon>0$, $t_{0}>0$ and $C>0$ such that
  \begin{subequations}\label{eq:mKchhconvratesfaar}
    \begin{align}
      \|\mK(\cdot,t)-\msK\|_{C(K)} \leq & Ct^{\vare},\label{eq:mKlimsfaar}\\
      \|\chh(\cdot,t)-\bh\|_{C(K)} \leq & Ct^{\vare},\label{eq:limitbABtobhchhsfaar}\\
      \|(\hU \phi)(\cdot,t)-\Phia\|_{C(K)} \leq & Ct^{\vare},\label{eq:Phialimaar}\\
      \|(\phi-\Phia\varrho)(\cdot,t)-\Phib\|_{C(K)} \leq & Ct^{\vare}\label{eq:Phiblimaar}
    \end{align}
  \end{subequations}
  for all $t\leq t_{0}$. Here $\mK$, $\chh$, $\hU$ and $\varrho$ are defined as in Definition~\ref{def:developmentsfLambda}.
\end{thm}
\begin{remark}\label{remark:hotanaly}
  Since the arguments are in the real analytic setting, it should be possible to upgrade (\ref{eq:mKchhconvratesfaar}) to $C^{l}$-estimates.
  However, we do not derive such estimates here. 
\end{remark}
\begin{remark}
  The results in \cite{aarendall,damouretal} also include uniqueness statements. Under additional conditions, the constructed solutions are
  thus unique. We omit the details. 
\end{remark}
\begin{remark}
  In the course of the proof, it is demonstrated that $t\theta$ converges to $1$ at a rate. The metric can thus be also be represented
  as in (\ref{eq:rinansatzintrocrushing}) with $N=1$. Moreover, on compact subsets, estimates of the form (\ref{eq:mKchhconvratesfN}) hold
  with $l=0$, though higher order estimates should also hold; cf. Remark~\ref{remark:hotanaly}.
\end{remark}
\begin{proof}
  The proof is to be found in Section~\ref{section:existenceofdev} below.
\end{proof}

Finally, let us make a remark about the degenerate case.

\begin{remark}\label{remark:thedegeneratecase}
  In Definitions~\ref{def:falid}, \ref{def:vacuumidos}, \ref{def:developmentvacuum}, \ref{def:ndsfidonbbs}, \ref{def:developmentsfLambda} and
  \ref{def:developmentsfLambdacrushing}, the condition of non-degeneracy is important. This is due to the fact that $\msO^{A}_{BC}$ is defined
  in terms of the eigenvector fields corresponding to the distinct eigenvalues and to the fact that the correspondence between $\bh$ and the
  spacetime metric is based on the fact that the eigenvalues are distinct. However, the notion of a velocity dominated solution, in the sense of
  \cite[Subsection~2.2, pp.~482--483]{aarendall}, offers a different perspective in the case of Gaussian foliations. In order to see this, let
  us focus on the $n=3$ dimensional case. Ordering the eigenvalues so that $p_{1}\leq p_{2}\leq p_{3}$, it is of interest to separate the cases
  $p_{1}=p_{2}$ and $p_{2}=p_{3}$. If $p_{1}>0$, the condition on the $\msO^{A}_{BC}$ is void. Let us therefore assume that $p_{1}\leq 0$. If
  $p_{1}=p_{2}\leq 0$, then, due to the conditions on the $p_{i}$, $p_{3}=1$ and $p_{1}=p_{2}=0$. This case is consistent with the flat Kasner
  solutions (which have Cauchy horizons). It may be of interest to allow this possibility, but it can be expected to be quite different from
  the situation we are interested in here. We therefore exclude this case. What remains is $p_{1}\leq 0$ and $p_{2}=p_{3}$. Then $p_{2}=p_{3}>0$.
  In particular the ($2$-dimensional) eigenspace corresponding to $p_{2}$ and $p_{3}$ is given by the orthogonal complement of the eigenspace of
  $p_{1}$. In this setting, one generalisation of the condition $\msO^{1}_{23}=0$ in a neighbourhood of a point $\bx$ such that $p_{1}(\bx)\leq 0$
  is the requirement that there is
  a neighbourhood, say $V$, of $\bx$ such that the distribution, say $D$, of $\bM$ of rank $2$ given by the tangent vectors orthogonal to the
  eigenspace of $p_{1}$, is involutive, so that $D$ is integrable (and, in fact, defines a foliation on $V$).

  The above considerations lead to the following generalisation in the case of $n=3$. Fix an asymptotically velocity dominated solution
  $({}^{0}g,{}^{0}k,{}^{0}\phi)$, assume that $\tr{}^{0}k=-t^{-1}$ and fix $\msK:=-t\cdot {}^{0}K$; cf. the statement of
  Proposition~\ref{prop:vdstoid}. Assume the eigenvalues of $\msK$ (which are automatically $\leq 1$) to be strictly less than $1$.
  Assume, moreover, that if there is an $\bx\in \bM$ such that the smallest eigenvalue of $\msK$ at $\bx$ is $\leq 0$, then there is an
  open neighbourhood $V$ of $\bx$ such that the distribution, say $D$, of $\bM$ of rank $2$ given by the tangent vectors orthogonal to the
  eigenspace of $p_{1}$, is involutive. Under these conditions, we expect the existence of a solution to the Einstein-scalar field equations
  with asymptotics determined by the asymptotically velocity dominated solution as in \cite[Theorem~1, p.~484]{aarendall}. 
\end{remark}

\subsection{Improving the asymptotics}\label{ssection:imprasympt}

The correspondence between the developments in Definitions~\ref{def:developmentvacuum}, \ref{def:developmentsfLambda}
and \ref{def:developmentsfLambdacrushing} and the initial data is given by (\ref{eq:mKchhconvrate}), (\ref{eq:mKchhconvratesf})
and (\ref{eq:mKchhconvratesfN}) respectively. However, as pointed out in connection with Definition~\ref{def:developmentvacuum},
more detailed information might be needed to deduce uniqueness; cf., e.g., \cite[Theorem~1.7, p.~1187--1188]{fal}. In particular, it may
be necessary to control the asymptotics of $\theta$, $\d_{t}\bk^{i}_{\phantom{i}j}$ etc. It may also be necessary to have more detailed
information concerning the rate of convergence. It is therefore of interest to assume that we have developments as in
Definitions~\ref{def:developmentvacuum}, \ref{def:developmentsfLambda} and \ref{def:developmentsfLambdacrushing} and to deduce more
detailed information concerning the asymptotics. We do so here. In addition, it turns out to be possible to derive some of the assumptions
included in the definitions of a development and of initial data. The statement of the following theorem therefore differs somewhat
from the definitions of initial data and of a development. 
\begin{thm}\label{thm:improvingasymptoticsgeneral}
  Let $3\leq n\in\nn{}$, $\Lambda\in\rn{}$, $(\bM,\bh)$ be a smooth and closed $n$-dimensional Riemannian manifold and $\msK$ be a
  smooth $(1,1)$-tensor field on $\bM$. Assume that
  \begin{enumerate}
  \item $\msK$ is symmetric with respect to $\bh$.
  \item $\tr\msK^{2}\leq 1$.    
  \item The eigenvalues $p_{A}$ of $\msK$ are distinct.
  \item $\msO^{A}_{BC}$ vanishes in a neighbourhood of $\bx\in\bM$ if $1+p_{A}(\bx)-p_{B}(\bx)-p_{C}(\bx)\leq 0$, where $\msO^{A}_{BC}$ is given by
    Remark~\ref{remark:msOdefsf}.    
  \end{enumerate}
  Assume that there is a smooth solution $(M_{a},g,\phi)$ to the Einstein-scalar field equations with a cosmological constant $\Lambda$ such that
  $M_{a}=\bM\times (0,t_{+})$ for some $t_{+}>0$; the mean curvature $\theta$ of the leaves $\bM_{t}:=\bM\times \{t\}$ is strictly positive; and
  $g$ takes the form
  \begin{equation}\label{eq:asmetricgeneralcase}
    g=-N^{2} dt\otimes dt+\textstyle{\sum}_{A,B}b_{AB}\theta^{-2p_{\max\{A,B\}}}\msY^{A}\otimes \msY^{B},
  \end{equation}
  where $\{\msY^{A}\}$ is the basis dual to $\{\msX_{A}\}$ introduced in Remark~\ref{remark:msOdefsf}, the lapse function $N$ is strictly positive and
  the $p_{A}$ are ordered so that $p_{1}<\cdots<p_{n}$. Moreover, $\theta$ is assumed to diverge uniformly to $\infty$ as $t\rightarrow 0+$. Let $\chh$
  be defined by (\ref{eq:chhdef}),
  $\mK$ be the expansion normalised Weingarten map of the leaves of the foliation and $\hU$ be the future directed unit normal with respect
  to $\hg:=\theta^{2}g$. Assume (\ref{eq:mKlimsfN}) and (\ref{eq:limitbABtobhchhsfN}) to hold. Finally, assume that there is a constant
  $C_{\rorel}$ such that
  \begin{equation}\label{eq:Crorelbd}
    |\bD\ln\hN|_{\bh}\leq C_{\rorel}
  \end{equation}
  on $M_{0}:=\bM\times (0,t_{0}]$ for some $t_{0}\in (0,t_{+})$, where $\hN:=\theta N$ and $\bD$ denotes the Levi-Civita connection induced by $\bh$, and
  that there are, for each $1\leq k\in\nn{}$, constants $a_{k}$ and $C_{k}$ such that
  \begin{equation}\label{eq:lnthetalnNphipolbd}
    |\bD^{k}\ln \hN|_{\bh}+|\bD^{k}\phi|_{\bh}+|\bD^{k-1}\hU\ln N|_{\bh}\leq C_{k}\ldr{\ln\theta}^{a_{k}}
  \end{equation}
  on $M_{0}$, where $\ldr{\xi}:=(1+|\xi|^{2})^{1/2}$ for $\xi\in\rn{n}$. Let the logarithmic volume density $\varrho$ be defined by the condition
  $e^{\varrho}\mu_{\bh}=\mu_{\bg}$, where $\mu_{\bh}$ is
  the volume form on $\bM$ induced by $\bh$ and $\mu_{\bg}$ is the volume form on $\bM$ given by the metric $\bg$ induced on the leaves $\bM_{t}$.
  Let, moreover, $\{X_{A}\}$ be the eigenframe of $\mK$, normalised so that the $X_{A}$ are unit vector fields with respect to $\bh$ and so that
  $X_{A}\rightarrow\msX_{A}$. Then $\bmu_{A}$ is defined by the condition that $|X_{A}|_{\bge}=e^{\bmu_{A}}$. With this notation, there
  is an $\eta>0$ and, for each $k\in\nn{}$, a constant $C_{k}$ such that
  \begin{equation}\label{eq:lnthetavarrhobmuApAlntheta}
    |\bD^{k}(\ln\theta+\varrho)|_{\bh}+\textstyle{\sum}_{A}|\bD^{k}(\bmu_{A}+p_{A}\ln\theta)|_{\bh}\leq C_{k}\theta^{-2\eta}
  \end{equation}
  on $M_{0}$. Next, there is an $\eta>0$ and, for each $k\in\nn{}$, a constant $C_{k}$ such that
  \begin{subequations}
    \begin{align}
      |\bD^{k}(\hml_{U}\mK)|_{\bh} \leq & C_{k}\theta^{-2\eta},\label{eq:hmlUmKrough}\\
      |\bD^{k}[\mK(\msY^{A},\msX_{B})-p_{B}\delta^{A}_{B}]|_{\bh} \leq & C_{k}\theta^{-2\eta}\min\{1,\theta^{-2(p_{B}-p_{A})}\},\label{eq:mKestrefined}\\
      |\bD^{k}[(\hml_{U}\mK)(\msY^{A},\msX_{B})]|_{\bh} \leq & C_{k}\theta^{-2\eta}\min\{1,\theta^{-2(p_{B}-p_{A})}\}\label{eq:hmlUmKestrefined}
    \end{align}
  \end{subequations}  
  on $M_{0}$ for all $A,B$ (no summation on $B$). In addition, if the deceleration parameter $q$ is defined by the condition that
  $\hU(n\ln\theta)=-1-q$, then there is an $\eta>0$ and, for each $k\in\nn{}$, a constant $C_{k}$ such that
  \begin{equation}\label{eq:qminusnminusoneimprprop}
    |\bD^{k}[q-(n-1)]|_{\bh}\leq C_{k}\theta^{-2\eta}
  \end{equation}
  on $M_{0}$. Moreover, there are $\Phia,\Phib\in C^{\infty}(\bM)$, an $\eta>0$ and, for each $k\in\nn{}$, a constant $C_{k}$ such that
  \begin{align*}
    |\bD^{k}(\hU\phi-\Phia)|_{\bh} \leq & C_{k}\theta^{-2\eta},\\
    |\bD^{k}(\phi-\Phia\varrho-\Phib)|_{\bh} \leq & C_{k}\theta^{-2\eta}
  \end{align*}
  on $M_{0}$. Finally, $\tr\msK=1$, $\tr\msK^{2}+\Phia^{2}=1$ and $\rodiv_{\bh}\msK=\Phia d\Phib$.
\end{thm}
\begin{remark}\label{remark:hmlUmKdef}
  The notation $\hml_{U}\mK$ is introduced in \cite[Section~A.2, pp.~203--204]{RinWave}. However, in the case of a vanishing shift vector field,
  it can be defined as follows. Let $\{E_{i}\}$ be a local frame on $\bM$ (in particular, it is independent of $t$) and $\{\omega^{i}\}$ be
  the dual frame. Then $(\hml_{U}\mK)^{i}_{\phantom{i}j}=\hU(\mK^{i}_{\phantom{i}j})$, where the sub- and superscripts refer to the frame $\{E_{i}\}$ and
  co-frame $\{\omega^{i}\}$; cf. \cite[(A.3), p.~204]{RinWave}. 
\end{remark}
\begin{remark}\label{remark:XAmsXAlocal}
  The frames $\{\msX_{A}\}$ and $\{X_{A}\}$ are only defined locally on $\bM$. However, the $\msX_{A}$ and $X_{A}$ are well defined up to a sign (since
  $\mK$ converges to $\msK$ and $\msK$ has distinct eigenvalues, we can assume $\mK$ to have distinct eigenvalues by restricting $t_{0}$, if
  necessary). This means that $\bmu_{A}$ is globally well defined. Moreover, the estimates (\ref{eq:mKestrefined}) and (\ref{eq:hmlUmKestrefined}) make
  sense globally. 
\end{remark}
\begin{remark}\label{remark:pAsltone}
  Note that we do not assume the asymptotic versions of the Hamiltonian and momentum constraints to hold. Moreover, we do not assume
  (\ref{eq:PhialimN}) or (\ref{eq:PhiblimN}) to hold. However, we do require that
  $\tr\msK^{2}\leq 1$. This requirement, when combined with the fact that $n\geq 3$ and the non-degeneracy, implies that $p_{A}<1$ for
  all $A$. 
\end{remark}
\begin{remark}\label{remark:noncompactsettingimp}
  In the statement of the theorem we assume $\bM$ to be compact. However, by a more technical argument, we expect it should be possible to obtain
  similar conclusions in the non-compact setting.
\end{remark}
\begin{remark}
  The requirement that (\ref{eq:lnthetalnNphipolbd}) holds is perhaps the least satisfying assumption. However, it is hard to avoid making assumptions
  on the lapse function. Here, we bound the first and the last term on the left hand side of (\ref{eq:lnthetalnNphipolbd}). What is perhaps more
  questionable is the bound on the middle term on the left hand side. However, in deriving estimates for the scalar field, we need control of the
  geometry, in particular on the deceleration parameter. However, such control is only obtained by combining the assumptions with Einstein's equations.
  In that context, we need some control of the scalar field. The control necessary could be obtained in different ways. One could also, for example,
  assume that $|\bD^{k}\hU^{2}\phi|_{\bh}$ decays or that $|\bD^{k}\hU\phi|_{\bh}$ does not grow faster than the right hand side of
  (\ref{eq:lnthetalnNphipolbd}). However, that would imply the assumed bound on $\phi$. 
\end{remark}
%\begin{remark}
%  For $\xi\in\rn{n}$, $\ldr{\xi}:=(1+|\xi|^{2})^{1/2}$.
%\end{remark}
\begin{remark}\label{remark:generalformulationlocalcoordinates}
  If $n=3$ and $\msO^{1}_{23}=0$ on $\bM$, then, due to Lemma~\ref{lemma:localcoordinates} below, there are, for each $p\in\bM$, local coordinates
  $(V,\bsfx)$ and strictly positive functions $\alpha_{A}$, $A=1,2,3$, such that $p\in V$ and $\bmsX_{A}=\alpha_{A}\msX_{A}$ (no summation), where the
  $\bmsX_{A}$ are given by (\ref{eq:preferredcoordinates}) below. As a consequence, $\bmsY^{A}=\alpha_{A}^{-1}\msY^{A}$ (no summation), where the $\bmsY^{A}$
  are given by (\ref{eq:msYAformulae}) below. Moreover,
  \[
  d\bsfx^{i}=\bmsX_{A}^{i}\bmsY^{A},\ \ \
  \d_{i}=\bmsY^{A}_{i}\bmsX_{A},
  \]
  where $\bmsX_{2}^{3}$, $\bmsX_{1}^{2}$ and $\bmsX_{1}^{3}$ are introduced in (\ref{eq:preferredcoordinates}) below. The remaining components
  of $\bmsX_{A}^{i}$ are determined by $\bmsX_{A}^{i}=0$ if $i<A$ and $\bmsX_{A}^{i}=1$ if $A=i$. Moreover,
  \[
  \bmsY^{3}_{2}=-\bmsX^{3}_{2},\ \ \
  \bmsY^{2}_{1}=-\bmsX^{2}_{1},\ \ \
  \bmsY^{3}_{1}=\bmsX^{2}_{1}\bmsX^{3}_{2}-\bmsX^{3}_{1},
  \]
  $\bmsY_{i}^{A}=0$ if $A<i$ and $\bmsY_{i}^{A}=1$ if $A=i$. Note also that 
  \[
  d\bsfx^{i}=\msX_{A}^{i}\msY^{A},\ \ \
  \d_{i}=\msY^{A}_{i}\msX_{A},
  \]
  where $\msX_{A}^{i}=\alpha_{A}^{-1}\bmsX_{A}^{i}$ (no summation) and $\msY_{i}^{A}=\alpha_{A}\bmsY_{i}^{A}$ (no summation). Combining these observations
  with (\ref{eq:mKestrefined}) and (\ref{eq:hmlUmKestrefined}), it can be verified that
  \begin{align*}
    |D^{k}[\mK(d\bsfx^{i},\d_{j})-\msK(d\bsfx^{i},\d_{j})]|_{\bh} \leq & C_{k}\theta^{-2\eta}\min\{1,\theta^{-2(p_{j}-p_{i})}\},\\
    |\bD^{k}[(\hml_{U}\mK)(d\bsfx^{i},\d_{j})]|_{\bh} \leq & C_{k}\theta^{-2\eta}\min\{1,\theta^{-2(p_{j}-p_{i})}\}
  \end{align*}  
  on $M_{a}$ for all $i,j$.
\end{remark}
\begin{proof}
  The proof is to be found in Subsection~\ref{ssection:proofpropimprovegen} below. 
\end{proof}

In the case of a Gaussian foliation, there is a similar statement. However, the assumptions are slightly different. 

\begin{thm}\label{thm:improvingasymptoticsgeneralG}
  Let $3\leq n\in\nn{}$, $\Lambda\in\rn{}$, $(\bM,\bh)$ be a smooth and closed $n$-dimensional Riemannian manifold, $\msK$ be a
  smooth $(1,1)$-tensor field and $\Phia$ a smooth function on $\bM$. Assume that
  \begin{enumerate}
  \item $\msK$ is symmetric with respect to $\bh$.
  \item $\tr\msK^{2}+\Phia^{2}=1$.    
  \item The eigenvalues $p_{A}$ of $\msK$ are distinct.
  \item $\msO^{A}_{BC}$ vanishes in a neighbourhood of $\bx\in\bM$ if $1+p_{A}(\bx)-p_{B}(\bx)-p_{C}(\bx)\leq 0$, where $\msO^{A}_{BC}$ is given by
    Remark~\ref{remark:msOdefsf}.    
  \end{enumerate}
  Assume that there is a smooth solution $(M_{a},g,\phi)$ to the Einstein-scalar field equations with a cosmological constant $\Lambda$ such that
  $M_{a}=\bM\times (0,t_{+})$ for some $t_{+}>0$; the mean curvature $\theta$ of the leaves $\bM_{t}:=\bM\times \{t\}$ satisfies $\theta\geq\theta_{0}$
  on $M_{a}$ for some $0<\theta_{0}\in\rn{}$, where $\theta_{0}$ is additionally required to satisfy $\theta_{0}>[2\Lambda/(n-1)]^{1/2}$ in case
  $\Lambda>0$. Assume, in addition, that $g$ takes the form
  \begin{equation}\label{eq:asmetricgeneralGaussiancase}
    g=-dt\otimes dt+\textstyle{\sum}_{A,B}b_{AB}t^{2p_{\max\{A,B\}}}\msY^{A}\otimes \msY^{B},
  \end{equation}
  where $\{\msY^{A}\}$ is the basis dual to $\{\msX_{A}\}$ introduced in Remark~\ref{remark:msOdefsf} and the $p_{A}$ are ordered so that
  $p_{1}<\cdots<p_{n}$. Let $\chh$ be defined by (\ref{eq:chhdef}),
  $\mK$ be the expansion normalised Weingarten map of the leaves of the foliation and $\hU$ be the future directed unit normal with respect
  to $\hg:=\theta^{2}g$. Assume that (\ref{eq:mKlimsf})--(\ref{eq:Phialim}) hold with $\vare\leq 2$. Let the logarithmic volume
  density $\varrho$ be defined by the condition $e^{\varrho}\mu_{\bh}=\mu_{\bg}$, where $\mu_{\bh}$ is the volume form on $\bM$ induced by $\bh$ and
  $\mu_{\bg}$ is the volume form on $\bM$ given by the metric $\bg$ induced on the leaves $\bM_{t}$. Let, moreover, $\{X_{A}\}$ be the eigenframe
  of $\mK$, normalised so that the $X_{A}$ are unit vector fields with respect to $\bh$ and so that $X_{A}\rightarrow\msX_{A}$. Then $\bmu_{A}$ is
  defined by the condition that $|X_{A}|_{\bge}=e^{\bmu_{A}}$. With this notation, there is an $\eta>0$, a $t_{0}\in (0,t_{+})$ and,
  for each $k\in\nn{}$, a constant $C_{k}$ such that
  \begin{subequations}\label{eq:lnthetavarrhorelationG}
    \begin{align}
      \|(\ln\theta+\varrho)(\cdot,t)\|_{C^{k}(\bM)}+\textstyle{\sum}_{A}\|\bmu_{A}(\cdot,t)-p_{A}\ln t\|_{C^{k}(\bM)} \leq & C_{k}t^{2\eta},
      \label{eq:lnthetavrbmuApAlntG}\\
      \|\ln\theta(\cdot,t)+\ln t\|_{C^{k}(\bM)}+\|t\theta(\cdot,t)-1\|_{C^{k}(\bM)} \leq & C_{k}t^{2\eta}\label{eq:lnthetalnttthetamoneG}
    \end{align}
  \end{subequations}  
  for $t\leq t_{0}$. Next, there is an $\eta>0$ and, for each $k\in\nn{}$, a constant $C_{k}$ such that
  \begin{subequations}
    \begin{align}
      \|(\hml_{U}\mK)(\cdot,t)\|_{C^{k}(\bM)} \leq & C_{k}t^{2\eta},\label{eq:hmlUmKroughG}\\
      \|[\mK(\msY^{A},\msX_{B})-p_{B}\delta^{A}_{B}](\cdot,t)\|_{C^{k}(\bM)} \leq & C_{k}t^{2\eta}\min\{1,t^{2(p_{B}-p_{A})}\},\label{eq:mKestrefinedG}\\
      \|[(\hml_{U}\mK)(\msY^{A},\msX_{B})](\cdot,t)\|_{C^{k}(\bM)} \leq & C_{k}t^{2\eta}\min\{1,t^{2(p_{B}-p_{A})}\}\label{eq:hmlUmKestrefinedG}
    \end{align}
  \end{subequations}  
  for  $t\leq t_{0}$ and all $A,B$ (no summation on $B$). In addition, if the deceleration parameter $q$ is defined by the condition that
  $\hU(n\ln\theta)=-1-q$, then there is an $\eta>0$ and, for each $k\in\nn{}$, a constant $C_{k}$ such that
  \[
  \|q(\cdot,t)-(n-1)\|_{C^{k}(\bM)}\leq C_{k}t^{2\eta}
  \]
  for $t\leq t_{0}$. Moreover, there is a $\Phib\in C^{\infty}(\bM)$, an $\eta>0$ and, for each $k\in\nn{}$, a constant $C_{k}$ such that,
  in addition to (\ref{eq:Phialim}),
  \begin{align*}
    \|\phi(\cdot,t)-\Phia\varrho(\cdot,t)-\Phib\|_{C^{k}(\bM)} \leq & C_{k}t^{2\eta}
  \end{align*}
  for $t\leq t_{0}$. Finally, $\tr\msK=1$ and $\rodiv_{\bh}\msK=\Phia d\Phib$.
\end{thm}
\begin{remark}
  Comments analogous to Remarks~\ref{remark:hmlUmKdef}, \ref{remark:XAmsXAlocal} and \ref{remark:noncompactsettingimp} are equally relevant here. 
\end{remark}
\begin{remark}
  The assumptions of this theorem are stronger than those of Theorem~\ref{thm:improvingasymptoticsgeneral} in the sense that we assume
  $\tr\msK^{2}+\Phia^{2}=1$ to be satisfied and $\hU\phi$ to converge to $\Phia$. On the other hand, they are weaker in the sense that
  we assume neither (\ref{eq:Crorelbd}) nor (\ref{eq:lnthetalnNphipolbd}) to hold. Note also that we obtain the conclusions
  (\ref{eq:lnthetalnttthetamoneG}) concerning the asymptotics of $\theta$, even though the assumptions involve expansion normalised quantities. 
\end{remark}
\begin{remark}
  Assuming, in addition to the conditions of the theorem, that $n=3$ and $\msO^{1}_{23}=0$ on $\bM$, then, combining the conclusions of the theorem
  with the observations of Remark~\ref{remark:generalformulationlocalcoordinates} yields
  \begin{align*}
    t|\bD^{k}(\bK^{i}_{j}-t^{-1}\msK^{i}_{j})|_{\bh} \leq & C_{k}t^{2\eta}\min\{1,t^{2(p_{j}-p_{i})}\},\\
    t^{2}|\bD^{k}\d_{t}(\bK^{i}_{j}-t^{-1}\msK^{i}_{j})|_{\bh} \leq & C_{k}t^{2\eta}\min\{1,t^{2(p_{j}-p_{i})}\}
  \end{align*}
  with respect to the local coordinates introduced in Remark~\ref{remark:generalformulationlocalcoordinates}. Here $\msK^{i}_{j}$ denotes the components of
  $\msK$ with respect to the frame $\{\d_{i}\}$ and its dual frame $\{d\bsfx^{i}\}$, and similarly for $\bK$. These estimates should be compared
  with (\ref{eq:kconvfal}).
\end{remark}
\begin{proof}
  The proof of this statement is given in Subsection~\ref{ssection:proofpropimprovegenG} below
\end{proof}

\subsection{From convergence to data on the singularity}\label{ssection:contodataonsing}
It is of great interest to characterise the set of vacuum solutions that converge to the Kasner circle. Moreover, for reasons clarified in
Subsection~\ref{ssection:oscibbsing} below, the non-degenerate subset of the Kasner circle is of greatest interest (i.e., the subset on which the
$p_{i}$ are distinct). More generally, the same question can be asked concerning the Kasner disc. The situation we have in mind is thus that
$\hml_{U}\mK$ (cf. Remark~\ref{remark:hmlUmKdef}) converges to zero; that $\hU^{2}\phi$ converges to zero; that $\mK$ converges to a limit $\msK$;
that $\hU\phi$ converges to a
limit $\Phia$; and that $\Phia^{2}+\tr\msK^{2}=1$. However, in order to derive meaningful estimates, we need to choose a norm with respect
to which we assume convergence. We should also specify a rate. In the case of a Gaussian foliation, one could use powers of $t$ to specify the
rate. However, this does not work more generally. Another option is negative powers of the mean curvature. However, this quantity involves one
derivative of the metric. On the other hand, in the situations discussed above, the logarithmic volume density $\varrho$ is comparable in size
to the logarithm of the mean curvature; cf., e.g., (\ref{eq:lnthetavarrhobmuApAlntheta}) and (\ref{eq:lnthetavarrhorelationG}). This quantity
does not involve any derivatives of
the metric and can serve as a substitute for the logarithm of the mean curvature. In what follows, we therefore use $\varrho$ to specify a
rate. For this reason, we, at the very minimum, need to require $\varrho$ to tend to $-\infty$ uniformly. This leads to the following definition.

\begin{definition}\label{def:volsing}
  Let $3\leq n\in\nn{}$, $\bM$ be a closed $n$-dimensional manifold, $0<t_{+}\in\rn{}$, $I=(0,t_{+})$ and $M:=\bM\times I$. Let $\Lambda\in\rn{}$.
  Assume $(M,g)$ to be a time oriented Lorentz manifold and $\phi$ to be a smooth function on $M$, solving the Einstein-scalar field equations
  with a cosmological constant $\Lambda$.
  Assume the leaves $\bM_{t}:=\bM\times\{t\}$ of the foliation to be spacelike and denote the future directed unit normal by $U$. 
  Define the lapse function $N$ and the shift vector field $\chi$ by $\d_{t}=NU+\chi$. Assume $N>0$ and $\chi=0$. Assume, moreover,
  that there is a $0<\theta_{0}\in\rn{}$ such that the mean curvature $\theta$ of the leaves of the foliation satisfies $\theta\geq\theta_{0}$.
  In case $\Lambda>0$, assume, in addition, that $\theta_{0}>[2\Lambda/(n-1)]^{1/2}$. Fix a smooth reference metric $\bg_{\refer}$ on
  $\bM$ and define the logarithmic volume density $\varrho$ by the requirement that $e^{\varrho}\mu_{\bg_{\refer}}=\mu_{\bg}$, where $\bg$ is
  the metric induced on the leaves of the foliation. Assume, finally, that $\varrho$ diverges uniformly to $-\infty$ as $t\rightarrow 0$
  in the sense that for every $C\in\rn{}$, there is a $T\in I$ such that $\varrho(\bx,t)\leq C$ for all $t\leq T$ and $\bx\in\bM$.
  Then $(M,g,\phi)$ is said to be \textit{a solution to the Einstein-scalar field equations with a cosmological constant $\Lambda$ and a
  uniform volume singularity at $t=0$}. 
\end{definition}

Turning to the formal requirement of decay of $\hml_{U}\mK$, we are interested in solutions to the Einstein-scalar field equations with a
cosmological constant $\Lambda$ and a uniform volume singularity at $t=0$ (in the sense of Definition~\ref{def:volsing}) such that there
is a function $\e:\bM\rightarrow (0,1)$ satisfying $\e\geq\e_{0}$ for some $\e_{0}>0$, a $k\in\nn{}$ and constants $a_{k}$ and $C_{k}$ such that
\begin{equation}\label{eq:CkexpdechmlUmK}
  \textstyle{\sum}_{l=0}^{k}|\bD^{l}\hml_{U}\mK|_{\bg_{\refer}}\leq C_{k}\ldr{\varrho}^{a_{k}}e^{2\e\varrho}
\end{equation}
on $M$, where $\bD$ is the Levi-Civita connection of $\bg_{\refer}$. In what follows, $a_{k}$ and $C_{k}$ will change from line to line, but the
function $\e$ remains fixed. In analogy with (\ref{eq:CkexpdechmlUmK}), we demand that
\begin{equation}\label{eq:bDlhUsqphiestimate}
  \textstyle{\sum}_{l=0}^{k}|\bD^{l}\hU^{2}\phi|_{\bg_{\refer}}\leq C_{k}\ldr{\varrho}^{a_{k}}e^{2\e\varrho}
\end{equation}
on $M$. 

In case (\ref{eq:CkexpdechmlUmK}) holds with $k=0$, it can be deduced that $\mK$ converges uniformly to a continuous limit $\msK$; cf.
Lemma~\ref{lemma:mKconvtomsK} below. As mentioned previously, we are here interested in the non-degenerate setting. We therefore assume the
eigenvalues of $\msK$ to be distinct. Due to Lemma~\ref{lemma:nondegenerate} below, it is no restriction to, in this situation, make the following
assumption.

\begin{definition}\label{def:standardassumptions}
  Let $(M,g,\phi)$ be a solution to the Einstein-scalar field equations with a cosmological constant $\Lambda$ and a uniform volume singularity at
  $t=0$. Let $\mK$ denote the associated expansion normalised Weingarten map and denote the eigenvalues of $\mK$ by $\ell_{A}$, $A=1,\dots,n$.
  Assume that (\ref{eq:CkexpdechmlUmK}) holds with $k=0$. Assume, moreover,
  that there is an $\e_{\rond}>0$ such that $\min_{A\neq B}|\ell_{A}-\ell_{B}|\geq\e_{\rond}$ on $M$, that $\ell_{1}<\dots<\ell_{n}$ and that there are
  global eigenvector fields $X_{A}$ such that $\mK X_{A}=\ell_{A}X_{A}$ (no summation) and $|X_{A}|_{\bg_{\refer}}=1$. Then the \textit{standard
  assumptions} are said to be satisfied. 
\end{definition}
\begin{remark}\label{remark:msXAmsYAdef}
  Under the standard assumptions, the expansion normalised Weingarten map $\mK$ converges uniformly to a limit $\msK$. Moreover, $\msK$ has
  distinct eigenvalues, denoted $p_{1}<\cdots<p_{n}$, and there is a frame of eigenvector fields $\{\msX_{A}\}$ of $\msK$
  (with $\msK\msX_{A}=p_{A}\msX_{A}$ (no summation)), with dual frame $\{\msY^{A}\}$, such that $X_{A}$ and $Y^{A}$ converge uniformly to $\msX_{A}$ and
  $\msY^{A}$ respectively, where $\{Y^{A}\}$ is the dual basis of $\{X_{A}\}$. In fact, (\ref{eq:msKmKexpunifconv}) and (\ref{eq:XAasYAas}) below hold. 
\end{remark}
Below we use the following terminology.
\begin{definition}\label{def:XABYABbmuAmuA}
  Given that the standard assumptions hold, let $\{X_{A}\}$, $\{Y^{A}\}$, $\{\msX_{A}\}$ and $\{\msY^{A}\}$ be given by
  Definition~\ref{def:standardassumptions} and Remark~\ref{remark:msXAmsYAdef}. Define $X_{A}^{B}$ and $Y^{B}_{A}$ by the relations
  $X_{A}=X_{A}^{B}\msX_{B}$ and $Y^{A}=Y^{A}_{B}\msY^{B}$. Finally, define $\bmu_{A}$ and $\mu_{A}$ by $|X_{A}|_{\bge}=e^{\bmu_{A}}$ 
  and $\mu_{A}:=\bmu_{A}+\ln\theta$. 
\end{definition}
Given these assumptions, a $C^{0}$-version of the form (\ref{eq:rinansatzintrocrushing}) can be derived with $\theta$ replaced by $e^{-\varrho}$.
\begin{thm}\label{thm:asymptoticformofmetric}
  Assume the standard assumptions to be satisfied and let $p_{A}$, $\{\msX_{A}\}$, $\{\msY^{A}\}$, $X_{A}^{B}$, $Y^{B}_{A}$ and $\bmu_{A}$ be given by
  Remark~\ref{remark:msXAmsYAdef} and Definition~\ref{def:XABYABbmuAmuA}. Then 
  \begin{equation}\label{eq:gbABgenasform}
    g=-N^{2}dt\otimes dt+\textstyle{\sum}_{A,B}b_{AB}e^{2p_{\max\{A,B\}}\varrho}\msY^{A}\otimes \msY^{B}
  \end{equation}
  on $M$. Moreover, there are constants $a$ and $C$, a $t_{0}\in (0,t_{+})$ and continuous functions $r_{A}$ on $\bM$ such that
  \begin{equation}\label{eq:bABasymptotics}
    |b_{AB}-e^{2r_{A}}\delta_{AB}|\leq C\ldr{\varrho}^{a}e^{2\e\varrho}
  \end{equation}
  (no summation on $A$) on $M_{0}:=\bM\times (0,t_{0}]$. Finally,
  \begin{subequations}
    \begin{align}
      |\bmu_{A}-p_{A}\varrho-r_{A}| \leq & C\ldr{\varrho}e^{2\e\varrho},\label{eq:bmuACzasymptoticsintro}\\
      |X_{B}^{A}-\delta_{B}^{A}|+|Y_{B}^{A}-\delta_{B}^{A}| \leq & C\ldr{\varrho}^{a}e^{2\e\varrho}\min\{1,e^{2(p_{B}-p_{A})\varrho}\},
      \label{eq:XABYABoptimalestimatesintro}\\
      |\mK(\msY^{A},\msX_{B})-p_{B}\delta^{A}_{B}| \leq & C\ldr{\varrho}^{a}e^{2\e\varrho}\min\{1,e^{2(p_{B}-p_{A})\varrho}\},\label{eq:mKmsYAmsXBgenest}\\
      |(\hml_{U}\mK)(\msY^{A},\msX_{B})| \leq & C\ldr{\varrho}^{a}e^{2\e\varrho}\min\{1,e^{2(p_{B}-p_{A})\varrho}\}\label{eq:hmlUfixedframeopt}
    \end{align}
  \end{subequations}  
  (no summation on $B$) on $M_{0}$. 
\end{thm}
\begin{remark}
  Assume, in addition, that (\ref{eq:bDlhUsqphiestimate}) holds for $k=0$, so that $\hU\phi$ converges to a limit, say $\Phia$; cf.
  Lemma~\ref{lemma:mKconvtomsK}. Assuming, moreover, that $\Phia^{2}+\tr\msK^{2}=1$, it can be deduced that there are constants $C$
  and $a$ and a continuous function $r_{\theta}$ such that 
  \[
  |\ln\theta+\varrho-r_{\theta}|\leq C\ldr{\varrho}^{a}e^{2\e\varrho}
  \]
  on $M$; cf. Lemma~\ref{lemma:lnthetavarrhoas} below. In particular, we can thus, effectively, replace $\varrho$ with $-\ln\theta$ in
  the conclusions of the theorem. 
\end{remark}
\begin{proof}
  The proof is to be found at the end of Subsection~\ref{ssection:Czlim}, prior to Lemma~\ref{lemma:lnthetavarrhoas}. 
\end{proof}

In order to obtain estimates for higher order derivatives, we also need to impose conditions on the lapse function. In practice, it turns out
to be convenient to impose conditions on $\hN=\theta N$.

\begin{definition}\label{def:kstandass}
  If the standard assumptions hold, cf. Definition~\ref{def:standardassumptions}; (\ref{eq:CkexpdechmlUmK}) holds for some $1\leq k\in\nn{}$;
  and there are constants $C_{k}$ and $a_{k}$ such that
  \begin{equation}\label{eq:bDllnhNconditions}
    \textstyle{\sum}_{l=1}^{k}|\bD^{l}\ln\hN|_{\bg_{\refer}}\leq C_{k}\ldr{\varrho}^{a_{k}}
  \end{equation}
  on $M$, then the $k$-\textit{standard assumptions} are said to hold. 
\end{definition}

\begin{thm}\label{thm:asymptoticformofmetrichod}
  Let $1\leq k\in\nn{}$ and assume the $k$-standard assumptions to hold; cf. Definition~\ref{def:kstandass}. Let $\msK$, $p_{A}$, $\{\msX_{A}\}$,
  $\{\msY^{A}\}$, $X_{A}^{B}$, $Y^{B}_{A}$, $\mu_{A}$ and $\bmu_{A}$ be given by Remark~\ref{remark:msXAmsYAdef} and Definition~\ref{def:XABYABbmuAmuA}.
  Then $\bmu_{A}$ and $\mu_{A}$ are $C^{\infty}$ and $p_{A}$, $\{\msX_{A}\}$, $\{\msY^{A}\}$, $X_{A}^{B}$ and $Y^{B}_{A}$ are $C^{k}$. Moreover,  
  \begin{equation}\label{eq:gbABgenasformhod}
    g=-N^{2}dt\otimes dt+\textstyle{\sum}_{A,B}b_{AB}e^{2p_{\max\{A,B\}}\varrho}\msY^{A}\otimes \msY^{B}
  \end{equation}
  and there are constants $C_{k}$ and $a_{k}$ and a $t_{0}\in (0,t_{+})$ such that
  \begin{equation}\label{eq:bABasymptoticshod}
    |\bD^{l}(b_{AB}-e^{2r_{A}}\delta_{AB})|_{\bg_{\refer}}\leq C_{k}\ldr{\varrho}^{a_{k}}e^{2\e\varrho}
  \end{equation}
  (no summation on $A$) on $M_{0}:=\bM\times (0,t_{0}]$ for all $A,B\in\{1,\dots,n\}$ and all $0\leq l\leq k$. Here $r_{A}$ are $C^{k}$ functions on
  $\bM$ that coincide with the functions $r_{A}$ obtained in Theorem~\ref{thm:asymptoticformofmetric}. Moreover,
  \begin{subequations}
    \begin{align}      
      |\bD^{l}(\bmu_{A}-p_{A}\varrho-r_{A})|_{\bg_{\refer}} \leq & C\ldr{\varrho}^{a_{k}}e^{2\e\varrho},\label{eq:CkrAasymptoticsintro}\\
      |\bD^{l}(X_{B}^{A}-\delta_{B}^{A})|_{\bg_{\refer}}+|\bD^{l}(Y_{B}^{A}-\delta_{B}^{A})|_{\bg_{\refer}}
      \leq & C_{k}\ldr{\varrho}^{a_{k}}e^{2\e\varrho}\min\{1,e^{2(p_{B}-p_{A})\varrho}\},\label{eq:XABYABoptimalestimateshigherorderintro}\\
      |\bD^{l}(\mK(\msY^{A},\msX_{B})-p_{B}\delta^{A}_{B})|_{\bg_{\refer}} \leq &
      C_{k}\ldr{\varrho}^{a_{k}}e^{2\e\varrho}\min\{1,e^{2(p_{B}-p_{A})\varrho}\},\label{eq:mKmsYAmsXBgenesthod}\\
      |\bD^{l}[(\hml_{U}\mK)(\msY^{A},\msX_{B})]|_{\bg_{\refer}} \leq &
      C_{k}\ldr{\varrho}^{a_{k}}e^{2\e\varrho}\min\{1,e^{2(p_{B}-p_{A})\varrho}\}\label{eq:hmlUfixedframeopthod}
    \end{align}
  \end{subequations}  
  (no summation on $B$) on $M_{0}$ for all $A,B\in\{1,\dots,n\}$ and all $0\leq l\leq k$. 

  Assume, in addition to the above, that (\ref{eq:bDlhUsqphiestimate}) holds. Then there are $C^{k}$-functions $\Phia$ and $\Phib$ on
  $\bM$ such that
  \begin{equation}\label{eq:bDlhUphimPhiaPhibestimateintro}
    |\bD^{l}[\hU(\phi)-\Phia]|_{\bg_{\refer}}
    +|\bD^{l}(\phi-\Phia\varrho-\Phib)|_{\bg_{\refer}}\leq C_{k}\ldr{\varrho}^{a_{k}}e^{2\e\varrho}
  \end{equation}
  on $M_{0}$ for all $0\leq l\leq k$. Assume that $\Phia$ satisfies $\Phia^{2}+\tr\msK^{2}=1$. Then there is a function $r_{\theta}\in C^{k}(\bM)$
  such that
  \begin{equation}\label{eq:lnthetapvarrhoasymptoticsintro}
    |\bD^{l}(\ln\theta+\varrho-r_{\theta})|_{\bg_{\refer}}\leq C_{k}\ldr{\varrho}^{a_{k}}e^{2\e\varrho}
  \end{equation}
  on $M_{0}$ for all $0\leq l\leq k$. Thus (\ref{eq:gbABgenasformhod}) and (\ref{eq:bABasymptoticshod}) can be rewritten
  \begin{subequations}\label{seq:gchcompas}
    \begin{align}
      g = & -N^{2}dt\otimes dt+\textstyle{\sum}_{A,B}\bar{b}_{AB}\theta^{-2p_{\max\{A,B\}}}\msY^{A}\otimes \msY^{B},\label{eq:gbABgenasformtheta}\\
      |\bD^{l}(\bar{b}_{AB}-e^{2\bar{r}_{A}}\delta_{AB})|_{\bg_{\refer}} \leq & C_{k}\ldr{\ln\theta}^{a_{k}}\theta^{-2\e}\label{eq:bABasymptoticstheta}
    \end{align}
  \end{subequations}  
  (no summation on $A$) on $M_{0}$ for $0\leq l\leq k$. Here $\bar{r}_{A}=r_{A}+p_{A}r_{\theta}$ and
  \begin{equation}\label{eq:barbABdef}
    \bar{b}_{AB}=b_{AB}e^{2p_{\max\{A,B\}}(\varrho+\ln\theta)},
  \end{equation}
  where $b_{AB}$ are the functions appearing in (\ref{eq:gbABgenasformhod}). Next,
  \begin{subequations}\label{eq:Ckestimatesasgeometrycrusing}
    \begin{align}
      |\bD^{l}(\bmu_{A}-p_{A}\varrho-r_{A})|_{\bg_{\refer}} \leq & C\ldr{\ln\theta}^{a_{k}}\theta^{-2\e},\label{eq:CkrAasymptoticsintrolnth}\\
      |\bD^{l}(\mu_{A}-(p_{A}-1)\varrho-r_{A}-r_{\theta})|_{\bg_{\refer}} \leq & C_{k}\ldr{\ln\theta}^{a_{k}}\theta^{-2\e},
      \label{eq:muAasymptoticsprelintro}\\
      |\bD^{l}(X_{B}^{A}-\delta_{B}^{A})|_{\bg_{\refer}}+|\bD^{l}(Y_{B}^{A}-\delta_{B}^{A})|_{\bg_{\refer}}
      \leq & C_{k}\ldr{\ln\theta}^{a_{k}}\theta^{-2\e}\min\{1,\theta^{-2(p_{B}-p_{A})}\},\label{eq:XABYABoptimalestimateshigherorderintrolnth}\\
      |\bD^{l}[\mK(\msY^{A},\msX_{B})-p_{B}\delta^{A}_{B}]|_{\bg_{\refer}} \leq &
      C_{k}\ldr{\ln\theta}^{a_{k}}\theta^{-2\e}\min\{1,\theta^{-2(p_{B}-p_{A})}\},\label{eq:mKmsYAmsXBgenesttheta}\\
      |\bD^{l}[(\hml_{U}\mK)(\msY^{A},\msX_{B})]|_{\bg_{\refer}} \leq &
      C_{k}\ldr{\ln\theta}^{a_{k}}\theta^{-2\e}\min\{1,\theta^{-2(p_{B}-p_{A})}\}\label{eq:hmlUfixedframeopttheta}
    \end{align}
  \end{subequations}  
  (no summation on $B$) on $M_{0}$ for $0\leq l\leq k$. In addition
  \begin{equation}\label{eq:hUlnthetaasymptoticsintermsoftheta}
    |\bD^{l}[q-(n-1)]|_{\bg_{\refer}}+|\bD^{l}[\hU(\ln\theta)+1]|_{\bg_{\refer}}\leq C_{k}\ldr{\ln\theta}^{a_{k}}\theta^{-2\e}
  \end{equation}
  on $M_{0}$ for $0\leq l\leq k$.
  
  If, in addition, $N=1$, then
  \begin{equation}\label{eq:thetaasymptoticsGaussianfoliation}
    |\bD^{l}(\d_{t}\theta^{-1}-1)|_{\bg_{\refer}}+|\bD^{l}(t\theta-1)|_{\bg_{\refer}}\leq C_{k}\ldr{\ln t}^{a_{k}}t^{2\e}
  \end{equation}
  on $M_{0}$ for $0\leq l\leq k$. Combining this estimate with (\ref{seq:gchcompas})--(\ref{eq:barbABdef}) yields
  \begin{subequations}
    \begin{align}
      g = & -dt\otimes dt+\textstyle{\sum}_{A,B}c_{AB}t^{2p_{\max\{A,B\}}}\msY^{A}\otimes \msY^{B},\label{eq:gbABgenasformGauss}\\
      |\bD^{l}(c_{AB}-e^{2\bar{r}_{A}}\delta_{AB})| \leq & C_{k}\ldr{\ln t}^{a_{k}}t^{2\e},\label{eq:bABasymptoticsGauss}    
    \end{align}
  \end{subequations}  
  (no summation on $A$ or $B$) on $M_{0}$ for $0\leq l\leq k$. Here
  \begin{equation}\label{eq:cABdef}
    c_{AB}:=\bar{b}_{AB}(t\theta)^{-2p_{\max\{A,B\}}}, 
  \end{equation}
  where $\bar{b}_{AB}$ is given by (\ref{eq:barbABdef}). 
\end{thm}
\begin{remark}
  When (\ref{eq:lnthetapvarrhoasymptoticsintro}) is satisfied, $\varrho$ can be replaced by $-\ln\theta+r_{\theta}$ in
  (\ref{eq:bDlhUphimPhiaPhibestimateintro}), (\ref{eq:CkrAasymptoticsintrolnth}) and (\ref{eq:muAasymptoticsprelintro}). In other words,
  modifying $\Phib$ etc., $\varrho$ can effectively be replaced by $-\ln\theta$. 
\end{remark}
\begin{remark}
  As in the previous remark, when (\ref{eq:thetaasymptoticsGaussianfoliation}) holds, $\theta$ can, effectively, be replaced by $t^{-1}$ in
  (\ref{eq:Ckestimatesasgeometrycrusing}).
\end{remark}
\begin{proof}
  The proof is to be found at the end of Subsection~\ref{ssection:higherorderderivativesconv}. 
\end{proof}

One of the main motivations behind the notion of initial data on the singularity introduced in this article is the expectation that it can be used
to parametrise the set of solutions to the Einstein-scalar field equations with quiescent singularities, in the non-degenerate setting. In order to
justify this expectation, it is necessary to extract initial data on the singularity from the asymptotics of a convergent solution. Moreover, it is
necessary to prove the desired relation between the solution and the extracted initial data. In other words, assuming a solution to be convergent,
we need to reproduce the conditions appearing in Definitions~\ref{def:ndsfidonbbs} and \ref{def:developmentsfLambdacrushing}. Assume, to this
end, that the conditions of Theorem~\ref{thm:asymptoticformofmetrichod} are satisfied, except for the requirement that $N=1$ (in the last part of the
theorem). Using the notation of Theorem~\ref{thm:asymptoticformofmetrichod}, introduce the $C^{k}$ Riemannian metric 
\[
\bh:=\textstyle{\sum}_{A}e^{2\br_{A}}\msY^{A}\otimes\msY^{A}.
\]
Before proceeding, note that the $\varrho$ appearing in the statement of Theorem~\ref{thm:asymptoticformofmetrichod} differs from the $\varrho$
appearing in Definition~\ref{def:developmentsfLambdacrushing}; the difference is due to the use of different reference metrics in the respective
definitions. Here we use $\varrho$ to denote the object appearing in the statement of Theorem~\ref{thm:asymptoticformofmetrichod}, but introduce
$\bvr$ to denote the logarithmic volume density defined using $\bh$ as a reference metric; i.e., $e^{\bvr}\mu_{\bh}=\mu_{\bge}$. Next, in order to be
consistent with the representation (\ref{eq:rinansatzintrocrushing}) of the metric, introduce $\bmsX_{A}=e^{-\br_{A}}\msX_{A}$ and let $\{\bmsY^{A}\}$
be the frame dual to $\{\bmsX_{A}\}$. Then $\{\bmsX_{A}\}$ is an orthonormal frame with respect to $\bh$ and $g$ can be written
\begin{equation}\label{eq:asmetricgeneralcasereproduce}
  g=-N^{2} dt\otimes dt+\textstyle{\sum}_{A,B}\chb_{AB}\theta^{-2p_{\max\{A,B\}}}\bmsY^{A}\otimes \bmsY^{B},
\end{equation}
where $\chb_{AB}=e^{-\br_{A}-\br_{B}}\bar{b}_{AB}$ (no summation) are $C^{k}$ functions satisfying
\begin{equation}\label{eq:chbmdeltaABestimate}
  |\bD^{l}(\chb_{AB}-\delta_{AB})|_{\bg_{\refer}} \leq C_{k}\ldr{\ln\theta}^{a_{k}}\theta^{-2\e}
\end{equation}
on $M_{0}$ for $0\leq l\leq k$; cf. (\ref{eq:bABasymptoticstheta}). Combining (\ref{eq:asmetricgeneralcasereproduce}) and
(\ref{eq:chbmdeltaABestimate}) with arguments similar to the ones given at the beginning of Subsection~\ref{ssection:estrelspvariationmeancurvature}
below yields the conclusion that $\bvr+\ln\theta$ converges to zero; cf. (\ref{eq:varrholnthetaCz}) below. Since $\bvr-\varrho$ is constant, and since
(\ref{eq:lnthetapvarrhoasymptoticsintro}) holds, we conclude that $\bvr=\varrho-r_{\theta}$. Recalling that (\ref{eq:bDlhUphimPhiaPhibestimateintro})
holds, where the functions $\Phia$ and $\Phib$ are $C^{k}$, we conclude that 
\begin{equation}\label{eq:bDlhUphimPhiaPhibestimatereprodverprel}
\textstyle{\sum}_{l=0}^{k}|\bD^{l}[\hU(\phi)-\bPhi_{a}]|_{\bg_{\refer}}+
\textstyle{\sum}_{l=0}^{k}|\bD^{l}(\phi-\bPhi_{a}\bvr-\bPhi_{b})|_{\bg_{\refer}} \leq 
C_{k}\ldr{\ln\theta}^{a_{k}}\theta^{-2\e}
\end{equation}
on $M_{0}$, where $\bPhi_{a}:=\Phia$ and $\bPhi_{b}:=\Phib+r_{\theta}\Phia$. Next, $\msK$ is $C^{k}$ due to Lemma~\ref{lemma:kmKconvtomsK} below.
With these definitions, the data at the singularity are $(\bM,\bh,\msK,\bPhi_{a},\bPhi_{b})$. Finally, as in
Definition~\ref{def:developmentsfLambdacrushing}, we define the $C^{k}$-family $\chh$ of Riemannian metrics on $\bM$ by
\begin{equation}\label{eq:chhdeffinal}
  \chh:=\textstyle{\sum}_{A,B}\chb_{AB}\bmsY^{A}\otimes\bmsY^{B}.
\end{equation}

\begin{thm}\label{thm:reprod}
  Assume the $2$-standard assumptions, cf. Definition~\ref{def:kstandass}, to hold and (\ref{eq:bDlhUsqphiestimate}) to hold with $k=1$. Using
  the notation introduced prior to the statement of the theorem, assume that $\bPhi_{a}^{2}+\tr\msK^{2}=1$. Then, if $\bx\in\bM$ is such that
  $1+p_{A}(\bx)-p_{B}(\bx)-p_{C}(\bx)\leq 0$, there is a neighbourhood $V$ of $\bx$ such that $\bmsY^{A}([\bmsX_{B},\bmsX_{C}])=0$ on $V$.

  Assume, in addition, that the $k$-standard assumptions hold with $k\geq 2$ and that (\ref{eq:bDlhUsqphiestimate}) holds for this $k$.
  Define $(\bM,\bh,\msK,\bPhi_{a},\bPhi_{b})$ as described prior to the statement of the theorem, let $\mK$ be the expansion normalised
  Weingarten map and let $\chh$ be defined by (\ref{eq:chhdeffinal}). Then $\tr\msK=1$; $\msK$ is symmetric with respect to $\bh$;
  $\mathrm{tr}\msK^{2}+\bPhi_{a}^{2}=1$; $\mathrm{div}_{\bh}\msK=\bPhi_{a}d\bPhi_{b}$; the eigenvalues of $\msK$ are distinct; and
  $\msO_{BC}^{A}=0$ vanishes in a neighbourhood of $\bx$ if $1+p_{A}(\bx)-p_{B}(\bx)-p_{C}(\bx)\leq 0$. Moreover, there are constants $C_{k}$ and $a_{k}$
  and a $t_{0}\in (0,t_{+})$ such that (\ref{eq:lnthetapvarrhoasymptoticsintro}) and
  \begin{subequations}
    \begin{align}
      \textstyle{\sum}_{l=0}^{k}|\bD^{l}(\mK-\msK)|_{\bg_{\refer}}+\textstyle{\sum}_{l=0}^{k}|\bD^{l}(\chh-\bh)|_{\bg_{\refer}}
      \leq & C_{k}\ldr{\ln\theta}^{a_{k}}\theta^{-2\e},\label{eq:mKmsKconvergencehhbhconvergence}\\
      \textstyle{\sum}_{l=0}^{k}|\bD^{l}[\hU(\phi)-\bPhi_{a}]|_{\bg_{\refer}}
      +\textstyle{\sum}_{l=0}^{k}|\bD^{l}(\phi-\bPhi_{a}\bvr-\bPhi_{b})|_{\bg_{\refer}} \leq &
      C_{k}\ldr{\ln\theta}^{a_{k}}\theta^{-2\e},\label{eq:bDlhUphimPhiaPhibestimatereprodver}
    \end{align}
  \end{subequations}  
  hold on $M_{0}:=\bM\times (0,t_{0}]$. In particular, since (\ref{eq:lnthetapvarrhoasymptoticsintro}) holds, $\theta$ converges uniformly
  to $\infty$ as $t\rightarrow 0+$. 
\end{thm}
\begin{remark}
  The theorem demonstrates the necessity of the conditions appearing in Definition~\ref{def:ndsfidonbbs}. Note also that we deduce the form
  (\ref{eq:rinansatzintrocrushing}) of the metric and the convergence of $\chh$ to $\bh$. 
\end{remark}
\begin{proof}
  The proof is to be found at the beginning of Subsection~\ref{ssection:reproducingcondonsing}. 
\end{proof}

Finally, let us return to the original setting: a Gaussian foliation and Einstein's vacuum equations in $3+1$-dimensions; cf.
Subsection~\ref{ssection:fal}.

\begin{thm}\label{thm:finaloriginalsetting}
  Assume that the $k$-standard assumptions hold with $n=3$ and a $2\leq k\in\nn{}$. Assume, moreover, that $\tr\msK^{2}=1$, that $N=1$ and
  that the scalar field vanishes. Then $\mK$ converges to its limit $\msK$ in $C^{k}$ and the eigenvalues $p_{A}$ of $\msK$ satisfy $p_{1}<p_{2}<p_{3}$,
  $\sum_{i}p_{i}=1$ and $\sum_{i}p_{i}^{2}=1$. Fix $\bx\in\bM$. Assuming $\msK$ to be $C^{\infty}$, there are local coordinates $(V,\bsfx)$ in a neighbourhood
  of $\bx$ such that
  \[
  g =-dt\otimes dt+\textstyle{\sum}_{i,j}d_{ij}t^{2p_{\max\{i,j\}}}d\bsfx^{i}\otimes d\bsfx^{j}
  \]
  on $V\times (0,t_{+})$. Fix a subset $W$ of $\bM$ such that the closure of $W$ is compact and contained in $V$. There are functions $\bar{d}_{ij}$ on
  $V$, symmetric in $i$ and $j$, constants $C_{k}$ and $0<\eta\in\rn{}$ and a $t_{0}\in (0,t_{+})$ such that
  \begin{equation}\label{eq:dijlimits}
    \textstyle{\sum}_{l\leq k}|\bnabla^{l}(d_{ij}-\bd_{ij})|_{\bge_{\refer}}\leq C_{k}t^{2\eta}
  \end{equation}
  on $W\times (0,t_{0}]$, where $\bnabla$ is the Levi-Civita connection associated with $\bge_{\refer}$, $\bd_{ii}>0$ (no summation), 
  \begin{equation}\label{eq:bdijrelations}
    \bd_{12}=\frac{\bd_{22}}{p_{2}-p_{1}}\msK^{2}_{1},\ \ \
    \bd_{23}=\frac{\bd_{33}}{p_{3}-p_{2}}\msK^{3}_{2},\ \ \
    \bd_{13}=\frac{\bd_{33}}{p_{3}-p_{1}}\left(\msK^{3}_{1}+\frac{1}{p_{3}-p_{2}}\msK^{2}_{1}\msK^{3}_{2}\right).
  \end{equation}
  Here $\msK^{i}_{j}=\msK(d\bsfx^{i},\d_{j})$, $\msK^{i}_{i}=p_{i}$ (no summation) and $\msK^{i}_{j}=0$ if $i<j$. Moreover,
  \begin{subequations}\label{eq:bKijasesthmlUbKijasest}
    \begin{align}
      t\textstyle{\sum}_{l\leq k}|\bnabla^{l}(\bK^{i}_{j}-t^{-1}\msK^{i}_{j})|_{\bge_{\refer}} \leq & C_{k}\ldr{\ln t}^{a_{k}}t^{2\e}\min\{1,t^{2(p_{j}-p_{i})}\},
      \label{eq:bKtinvmsKestfalset}\\
      t^{2}\textstyle{\sum}_{l\leq k}|\bnabla^{l}\d_{t}(\bK^{i}_{j}-t^{-1}\msK^{i}_{j})|_{\bge_{\refer}} \leq & C_{k}\ldr{\ln t}^{a_{k}}t^{2\e}\min\{1,t^{2(p_{j}-p_{i})}\}
    \end{align}
  \end{subequations}
  on $W\times (0,t_{0}]$. Finally, (\ref{eq:falciicond}) holds on $V$ with $c_{ij}$ replaced by $\bd_{ij}$ and $\kappa_{i}^{\phantom{i}j}=-\msK_{i}^{\phantom{i}j}$.
\end{thm}
\begin{remark}
  It is of course not necessary to assume $\msK$ to be smooth. The reason we do so here is to avoid technicalities. The interested
  reader is encouraged to generalise Lemma~\ref{lemma:localcoordinates} below to the case that $\msK$ is $C^{k}$. This means that the coordinate
  system constructed in Lemma~\ref{lemma:localcoordinates} below has finite regularity, and one has to trace the consequences of this for
  the following estimates. 
\end{remark}
\begin{remark}
  The conclusions should be compared with the results of \cite{fal}, in particular, \cite[Theorem~1.1, pp.~1185--1186]{fal} and
  \cite[Theorem~1.7, p.~1187--1188]{fal}, see also Subsection~\ref{ssection:fal}. Note, in addition, that the assumption (\ref{eq:bDllnhNconditions}) in the
  current setting corresponds to the requirement that the relative spatial variation of $\theta$ does not grow faster than polynomially in $\varrho$.
  On the other hand, the conclusion (\ref{eq:bKtinvmsKestfalset}) implies that the relative spatial variation of $\theta$ decays exponentially in
  $\varrho$. 
\end{remark}
\begin{proof}
  The proof is to be found at the end of Subsection~\ref{ssection:reproducingcondonsing}. 
\end{proof}

\subsection{Conclusions and outlook}\label{ssection:concloutl}

Due to the results of this section, there are reasons to conjecture that initial data as in Definition~\ref{def:ndsfidonbbs}  parametrise the
stable manifolds of the corresponding subsets of the Kasner disc. The reason for this is that by assuming the solution to converge to a non-degenerate
subset of the Kasner disc, one reproduces the conditions of this definition; cf. Theorem~\ref{thm:reprod}. Moreover, at least in the $n=3$ dimensional
vacuum setting, the conditions on initial data on the singularity are such that one locally reproduces the conditions appearing in \cite{fal}. If one
could localise the existence and uniqueness results of \cite{fal}, one would also obtain existence and uniqueness of solutions corresponding to the data
on the singularity. Moreover, when condition 4 in Definition~\ref{def:ndsfidonbbs} is void and the data are real analytic, the results of
\cite{aarendall,damouretal} apply; cf. Subsection~\ref{ssection:stablequiescentregime}. 

Clearly, there are several open problems associated with the topics addressed in this article. First, existence and uniqueness of solutions corresponding
to initial data on the singularity should be demonstrated. This should be done using various gauges; cf. the
discussion in Subsection~\ref{ssection:oscibbsing} below. It would also be of interest to generalise the formulations of
this article to the case of a non-zero shift vector field. However, it might then be useful to specialise to gauges for which there
is a well-posed initial value problem. Moreover, in the context of such gauges, it would be of interest to deduce estimates for higher order
derivatives given information concerning a finite number of derivatives; in the context of the asymptotics of solutions to systems of non-linear
wave equations, control of a fixed finite number of derivatives usually yields control over all derivatives. The goal would of course be to
minimize the number of derivatives involved in Definitions~\ref{def:developmentsfLambda} and \ref{def:developmentsfLambdacrushing},
Theorems~\ref{thm:improvingasymptoticsgeneral} and \ref{thm:improvingasymptoticsgeneralG} etc. An additional issue is associated with more
general matter models. The matter models considered here are isotropic. However, it would be of interest to analyse what happens in the
case of anisotropic matter models such as Maxwell's equations and the Vlasov equation; cf., e.g., \cite{wea,lrt,cah10,cah11}. It is also
of interest to analyse how the subset of ordinary initial data corresponding to data on the singularity (i.e., conjecturally, the stable manifold
corresponding to the Kasner circle or the Kasner disc) sits in the set of all regular initial data: is this subset a submanifold? 

Finally, once a clearer picture of the above issues have been obtained, it is of interest to see if the dynamical systems arguments of, e.g.,
\cite{lea,beguin,lrt,du} can be upgraded to the absence of symmetries; cf. Subsection~\ref{ssection:oscibbsing} below for a further discussion. 

\subsection{Acknowledgements}
This research was funded by the Swedish Research Council, dnr. 2017-03863 and 2022-03053. 

\section{Applications}\label{section:applications}

In this section we provide two applications of the notion of initial data on a singularity introduced here. 

\subsection{A unified perspective on results in the quiescent setting}
One advantage of the ideas developed in this article is that they yield a unified perspective on many of the results in the literature. The
justification of this statement requires a fairly detailed discussion of a large number of different results, and is, for this reason, quite lengthy.
We therefore provide the relevant analysis elsewhere. Nevertheless, we here give one example: we calculate the initial data on the singularity
for Bianchi class A solutions to Einstein's vacuum equations (i.e., the maximal globally hyperbolic developments corresponding to left
invariant vacuum initial data on $3$-dimensional unimodular Lie groups). As discussed in greater detail in \cite{RinDoS}, the
relevant formulation of initial data on the singularity is then the following.
\begin{definition}\label{def:ndvacidonbbssh}
  Let $G$ be a $3$-dimensional unimodular Lie group, $\bh$ be a left invariant Riemannian metric on $G$ and $\msK$ be a left invariant
  $(1,1)$-tensor field on $G$. Then $(G,\bh,\msK)$ are \textit{non-degenerate quiescent Bianchi class A vacuum initial data on the
    singularity} if
  \begin{enumerate}
  \item $\tr\msK=1$ and $\msK$ is symmetric with respect to $\bh$,
  \item $\tr\msK^{2}=1$ and $\mathrm{div}_{\bh}\msK=0$,
  \item the eigenvalues of $\msK$ are distinct,
  \item $\msO^{1}_{23}\equiv 0$ on $\bM$, where (if $p_{1}<p_{2}<p_{3}$ are the eigenvalues of $\msK$, $\msX_{A}$ is an eigenvector field corresponding
    to $p_{A}$, normalised so that $|\msX_{A}|_{\bh}=1$, and $\g_{BC}^{A}$ is defined by $[\msX_{B},\msX_{C}]=\g^{A}_{BC}\msX_{A}$) the quantity $\msO^{A}_{BC}$
    is defined by $\msO^{A}_{BC}:=(\g^{A}_{BC})^{2}$.
  \end{enumerate}
\end{definition}
As argued in \cite{RinDoS}, the conditions of Definition~\ref{def:ndvacidonbbssh} yield the existence of an orthonormal (with respect to $\bh$) basis
$\{e_{i}\}$ of the Lie algebra $\mathfrak{g}$ of $G$, satisfying $\msK e_{i}=p_{i}e_{i}$ (no summation), with $p_{1}<p_{2}<p_{3}$, as well as 
$[e_{i},e_{j}]=\e_{ijk}n_{k}e_{k}$ (no summation), where $\e_{123}=0$ and $\e_{ijk}$ is antisymmetric under permutations of the indices. Here $p_{i}$ and
$n_{k}$ are constants. Moreover, the vanishing/non-vanishing and signs of the $n_{k}$ classify the Lie group under consideration: Bianchi type I
corresponds to $n_{k}=0$ for all $k$; Bianchi type II corresponds to all but one of the $n_{k}$ vanishing; Bianchi type VI${}_{0}$ corresponds to
one of the $n_{k}$ vanishing and the remaining ones having opposite sign; Bianchi type VII${}_{0}$ corresponds to one of the $n_{k}$ vanishing and
the remaining ones having the same sign; Bianchi type VIII corresponds to all the $n_{k}$ being different from zero and not having the same sign;
and Bianchi type IX corresponds to all the $n_{k}$ being different from zero and having the same sign. Since $\msO^{1}_{23}=0$ is, in this setting,
equivalent to $n_{1}=0$, it is clear that there are no Bianchi type VIII and no Bianchi type IX initial data in the sense of
Definition~\ref{def:ndvacidonbbssh}. This corresponds to the fact that, disregarding the locally rotationally symmetric solutions (which lead to
Cauchy horizons and extendibility of the maximal globally hyperbolic development, cf., e.g., \cite[Theorem~24.12, p.~258]{RinCauchy}), all
Bianchi type VIII and IX vacuum solutions have oscillatory singularities \cite{cbu}; i.e., they are not quiescent. However, for the remaining
Bianchi class A types, if the solution is not a quotient of a part of Minkowski space, it leads to data on the singularity in the sense of
Definition~\ref{def:ndvacidonbbssh} (with the caveat that in some exceptional cases, two $p_{i}$'s could coincide). Moreover, one obtains
asymptotics as in Definition~\ref{def:developmentvacuum}. We refer the interested reader to \cite{RinDoS} for a justification of these statements.
Finally, in \cite{RinDoS}, we demonstrate that given data as in Definition~\ref{def:ndvacidonbbssh}, there are unique corresponding developments.

To summarise, Definition~\ref{def:ndvacidonbbssh} can be used to parametrise quiescent Bianchi class A vacuum solutions. In the case of Bianchi
types VIII and IX, the absence of initial data in the sense of Definition~\ref{def:ndvacidonbbssh} leads to the expectation that solutions should
be oscillatory, an expectation which is borne out by \cite{cbu}. The above is only one very simple example of how previous results can be given a
unified understanding by the notion of initial data on the singularity introduced in this article. 

\subsection{Oscillatory big bang singularities}\label{ssection:oscibbsing}

A second, potential, application is to the understanding of oscillatory big bang singularities. The word ``oscillations'' here refers to the behaviour
of the eigenvalues, say $\ell_{A}$, of the expansion normalised Weingarten map, say $\mK$, along a foliation of a crushing singularity. A singularity
is crushing if there is a foliation of a neighbourhood of it such that the mean curvature $\theta$ of the leaves of the foliation
diverges to infinity uniformly in the direction of the singularity. In this setting, the expansion normalised Weingarten map is defined by
$\mK:=\bK/\theta$, where $\bK$ denotes the Weingarten map of the leaves of the foliation. In particular, $\tr\mK=1$. Note that $\mK$ is symmetric
with respect to the induced metric $\bge$, so that the $\ell_{A}$ are real. Since the sum of the $\ell_{A}$ equals $1$, it is, in $3+1$-dimensions,
convenient to introduce the notation
\begin{subequations}\label{eq:ellpmdef}
  \begin{align}
    \ell_{+} := & \frac{3}{2}\left(\ell_{2}+\ell_{3}-\frac{2}{3}\right)=\frac{3}{2}\left(\frac{1}{3}-\ell_{1}\right),\label{eq:ellplus}\\
    \ell_{-} := & \frac{\sqrt{3}}{2}(\ell_{2}-\ell_{3}).\label{eq:ellminus}
  \end{align}
\end{subequations}
Then the $\ell_{\pm}$ contain all the information concerning the eigenvalues of $\mK$, and, in $3+1$-dimensions the word ``oscillations'' refers to
oscillations in $\ell_{\pm}$. Moreover, ``oscillations'' should not here be understood as merely the absence of convergence, but rather a specific
type of oscillations illustrated in Figure~\ref{fig:TheKasnerMap} below.

The origin of the idea that big bang singularities should be oscillatory can be traced back to the physics literature, more particularly the
work of Belinski\v{\i}, Khalatnikov and Lifschitz (BKL); cf., e.g., \cite{LK,bkl1,bkl2}. However, the related ideas of Misner, published around the
same time, should also be mentioned, see \cite{misner}. The proposal by BKL, often referred to as the BKL conjecture, suggests that the generic
dynamics in the direction of a big bang singularity should localise in space, and, along a causal curve, be similar to that of an oscillatory
spatially homogeneous solution. In fact, Figure~\ref{fig:TheKasnerMap} should give a good description of the relevant dynamics. The interested
reader is referred to \cite{LK,bkl1,bkl2} for the original version of these ideas, and to \cite{dhn,dah,dadB,HUL,huar} for more recent refinements.
Even though the BKL proposal dates back more than 50 years, there are only results concerning oscillatory behaviour in the spatially homogeneous
setting; cf., e.g., \cite{cbu,wea,BianchiIXattr,HU,lea,beguin,lrt,brehm,du}. This is unsatisfactory, and can be contrasted with a growing number of
results concerning the future/past (or both) global non-linear stability of solutions to Einstein's equations; cf., e.g.,
\cite{heflambda,hefmink,hefYM,cak,zipser,aammilne,larcmp,rininv,rinpl,bieri,baz,larannals,aameflow,svedberg,speck,rintopstab,ial,lavk,rasJEMS,
speckradiation,haz,AAR,oliynyk,dal,havKdS,lat,havmink,aaf,kas,betal,fjs,fow,faw,dafetal}. In fact, there are even stability results in the direction
of the big bang singularity; cf. \cite{rasql,rasq,rsh,specks3,GIJ}. The latter results demonstrate, roughly speaking, that under the circumstances
discussed in Subsection~\ref{ssection:stablequiescentregime},
big bang formation is stable. The settings considered by the authors are the Einstein-scalar field equations, the Einstein-stiff fluid equations and
the Einstein vacuum equations in $n+1$-dimensions for $n\geq 10$. Moreover, the results concerning the asymptotics include the conclusion that
the singularities are quiescent; i.e., that the eigenvalues of $\mK$ converge. This may seem to contradict the BKL proposal. However, special
matter models and higher dimensions have long been expected to lead to quiescent behaviour. In the physics literature, this expectation appeared
in \cite{bkl15,Henneauxetal}. Two corresponding mathematical results are contained in \cite{aarendall,damouretal}; cf.
Subsection~\ref{ssection:stablequiescentregime}. Nevertheless, the stability results do not apply in
$3+1$-dimensions unless particular types of matter are present. For this reason, it is of interest to understand oscillatory singularities.

\begin{figure}
  \begin{center}
    \includegraphics{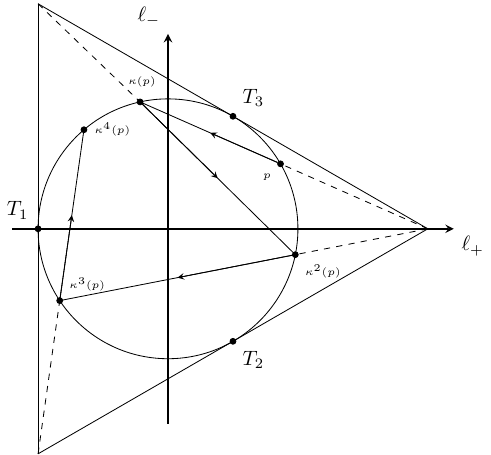}    
  \end{center}
  \caption{The Kasner circle and the BKL map, using the notation introduced in (\ref{eq:ellpmdef}). The circle
    represents the Kasner solutions; i.e., the geodesically incomplete maximal globally hyperbolic developments
    corresponding to left invariant vacuum initial data on $\rn{3}$. Note, in particular, that the solutions are parametrised
    by the eigenvalues of the expansion normalised Weingarten map. Since these solutions are invariant under permutations of the eigenvalues, there
    is a sixfold symmetry in this picture, meaning that only one sixth of the circle represents geometrically distinct solutions. In particular, it
    is, e.g., sufficient to focus on the segment between $Q_{1}$ (the point antipodal to $T_{1}$) and $T_{3}$. The points $T_{i}$ represent the flat
    Kasner solutions (the corresponding Riemann curvature tensor is identically zero) and the $T_{i}$ together with their antipodal points,
    denoted $Q_{i}$, represent the locally rotationally symmetric Kasner solutions. Given a point $p$ on the circle, the BKL map applied to $p$,
    denoted $\kappa(p)$, is obtained as follows: take the corner of the triangle closest to $p$; draw a straight line from this corner to $p$;
    continue the straight line to the next intersection with the circle; and define $\kappa(p)$ to be this second intersection. In the figure,
    we illustrate four iterations of $\kappa$ on a specific point. The chaoticity of the BKL map follows from, e.g.,
    \cite[Section~8, p.~22]{beguin}.}\label{fig:TheKasnerMap}
\end{figure}

Due to the above observations, it is clear that it is important to prove results in the spatially inhomogeneous and oscillatory setting. Given the
existing results, it is
therefore natural to ask if ideas similar to those used to prove stability, in particular the arguments of \cite{GIJ}, can be used. The answer
to this question is most probably no, and the main reason is that the background solutions considered in \cite{GIJ} are orbitally stable, whereas
the model dynamics in the oscillatory setting are chaotic; cf. Figure~\ref{fig:TheKasnerMap}. In particular, the model behaviour in the oscillatory
setting is by its very nature unstable. The arguments required to make progress on analysing this setting can thus be expected to be quite different.
Next, since there are, so far, only results in the spatially homogeneous setting, it is natural to ask if the corresponding arguments can be used.
Unfortunately, the ideas used in the first results in this setting are based on monotonicity principles and contradiction arguments; cf., e.g.,
\cite{cbu,wea,BianchiIXattr,HU,brehm}. In particular, global existence is guaranteed by soft methods. Moreover, contradiction arguments (combined
with, e.g., the monotonicity principle) are used to derive information concerning the asymptotics. In the spatially inhomogeneous setting, it is
unlikely to be possible to derive global existence without at the same time quantitatively controlling the asymptotic behaviour of solutions.
Moreover, to the extent that there are monotonic quantities, these quantities can not be expected to yield much information concerning the
asymptotics. To conclude, the ideas of \cite{cbu,wea,BianchiIXattr,HU,brehm} are unlikely to be useful when trying to obtain results in the
spatially inhomogeneous setting. On the other hand, more recently, results based on estimates have appeared; cf., e.g., \cite{lea,beguin,lrt,du}.
The idea of these articles is to take an orbit of the model system, cf. Figure~\ref{fig:TheKasnerMap}, and then to
prove that there is a stable manifold of solutions to the actual equations converging to this orbit. In \cite{lea,beguin,lrt}, the union of the
corresponding stable manifolds is conjectured to be non-generic. However, in \cite{du}, the author treats general enough orbits of the model
system that the union of the corresponding stable manifolds has positive Lebesgue measure in the relevant set of initial data.

In order to understand how the ideas of \cite{lea,beguin,lrt,du} could, potentially, be used in the spatially inhomogeneous setting, it is convenient
to focus on a specific situation. In what follows, we therefore discuss vacuum Bianchi type IX solutions, the goal being to give a rough idea of the
arguments in \cite{lea,du}. Note, to begin with, that generic vacuum Bianchi type IX solutions converge to the so-called Bianchi IX attractor, say
$\mA$; cf. \cite{BianchiIXattr}. The dynamics on the attractor are essentially described by Figure~\ref{fig:TheKasnerMap}. However, this figure
suppresses three dimensions. The attractor is, in fact, the union of three half ellipsoids in five dimensions that intersect in the Kasner circle.
In particular, an orbit that starts on the segment between $T_{i}$ and $T_{j}$ ($i\neq j$) travels along one of the ellipsoids (determined by $i$ and
$j$), and if $\{i,j\}\neq \{k,l\}$, where $k\neq l$, the ellipsoid corresponding to $\{i,j\}$ differs from the ellipsoid
corresponding to $\{k,l\}$. What is depicted in Figure~\ref{fig:TheKasnerMap} is the projection of the orbits in the ellipsoids to the
$\ell_{+}\ell_{-}$-plane. Fixing a point on the Kasner circle, say $p$, there are two orbits (along two separate ellipsoids) converging to
$p$. Moreover, there is one orbit starting at $p$ (along a third ellipsoid). In order to understand the dynamics, it is useful to pick one of the
orbits going into $p$, and then to take an appropriate hypersurface, say $S_{\roin,0}$, in the full state space (not only in the attractor), which is
transverse to this orbit (and, more generally, transverse to the vector field on the state space determined by Einstein's equations). Moreover,
one picks a hypersurface transverse to the orbit starting at $p$, say $S_{\roout,0}$. The model solution (on the attractor) travels from
$S_{\roin,0}\cap \mA$ to $S_{\roout,0}\cap \mA$ in infinite time. However, under suitable restrictions on the hypersurfaces, one can define a map from
$S_{\roin,0}$ to $S_{\roout,0}$, say $\Phi_{\roK,0}$. On the other hand, the orbit starting at $p$ turns into an orbit converging to $\kappa(p)$. As before,
one can define a suitable hypersurface $S_{\roin,1}$, transverse to the orbit starting at $p$ and converging to $\kappa(p)$, a corresponding
$S_{\roout,1}$ and so on. Next, one has to define a map from $S_{\roout,0}$ to $S_{\roin,1}$, say $\Psi_{\roK,0}$. The map $\Psi_{\roK,0}\circ\Phi_{\roK,0}$
then gives a map from $S_{\roin,0}$ to $S_{\roin,1}$. This can then be continued to $\Xi_{\roK,i}:=\Psi_{\roK,i}\circ\Phi_{\roK,i}$, $i=0,1,2,\dots$, etc. 
In the case of a periodic orbit of the model system, one can study a finite composition of maps of this type. In the non-periodic setting, one
has to study an infinite number of maps. However, when proving results, it is useful to group them into compositions such as the so-called
\textit{era return map} and the \textit{double era return map}; cf. \cite[Definitions~5.8 and 5.9, p.~61]{du}. The main point is that the double
era return map has nice hyperbolicity properties; it expands in the direction parallel to the Kasner circle and contracts in the directions
``perpendicular'' to the attractor (note that there is no canonical metric on the state space). This, in the end, allows one to prove the existence
of a stable manifold of solutions corresponding the the orbit one started with. The proof proceeds via a contraction mapping argument on, in the
non-periodic setting, an infinite collection of graphs (one for each point in the iteration of the double era return map). The graph at one of
these points roughly speaking maps the non-Kasner parameter variables in the relevant hypersurface, say $S_{\roin,i}$, to the Kasner parameter. 
However, to summarise, the main reason one can show that the mapping from the infinite collections of graphs to itself is a contraction is that 
the double era return map has nice hyperbolicity properties.

In the above description, we have omitted many of the complications. First, care has
to be taken so that the flow of the vector field maps $S_{\roin,0}$ into $S_{\roout,0}$; $S_{\roout,0}$ into $S_{\roin,1}$ etc. Then one has to derive detailed
estimates corresponding to, e.g., $\Xi_{\roK,0}$, extending to the attractor (on which the transition time is infinite) etc. It turns out that, for
an arbitrary orbit of the model solution, this will not be possible. Therefore one has to isolate the ``good'' orbits. It turns out that there
is a set of full measure on the Kasner circle corresponding to ``good'' orbits. Unfortunately, the complement is dense. We refer the interested
reader to \cite{lea,du} for more details. However, it is of interest to note that ``good'' orbits do not include degenerate points on the Kasner
circle (i.e., they do not include points for which two of the $p_{i}$ coincide). The reason for this is that the only degenerate points on the Kasner
circle are the $T_{i}$'s and their antipodal points; cf. Figure~\ref{fig:TheKasnerMap}. On the other hand, orbits of the model system which include
$T_{i}$ (or its antipodal point) terminate at $T_{i}$. In the case of Bianchi type VIII and IX, it can be demonstrated that there are no solutions
to the actual equations that terminate at the point $T_{i}$. However, there are solutions whose $(\ell_{+},\ell_{-})$-coordinates terminate on a $T_{i}$.
On the other hand, those solutions are locally rotationally symmetric (cf. \cite[Proposition~22.10, p.~237]{RinCauchy}) and have a Cauchy horizon
through which the solution can be extended (cf. \cite[Theorem~24.12, p.~258]{RinCauchy}). In particular, these solutions are non-generic and do not
reflect the behaviour of generic solutions. For these reasons, the stable manifold associated with degenerate points on the Kasner circle is not
expected to be of any greater interest. 

The proofs in \cite{lea,du} are highly non-trivial. Nevertheless, the arguments are substantially simplified not only by the fact that the authors
consider spatially homogeneous solutions, but by the fact that they consider Bianchi class A solutions (maximal globally hyperbolic developments
corresponding to left invariant initial data on $3$-dimensional unimodular Lie groups). The reason for the simplification is that in the unimodular
setting, the stable and unstable manifolds associated with the Kasner circle correspond to the vanishing/sign of basic variables that, at the
same time, classify the Lie group under consideration; cf. \cite[Section~3.4]{RinGeo} for a more detailed discussion. Already the Bianchi class
B (i.e., the non-unimodular) setting is more complicated, but it is reasonable to expect the
spatially inhomogeneous setting to represent a different level of difficulty. This is partly due to the problem of gauge invariance;
there is a large number of choices of gauge (including choices of foliation). Moreover, instead of the simple dichotomy of being on the stable
manifold or off (as in the Bianchi class A setting), spatial variations can be expected to give rise to complications such as spikes; cf. \cite{raw}
for a discussion of this topic in the context of $\tn{3}$-Gowdy vacuum spacetimes. Causal localisation might also be necessary in order to deal with
the substantial spatial variations that are to be expected. However, such a localisation would limit the allowed gauge choices. Since the causal structure
is not known a priori, causal localisation also makes it necessary to keep track of the causal structure ``along an orbit''. And so on.
Finally, it is not clear that it would be possible to carry out ideas analogous to \cite{lea,du} in the general spatially inhomogeneous setting.
Nevertheless, at present, this approach seems to be the most promising. 

If one wants to generalise the ideas of \cite{lea,du} to the spatially inhomogeneous setting, the natural first step is to identify the stable
manifold corresponding to the Kasner circle in the absence of symmetries. However, as pointed out in Subsection~\ref{ssection:concloutl}, it is
reasonable to conjecture that Definition~\ref{def:vacuumidos} parametrises this stable manifold. 

\section{Relating different notions of initial data, the $3$-dimensional vacuum setting}

The purpose of the present section is to relate the two notions of initial data on the singularity given by Definitions~\ref{def:falid} and
\ref{def:vacuumidos}. We start by, given initial data in the sense of Definition~\ref{def:vacuumidos}, constructing appropriate local coordinates.

\begin{lemma}\label{lemma:localcoordinates}
  Let $\bM$ be a $3$-dimensional manifold and $\msK$ be a smooth $(1,1)$-tensor field on $\bM$. Assume $\msK$ to have distinct real eigenvalues
  $p_{1}<p_{2}<p_{3}$. Let $\bmsX_{A}$ be an eigenvector field of $\msK$ corresponding to the eigenvalue $p_{A}$ and let $\g^{A}_{BC}$ be defined
  by $[\bmsX_{B},\bmsX_{C}]=\g_{BC}^{A}\bmsX_{A}$. Assume that $\g^{1}_{23}=0$. Then, if $p\in\bM$, there are local coordinates $(V,\bsfx)$ with $p\in V$,
  such that, after a renormalisation of the $\bmsX_{A}$, if necessary,
  \begin{equation}\label{eq:preferredcoordinates}
    \bmsX_{3}=\d_{3},\ \ \
    \bmsX_{2}=\d_{2}+\bmsX_{2}^{3}\d_{3},\ \ \
    \bmsX_{1}=\d_{1}+\bmsX_{1}^{2}\d_{2}+\bmsX_{1}^{3}\d_{3}
  \end{equation}
  for a suitable choice of smooth functions $\bmsX_{1}^{2}$, $\bmsX_{2}^{3}$, $\bmsX_{1}^{3}$. Moreover, if $\{\bmsY^{A}\}$ is the frame dual to
  $\{\bmsX_{A}\}$, then
  \begin{equation}\label{eq:msYAformulae}
    \bmsY^{1}=d\bsfx^{1},\ \ \
    \bmsY^{2}=d\bsfx^{2}-\bmsX_{1}^{2}d\bsfx^{1},\ \ \
    \bmsY^{3}=d\bsfx^{3}-\bmsX_{2}^{3}d\bsfx^{2}+(\bmsX_{1}^{2}\bmsX_{2}^{3}-\bmsX_{1}^{3})d\bsfx^{1}.
  \end{equation}
  Define $\msK_{i}^{j}$ by the condition that $\msK\d_{i}=\msK_{i}^{j}\d_{j}$. Then $\msK_{i}^{j}=0$ if $j<i$; $\msK_{i}^{i}=p_{i}$ (no summation); and
  \begin{subequations}\label{eq:bmsXijdef}
    \begin{align}
      \bmsX_{1}^{2} = & \frac{1}{p_{1}-p_{2}}\msK_{1}^{2},\label{eq:Xot}\\
      \bmsX_{2}^{3} = & \frac{1}{p_{2}-p_{3}}\msK_{2}^{3},\label{eq:Xtth}\\
      \bmsX_{1}^{3} = & \frac{1}{p_{1}-p_{3}}\msK_{1}^{3}+\frac{1}{(p_{1}-p_{3})(p_{1}-p_{2})}\msK_{1}^{2}\msK_{2}^{3}.\label{eq:Xoth}
    \end{align}
  \end{subequations}  
\end{lemma}
\begin{remark}
  In \cite[Theorem~1.1, pp.~1185--1186]{fal}, a quantity $\kappa_{i}^{\phantom{i}j}$ appears. Due to the differences in convention,
  $\msK_{i}^{j}=-\kappa_{i}^{\phantom{i}j}$.
\end{remark}
\begin{proof}
  To begin with, there are local coordinates around $p$ such that $\bmsX_{3}=\d_{3}$; cf., e.g., \cite[Theorem~9.22, p.~220]{Lee}. By assumption
  \[
  [\d_{3},\bmsX_{2}]=A\d_{3}+B\bmsX_{2}.
  \]
  In order to find a vector field which commutes with $\d_{3}$, and which is, at the same time, in the span of $\d_{3}$ and $\bmsX_{2}$, let
  $X_{2}=\alpha\bmsX_{2}+\beta\d_{3}$. Then
  \[
  [\d_{3},X_{2}]=(\d_{3}\alpha)\bmsX_{2}+\alpha [\d_{3},\bmsX_{2}]+(\d_{3}\beta)\d_{3}=(\d_{3}\alpha+B\alpha)\bmsX_{2}+(\d_{3}\beta+A\alpha)\d_{3}.
  \]
  Assume that $p$ corresponds to the origin in the local coordinates, and fix $\alpha=1$ on the plane defined by $x^{3}=0$. Then we can calculate
  $\alpha$ by solving the equation $\d_{3}\alpha+B\alpha=0$. Moreover, $\alpha$ will be smooth and strictly positive in a neighbourhood of the origin.
  Next, we solve for $\beta$ by choosing $\b=0$ on the plane $x^{3}=0$ and integrating $\d_{3}\beta+A\alpha=0$. Doing so yields a vector field $X_{2}$
  with the property that the span
  of $X_{2}$ and $\d_{3}$ equals the span of $\bmsX_{2}$ and $\bmsX_{3}$ in a neighbourhood of $p$. Moreover, $\d_{3}$ and $X_{2}$ commute. Due
  to \cite[Theorem~9.46, p.~234]{Lee}, we conclude that there are local coordinates in a neighbourhood of $p$ such that $\bmsX_{3}=\d_{3}$
  and $X_{2}=\d_{2}$. Since $\bmsX_{2}$ is in the span of $\bmsX_{3}$ and $X_{2}$, it is clear that $\bmsX_{2}=a\d_{2}+b\d_{3}$. Since $\bmsX_{2}$
  and $\bmsX_{3}$ are linearly independent, it is clear that $a$ is never allowed to vanish. In particular, after renormalising the
  eigenvector field $\bmsX_{2}$, if neccessary, we can assume that $\bmsX_{2}=\d_{2}+\bmsX_{2}^{3}\d_{3}$. That $\bmsX_{1}$ can be assumed to
  take the form $\bmsX_{1}=\d_{1}+\bmsX_{1}^{2}\d_{2}+\bmsX_{1}^{3}\d_{3}$ follows by a similar argument. 

  Finally (\ref{eq:msYAformulae}) follows from straightforward calculations and (\ref{eq:bmsXijdef}) and the statements concerning $\msK_{i}^{j}$
  follow from the fact that $\bmsX_{A}$ is an eigenvector field of $\msK$ with eigenvalue $p_{A}$.
\end{proof}

\subsection{Relating the limits}
In Proposition~\ref{prop:geometrictofal}, the goal is to, locally, associate data in the sense of Definition~\ref{def:falid} with initial data
in the sense of Definition~\ref{def:vacuumidos}. In the present section we provide the intuition behind the definition (\ref{eq:cotetcdef}). We
do so by assuming that we have a development in the sense of Definition~\ref{def:developmentvacuum}. We then prove that this development can locally
be written in the form (\ref{eq:gfalasformofmetric}). Finally, we calculate the limit (\ref{eq:aijlimcijfal}).

\begin{lemma}
  Let $(\bM,\bh,\msK)$ be non-degenerate quiescent vacuum initial data on the singularity in the sense of Definition~\ref{def:vacuumidos}. Assume
  that there is a locally Gaussian development $(M,g)$ of the initial data in the sense of Definition~\ref{def:developmentvacuum}. Fix a
  $p\in\bM$ and let $(V,\bsfx)$ be local coordinates with the properties stated in Lemma~\ref{lemma:localcoordinates}. Then, on $V\times (0,t_{+})$,
  the development (\ref{eq:rinansatzintroPsi}) can be written
  \[
  \Psi^{*}g=-dt\otimes dt+\textstyle{\sum}_{i,j=1}^{3}a_{ij}t^{2p_{\max\{i,j\}}}d\bsfx^{i}\otimes d\bsfx^{j},
  \]
  using the notation of Definition~\ref{def:developmentvacuum}.
  Moreover, (\ref{eq:aijlimcijfal}) is satisfied, where $c_{ii}=\bh(\bmsX_{i},\bmsX_{i})$ (no summation) and the $\bmsX_{i}$ are given in the
  statement of Lemma~\ref{lemma:localcoordinates}; the $c_{ij}$, $i<j$, are given by (\ref{eq:cotetcdef}), where the sub- and superscripts
  refer to the coordinate frame associated with $(V,\bsfx)$; and $c_{ij}=c_{ji}$. Moreover, if $W$ has compact closure contained in $V$, there
  is an $\eta>0$, a $t_{0}\in(0,t_{+})$ and, for every $l\in\nn{}$, a $C_{l}$ such that
  \begin{equation}\label{eq:aijtocijatrate}
    \|a_{ij}(\cdot,t)-c_{ij}\|_{C^{l}(W)}\leq C_{l}t^{2\eta}
  \end{equation}
  on $(0,t_{0}]$. 

  Define the $p_{i}$ and $c_{ij}$ as above and $\kappa_{i}^{\phantom{i}j}=-\msK_{i}^{\phantom{i}j}$.
  Then the $p_{i}$ and $c_{ij}$ are smooth real valued functions on $V$ and the $p_{i}$ satisfy (\ref{eq:Kasnerrelationsfal}). Moreover,
  $\kappa_{i}^{\phantom{i}i}=-p_{i}$ (no summation); $\kappa_{i}^{\phantom{i}l}=0$ if $l<i$; and (\ref{eq:kappacijrelfal}) is satisfied. 
\end{lemma}
\begin{remark}\label{remark:geotofalexcmom}
  Arguments similar to the proof of the lemma provide a proof of all the statements of Proposition~\ref{prop:geometrictofal} except for the
  statement that (\ref{eq:falciicond}) holds. They also demonstrate that $\kappa_{i}^{\phantom{i}j}=-\msK_{i}^{\phantom{i}j}$. 
\end{remark}
\begin{proof}
  Let $\bmsX_{A}$ be the vector fields defined by (\ref{eq:preferredcoordinates}). Note that $\msK\bmsX_{A}=p_{A}\bmsX_{A}$ (no summation). Since
  $\msK$ is symmetric with respect to $\bh$, the $\msX_{A}$ are orthogonal with respect to $\bh$. In particular, there are thus smooth functions
  $\b_{A}>0$ such that
  \begin{equation}\label{eq:bhwrtframe}
    \bh=\textstyle{\sum}_{A}\b_{A}^{2}\bmsY^{A}\otimes \bmsY^{A}.
  \end{equation}
  Note that $\b_{A}^{2}=\bh(\bmsX_{A},\bmsX_{A})$ (no summation). Combining this observation with (\ref{eq:limitbABtobhchh}) yields the conclusion that,
  for $A\neq B$, $b_{AB}\rightarrow 0$ as $t\rightarrow 0+$, where the convergence is uniform on compact subsets of $V$. Moreover,
  $b_{AA}\rightarrow \b_{A}^{2}$ as $t\rightarrow 0+$ (no summation), where the convergence is uniform on compact subsets of $V$.

  Next, we wish to rewrite the metric in terms of the coordinate basis $d\bsfx^{i}\otimes d\bsfx^{j}$, where $(V,\bsfx)$ are the coordinates constructed in
  Lemma~\ref{lemma:localcoordinates}. Note, to this end, that
  \[
  \textstyle{\sum}_{A,B}b_{AB}t^{2p_{\max\{A,B\}}}\bmsY^{A}\otimes \bmsY^{B}
  =\textstyle{\sum}_{i,j}a_{ij}t^{2p_{\max\{i,j\}}}d\bsfx^{i}\otimes d\bsfx^{j},
  \]
  where
  \begin{equation}\label{eq:aijintermsofbAB}
    a_{ij}=\textstyle{\sum}_{A,B}b_{AB}t^{2p_{\max\{A,B\}}-2p_{\max\{i,j\}}}\bmsY^{A}_{i}\bmsY^{B}_{j}.
  \end{equation}
  Here, due to (\ref{eq:msYAformulae}), $\bmsY_{i}^{A}=0$ if $A<i$, $\bmsY_{i}^{A}=1$ if $A=i$ and
  \[
  \bmsY^{2}_{1}=-\bmsX_{1}^{2},\ \ \
  \bmsY^{3}_{2}=-\bmsX_{2}^{3},\ \ \
  \bmsY^{3}_{1}=\bmsX_{1}^{2}\bmsX_{2}^{3}-\bmsX_{1}^{3}.
  \]
  Consider $a_{ii}$ (no summation). If $A>i$ or $B>i$, then the corresponding term in (\ref{eq:aijintermsofbAB}) converges to zero, since
  the $b_{AB}$ converge and $2p_{\max\{A,B\}}-2p_{\max\{i,j\}}>0$ in that case. The only term we need to take into consideration is thus the one
  with $i=j$ and $A=B=i$. This term reads $b_{ii}$ (no summation) and, by the above, it converges to $\b_{i}^{2}$. To conclude, $a_{ii}$ converges
  to $\b_{i}^{2}$, and in order for the notation to be consistent with that of \cite{fal}, we introduce $c_{ii}=\b_{i}^{2}$ (no summation).
  In particular $c_{ii}$ is thus determined by $\bh$ and our choice of frame; cf. (\ref{eq:bhwrtframe}). 

  Next, consider $a_{12}$. If $A$ or $B$ equals $3$, then the corresponding term in (\ref{eq:aijintermsofbAB}) converges to zero. If $A=B=1$,
  the corresponding term equals zero, since $\bmsY_{2}^{1}=0$. The only possibilities for $(A,B)$ that remain to be considered are $(1,2)$, $(2,1)$
  and $(2,2)$. However, since $b_{12}$ converges to zero, it is sufficient to consider
  \[
  b_{22}\bmsY^{2}_{2}\bmsY^{2}_{1}\rightarrow -\bmsX_{1}^{2}c_{22},
  \]
  where we used the fact that $b_{22}$ converges to $c_{22}$, the fact that $\bmsY^{2}_{2}=1$ and the fact that $\bmsY^{2}_{1}=-\bmsX_{1}^{2}$. Again, to be
  consistent with \cite{fal}, we use the notation $c_{12}:=-\bmsX_{1}^{2}c_{22}$ (and note that $a_{12}$ converges to $c_{12}$). Combining this relation with
  (\ref{eq:Xot}) yields
  \begin{equation}\label{eq:msKotcijformula}
    \msK_{1}^{2}=(p_{2}-p_{1})\frac{c_{12}}{c_{22}}.
  \end{equation}
  Letting $\kappa_{i}^{\phantom{i}j}=-\msK_{i}^{\phantom{i}j}$, (\ref{eq:msKotcijformula}) contains the same information as the first equality in
  (\ref{eq:kappacijrelfal}). Next, consider $a_{23}$. Then $B$ has to equal $3$. Moreover, if $A<3$, the corresponding term converges to zero since
  $b_{AB}\rightarrow 0$ for $A\neq B$. The only term we need to consider is thus
  \[
  b_{33}\bmsY_{3}^{3}\bmsY_{2}^{3}=-\bmsX_{2}^{3}b_{33}\rightarrow -\bmsX_{2}^{3}c_{33}.
  \]
  In particular, $a_{23}$ converges to $c_{23}:=-\bmsX_{2}^{3}c_{33}$. Keeping (\ref{eq:Xtth}) in mind, we obtain
  \begin{equation}\label{eq:msKtthcijformula}
    \msK_{2}^{3}=(p_{3}-p_{2})\frac{c_{23}}{c_{33}};
  \end{equation}
  cf. the second equality in (\ref{eq:kappacijrelfal}). Finally, consider $a_{13}$. For the same reasons as above, the only term we have to take into
  consideration is
  \[
  b_{33}\bmsY_{1}^{3}\bmsY_{3}^{3}=b_{33}(\bmsX_{1}^{2}\bmsX_{2}^{3}-\bmsX_{1}^{3})\rightarrow c_{33}(\bmsX_{1}^{2}\bmsX_{2}^{3}-\bmsX_{1}^{3})=:c_{13}.
  \]
  Appealing to (\ref{eq:bmsXijdef}) as well as (\ref{eq:msKotcijformula}) and (\ref{eq:msKtthcijformula}) yields the conclusion that
  \begin{equation}\label{eq:msKothcijformula}
    \msK_{1}^{3}=(p_{3}-p_{1})\frac{c_{13}}{c_{33}}-(p_{2}-p_{1})\frac{c_{12}c_{23}}{c_{22}c_{33}};
  \end{equation}
  cf. the third equality in (\ref{eq:kappacijrelfal}). Due to $c_{ii}=\b_{i}^{2}$, (\ref{eq:msKotcijformula}), (\ref{eq:msKtthcijformula}), and
  (\ref{eq:msKothcijformula}), it is clear that if the $c_{ij}$ are defined by (\ref{eq:cotetcdef}) for $i<j$, then (\ref{eq:aijlimcijfal}) holds. In
  fact, if $W$ has compact closure contained in $V$, (\ref{eq:aijtocijatrate}) holds.

  The final statements of the lemma follow from the fact that $\tr\msK=\tr\msK^{2}=1$, the conclusions of Lemma~\ref{lemma:localcoordinates}
  and the above observations. 
\end{proof}

\subsection{The asymptotic momentum constraint}\label{ssection:specasdata}
Due to Remark~\ref{remark:geotofalexcmom}, what remains (in order to prove Proposition~\ref{prop:geometrictofal}) is to prove that
(\ref{eq:falciicond}) holds. 

\begin{lemma}\label{lemma:momcondivmsKz}  
  Let $(\bM,\bh)$ be a smooth $3$-dimensional Riemannian manifold and $\msK$ be a smooth $(1,1)$-tensor field on $\bM$ which is symmetric with respect to
  $\bh$. Assume $\msK$ to have distinct eigenvalues $p_{1}<p_{2}<p_{3}$. Assume, moreover, that there are local coordinates $(V,\bsfx)$ such that for
  $A=1,2,3$, there are eigenvector fields $\bmsX_{A}$, $A=1,2,3$, of $\msK$ of the form (\ref{eq:preferredcoordinates}) for a suitable choice of smooth
  functions $\bmsX_{1}^{2}$, $\bmsX_{2}^{3}$, $\bmsX_{1}^{3}$. Assume also that $\msK\bmsX_{A}=p_{A}\bmsX_{A}$ (no summation) for $A=1,2,3$. Define $c_{ij}$
  as in the statement of Proposition~\ref{prop:geometrictofal}. Then (\ref{eq:falciicond}) (where $\kappa_{i}^{\phantom{i}j}$ is defined in terms of
  $c_{ij}$ and $p_{i}$ as described in Definition~\ref{def:falid}) is equivalent to $\mathrm{div}_{\bh}\msK=0$. 
\end{lemma}
%\begin{remark}\label{remark:momcondivmsKz}
%  By the same argument, it follows that if $\msK$ and $\bh$ are defined as in Subsection~\ref{ssection:geometricformfv}, then (\ref{eq:falciicond})
%  is equivalent to $\rodiv_{\bh}\msK=0$.
%\end{remark}
\begin{proof}
  Define $\{\bmsY^{A}\}$ via (\ref{eq:msYAformulae}). Since $\msK$ is symmetric with respect to $\bh$, it follows that $\{\bmsX_{A}\}$ is an orthogonal
  frame with respect to $\bh$. Thus
  \[
  \bh=\textstyle{\sum}_{A}c_{AA}\bmsY^{A}\otimes\bmsY^{A},
  \]
  where $c_{AA}:=\bh(\bmsX_{A},\bmsX_{A})$ (no summation); cf. Proposition~\ref{prop:geometrictofal}. Next, note that
  \begin{equation*}
    \begin{split}
      \bD_{A}\msK^{A}_{B} = & \bmsX_{A}(\msK^{A}_{B})-(\bD_{\bmsX_{A}}\bmsY^{A})(\msK\bmsX_{B})-\bmsY^{A}(\msK\bD_{\bmsX_{A}}\bmsX_{B})\\
      = & \bmsX_{A}(p_{B}\delta^{A}_{B})+\bmsY^{A}(\bD_{\bmsX_{A}}\bmsX_{C})\bmsY^{C}(\msK\bmsX_{B})-\bmsY^{A}(\msK\bD_{\bmsX_{A}}\bmsX_{B})\\
      = & \bmsX_{B}(p_{B})+p_{B}\bmsY^{A}(\bD_{\bmsX_{A}}\bmsX_{B})-\bmsY^{A}(\msK\bD_{\bmsX_{A}}\bmsX_{B})
    \end{split}
  \end{equation*}
  (no summation on $B$), where $\bD$ denotes the Levi-Civita connection induced by $\bh$. Introducing the notation $\Gamma^{A}_{BC}$ by the
  requirement that
  \[
  \bD_{\bmsX_{A}}\bmsX_{B}=\Gamma_{AB}^{C}\bmsX_{C},
  \]
  this equality can be written (again, no summation on $B$)
  \[
  \bD_{A}\msK^{A}_{B}=\bmsX_{B}(p_{B})+\textstyle{\sum}_{A}(p_{B}-p_{A})\Gamma_{AB}^{A}.
  \]
  What remains is thus to compute (no summation on $A$)
  \[
  \Gamma_{AB}^{A}=c_{AA}^{-1}\bh(\bD_{\bmsX_{A}}\bmsX_{B},\bmsX_{A})=
  c_{AA}^{-1}\bh([\bmsX_{A},\bmsX_{B}]+\bD_{\bmsX_{B}}\bmsX_{A},\bmsX_{A})=\frac{1}{2}\bmsX_{B}(\ln c_{AA})+\g^{A}_{AB},
  \]
  where $\g^{A}_{BC}$ is defined by (\ref{eq:gABCdef}). Next, we need to calculate $\g^{A}_{AB}$ (no summation). Keeping (\ref{eq:preferredcoordinates})
  and (\ref{eq:msYAformulae}) in mind, it can be computed that the only non-zero $\g^{A}_{AB}$ (no summation) are
  \begin{align*}
    \g^{2}_{21} = & \d_{2}\bmsX_{1}^{2}+\bmsX_{2}^{3}\d_{3}\bmsX_{1}^{2},\\
    \g^{3}_{31} = & \d_{3}\bmsX_{1}^{3}-\bmsX_{2}^{3}\d_{3}\bmsX_{1}^{2},\\
    \g^{3}_{32} = & \d_{3}\bmsX_{2}^{3}.
  \end{align*}
  Combining the above observations, we can compute that
  \[
  2\bD_{A}\msK^{A}_{3}=2\d_{3}p_{3}+\textstyle{\sum}_{A}(p_{3}-p_{A})\d_{3}\ln c_{AA}.
  \]
  In particular $-2\bD_{A}\msK^{A}_{3}$ equals the left hand side of (\ref{eq:falciicond}) in case $i=3$. In particular, $\bD_{A}\msK^{A}_{3}=0$
  is equivalent to (\ref{eq:falciicond}) in case $i=3$. Next, compute
  \begin{equation*}
    \begin{split}
      2\bD_{A}\msK^{A}_{2} = & 2\bmsX_{2}(p_{2})+\textstyle{\sum}_{A}(p_{2}-p_{A})\bmsX_{2}(\ln c_{AA})+2(p_{2}-p_{3})\d_{3}\bmsX_{2}^{3}\\
      = & 2\d_{2}p_{2}+2\bmsX_{2}^{3}\d_{3}p_{2}+\textstyle{\sum}_{A}(p_{2}-p_{A})\d_{2}(\ln c_{AA})\\
      & +\bmsX_{2}^{3}\textstyle{\sum}_{A}(p_{2}-p_{A})\d_{3}(\ln c_{AA})+2(p_{2}-p_{3})\d_{3}\bmsX_{2}^{3}.
    \end{split}
  \end{equation*}
  At this stage, it is of interest to note that
  \begin{equation*}
    \begin{split}
      & 2\bmsX_{2}^{3}\d_{3}p_{2}+\bmsX_{2}^{3}\textstyle{\sum}_{A}(p_{2}-p_{A})\d_{3}(\ln c_{AA})+2(p_{2}-p_{3})\d_{3}\bmsX_{2}^{3}\\
      = & \bmsX_{2}^{3}\left[2\d_{3}p_{3}+\textstyle{\sum}_{A}(p_{3}-p_{A})\d_{3}(\ln c_{AA})\right]+2\bmsX_{2}^{3}\d_{3}(p_{2}-p_{3})\\
      & +\bmsX_{2}^{3}(p_{2}-p_{3})\textstyle{\sum}_{A}\d_{3}(\ln c_{AA})+2(p_{2}-p_{3})\d_{3}\bmsX_{2}^{3}\\
      = & 2\bmsX_{2}^{3}\bD_{A}\msK^{A}_{3}+\msK_{2}^{3}\textstyle{\sum}_{A}\d_{3}(\ln c_{AA})+2\d_{3}\msK_{2}^{3},
    \end{split}
  \end{equation*}
  where we appealed to (\ref{eq:bmsXijdef}) in the last step. Summing up,
  \begin{equation*}
    \begin{split}
      & 2\bD_{A}\msK^{A}_{2}-2\bmsX_{2}^{3}\bD_{A}\msK^{A}_{3}\\
      = & \textstyle{\sum}_{A}(p_{2}-p_{A})\d_{2}(\ln c_{AA})+2\d_{2}p_{2}+2\d_{3}\msK_{2}^{3}+\msK_{2}^{3}\textstyle{\sum}_{A}\d_{3}(\ln c_{AA}).
    \end{split}
  \end{equation*}
  In particular, minus the right hand side of this equation equals the left hand side of (\ref{eq:falciicond}) in case $i=2$. 

  Finally, let us consider
  \begin{equation*}
    \begin{split}
      & 2\bD_{A}\msK^{A}_{1}-2\bmsX_{1}^{2}\bD_{A}\msK^{A}_{2}-2\bmsX_{1}^{3}\bD_{A}\msK^{A}_{3}+2\bmsX_{1}^{2}\bmsX_{2}^{3}\bD_{A}\msK^{A}_{3}\\
      = & 2\bmsX_{1}(p_{1})+\textstyle{\sum}_{A}(p_{1}-p_{A})\bmsX_{1}(\ln c_{AA})+2(p_{1}-p_{2})\d_{2}\bmsX_{1}^{2}\\
      & +2(p_{1}-p_{3})\d_{3}\bmsX_{1}^{3}-2\msK_{2}^{3}\d_{3}\bmsX_{1}^{2}\\
      & -2\bmsX_{1}^{2}\bmsX_{2}(p_{2})-\textstyle{\sum}_{A}(p_{2}-p_{A})\bmsX_{1}^{2}\bmsX_{2}(\ln c_{AA})-2(p_{2}-p_{3})\bmsX_{1}^{2}\d_{3}\bmsX_{2}^{3}\\
      & -2(\bmsX_{1}^{3}-\bmsX_{1}^{2}\bmsX_{2}^{3})\d_{3}p_{3}-\textstyle{\sum}_{A}(p_{3}-p_{A})(\bmsX_{1}^{3}-\bmsX_{1}^{2}\bmsX_{2}^{3})\d_{3}\ln c_{AA}.
    \end{split}
  \end{equation*}
  Collecting the terms on the right hand side involving sums over $A$ yields
  \begin{equation*}
    \begin{split}
      & \textstyle{\sum}_{A}(p_{1}-p_{A})\d_{1}(\ln c_{AA})+\msK_{1}^{2}\textstyle{\sum}_{A}\d_{2}(\ln c_{AA})+\msK_{1}^{3}\textstyle{\sum}_{A}\d_{3}(\ln c_{AA}).
    \end{split}
  \end{equation*}
  The remaining terms sum up to
  \[
  2\d_{1}p_{1}+2\d_{2}\msK_{1}^{2}+2\d_{3}\msK_{1}^{3}.
  \]
  To conclude
  \begin{equation*}
    \begin{split}
      & 2\bD_{A}\msK^{A}_{1}-2\bmsX_{1}^{2}\bD_{A}\msK^{A}_{2}-2\bmsX_{1}^{3}\bD_{A}\msK^{A}_{3}+2\bmsX_{1}^{2}\bmsX_{2}^{3}\bD_{A}\msK^{A}_{3}\\
      = & \textstyle{\sum}_{A}(p_{1}-p_{A})\d_{1}(\ln c_{AA})+2\d_{A}\msK^{A}_{1}
      +\msK_{1}^{2}\textstyle{\sum}_{A}\d_{2}(\ln c_{AA})+\msK_{1}^{3}\textstyle{\sum}_{A}\d_{3}(\ln c_{AA}).
    \end{split}
  \end{equation*}
  Moreover, the right hand side equals minus the left hand side of (\ref{eq:falciicond}) in case $i=1$. The statement follows.
\end{proof}

Finally, we prove Proposition~\ref{prop:geometrictofal}.

\begin{proof}[Proposition~\ref{prop:geometrictofal}]
  The conclusion of Proposition~\ref{prop:geometrictofal} is an immediate consequence of Remark~\ref{remark:geotofalexcmom} and
  Lemma~\ref{lemma:momcondivmsKz}. 
\end{proof}

\subsection{Deriving geometric asymptotics}
In Subsection~\ref{ssection:specasdata}, we verified that, locally, initial data in the sense of Definition~\ref{def:vacuumidos} are initial
data in the sense of Definition~\ref{def:falid}. Next, we wish to verify that if we have a development in the sense of \cite{fal}, we obtain
a development in the sense of Definition~\ref{def:developmentvacuum}. 

\begin{prop}\label{prop:dergeoas}
  Let $(\bM,\bh)$ be a smooth $3$-dimensional Riemannian manifold and $\msK$ be a smooth $(1,1)$-tensor field on $\bM$ which is symmetric with respect to
  $\bh$. Assume $\msK$ to have distinct eigenvalues $p_{1}<p_{2}<p_{3}$. Assume that there are local coordinates $(V,\bsfx)$ such that for
  $A=1,2,3$, there are eigenvector fields $\bmsX_{A}$, $A=1,2,3$, of $\msK$ of the form (\ref{eq:preferredcoordinates}) for a suitable choice of smooth
  functions $\bmsX_{1}^{2}$, $\bmsX_{2}^{3}$, $\bmsX_{1}^{3}$. Assume also that $\msK\bmsX_{A}=p_{A}\bmsX_{A}$ (no summation) for $A=1,2,3$. Define $c_{ij}$
  as in the statement of Proposition~\ref{prop:geometrictofal} and $\kappa_{i}^{\phantom{i}j}$ as in Definition~\ref{def:falid}. Assume, finally, that
  there is a solution to Einstein's vacuum equations of the form (\ref{eq:gfalasformofmetric}) on $V\times (0,t_{0}]$ for some $t_{0}>0$, where
  $dx^{i}$ is replaced by $d\bsfx^{i}$, and an $\vare>0$ such that
  \begin{align}
    \|a_{ij}(\cdot,t)-c_{ij}\|_{C^{l}(W)} \leq & C_{l}t^{\vare},\label{eq:aijtocijClatrate}\\
    \|t\bk^{\phantom{j}i}_{j}(\cdot,t)+\kappa^{\phantom{j}i}_{j}\|_{C^{l}(W)} \leq & C_{l}t^{\vare}\label{eq:tbkijtominuskappaclrate}
  \end{align}
  on $(0,t_{0}]$ for some constant $C_{l}$, where $l\in\nn{}$, $C_{l}\in\rn{}$ and $W\subset\bM$ is such that it has compact closure contained in $V$.
  Then the development (\ref{eq:gfalasformofmetric}) can be represented by the right hand side of (\ref{eq:rinansatzintroPsi}) on $V\times (0,t_{0}]$,
    where $\msY^{A}$ is replaced by $\bmsY^{A}$ (and $\{\bmsY^A\}$ is the dual frame of $\{\bmsX_A\}$). Moreover, defining $\chh$ by (\ref{eq:chhdef}),
    where $\msY^{A}$ is replaced by $\bmsY^{A}$, there is an $\eta>0$ and a constant $C_{l}$ such that
  \begin{align}
    \|\mK(\cdot,t)-\msK\|_{C^{l}(W)} \leq & C_{l}t^{\eta},\label{eq:mKmsKfaltoge}\\
    \|\chh(\cdot,t)-\bh\|_{C^{l}(W)} \leq & C_{l}t^{\eta},\label{eq:chhbhfaltoge}
  \end{align}
  hold on $(0,t_{0}]$. 
\end{prop}
\begin{remark}
  The conclusion that the development (\ref{eq:gfalasformofmetric}) can be represented by the right hand side of (\ref{eq:rinansatzintroPsi}) on
  $V\times (0,t_{0}]$, where $\msY^{A}$ is replaced by $\bmsY^{A}$, follows from the form of $g$ in (\ref{eq:gfalasformofmetric}) (that there is a
  Gaussian foliation so that all the information concerning the spacetime metric is contained in the induced Riemannian metrics
  on the leaves of the foliation) and the fact that $\{\bmsY^A\}$ is a co-frame. In particular, the estimates (\ref{eq:aijtocijClatrate}) and
  (\ref{eq:tbkijtominuskappaclrate}) are not necessary to obtain this conclusion. 
\end{remark}
\begin{proof}
  First note that the definitions ensure that $\kappa_{i}^{\phantom{i}j}=-\msK_{i}^{\phantom{i}j}$. 
  Due to (\ref{eq:tbkijtominuskappaclrate}), it is clear that $t\theta$ converges to $1$ in $C^{l}$ at the rate $t^{\vare}$. Combining this information
  with (\ref{eq:tbkijtominuskappaclrate}) again yields the conclusion that (\ref{eq:mKmsKfaltoge}) holds. It remains to relate $\chh$ to $\bh$. Note,
  to this end, that
  \[
  \bh=\textstyle{\sum}_{A}\bh(\bmsX_{A},\bmsX_{A})\bmsY^{A}\otimes \bmsY^{A}.
  \]
  Note, moreover, that $d\bsfx^{i}=\bmsX_{A}^{i}\bmsY^{A}$, where the $\bmsX_{A}^{i}$ are defined by
  (\ref{eq:preferredcoordinates}); $\bmsX_{A}^{i}=0$ if $i<A$; and $\bmsX_{A}^{i}=1$ if $A=i$. Thus the metric can be represented by the right hand side
  of (\ref{eq:rinansatzintroPsi}), with $\msY^{A}$ is replaced by $\bmsY^{A}$, where
  \begin{equation}\label{eq:bABitoaijetc}
    b_{AB}=\textstyle{\sum}_{i,j}a_{ij}t^{2p_{\max\{i,j\}}-2p_{\max\{A,B\}}}\bmsX^{i}_{A}\bmsX^{j}_{B}.
  \end{equation}
  Consider, to begin with, $b_{AA}$ (no summation). If $i$ or $j$ in (\ref{eq:bABitoaijetc}) is strictly larger than $A$, then the corresponding
  term in (\ref{eq:bABitoaijetc}) converges to zero. Moreover, we have to have $i\geq A$ and $j\geq A$ since $\bmsX_{A}^{i}=0$ if $i<A$. To conclude,
  the only term in (\ref{eq:bABitoaijetc}) which contributes to the limit is the one with $i=j=A$. Since $\bmsX_{A}^{i}=1$ if $A=i$, this term equals
  $a_{AA}\rightarrow c_{AA}=\bh(\bmsX_{A},\bmsX_{A})$ (no summation). In other words, $b_{AA}\rightarrow \bh(\bmsX_{A},\bmsX_{A})$ (no summation), as desired.
  Next, note that
  \[
  b_{23}=a_{33}\bmsX_{2}^{3}\bmsX_{3}^{3}+a_{23}\bmsX_{2}^{2}\bmsX_{3}^{3}=a_{33}\bmsX_{2}^{3}+a_{23}\rightarrow c_{33}\bmsX_{2}^{3}+c_{23}=0,
  \]
  by the definition of $c_{23}$. Next consider $b_{12}$. In this case the only terms in (\ref{eq:bABitoaijetc}) that contribute to the limit are
  the ones with $j=2$ and $i\in\{1,2\}$:
  \[
  a_{12}\bmsX_{1}^{1}\bmsX_{2}^{2}+a_{22}\bmsX_{1}^{2}\bmsX_{2}^{2}=a_{12}+a_{22}\bmsX_{1}^{2}\rightarrow c_{12}+c_{22}\bmsX_{1}^{2}=0,
  \]
  by the definition of $c_{12}$. Finally, consider $b_{13}$. In this case, the only terms in (\ref{eq:bABitoaijetc}) that contribute to the limit are
  the ones with $j=3$ and $i\in\{1,2,3\}$:
  \[
  a_{13}+a_{23}\bmsX_{1}^{2}+a_{33}\bmsX_{1}^{3}\rightarrow c_{13}+c_{23}\bmsX_{1}^{2}+c_{33}\bmsX_{1}^{3}=0,
  \]
  by the definition of $c_{13}$ and $c_{23}$. The lemma follows. 
\end{proof}

\section{Existence of developments}\label{section:existenceofdev}

The purpose of the present section is to prove Propositions~\ref{prop:idtovds} and \ref{prop:vdstoid}, as well as Theorem~\ref{thm:aardetal}.

\begin{proof}[Proposition~\ref{prop:idtovds}]
  In what follows, we refer to \cite{aarendall}, which covers the case $n=3$. However, the equations for the velocity dominated solutions are the
  same in higher dimensions; cf. \cite{damouretal}. Nevertheless, in all the verifications to follow, Remark~\ref{remark:eightpi} should be kept in mind.
  Define $({}^{0}g,{}^{0}k,{}^{0}\phi)$ by (\ref{eq:AVTDsolution}). Then
  ${}^{0}k^{A}_{\phantom{A}B}=-p_{A}t^{-1}\delta^{A}_{B}$ (no summation), where the indices are raised with ${}^{0}g$, and the indices refer to the frame
  $\{\msX_{A}\}$ and coframe $\{\msY^{A}\}$ appearing in the statement of the proposition. One particular consequence of this equality is that
  $\tr{}^{0}k=-t^{-1}$. Note also that $\d_{t}{}^{0}g_{AB}=-2{}^{0}k_{AB}$; i.e., \cite[(11a), p.~483]{aarendall} is satisfied. Next we define, in analogy
  with \cite[(12), p.~483]{aarendall},
  \begin{subequations}
    \begin{align}
      {}^{0}\rho := & \frac{1}{2}(\d_{t}{}^{0}\phi)^{2},\\
      {}^{0}j_{A} := & -(\d_{t}{}^{0}\phi)\cdot \msX_{A}({}^{0}\phi),\\
      {}^{0}S_{AB} := & \frac{1}{2}(\d_{t}{}^{0}\phi)^{2}\cdot ({}^{0}g_{AB}).
    \end{align}
  \end{subequations}
  With these definitions, the matter terms appearing on the right hand side of \cite[(11b), p.~483]{aarendall} vanish. Due to this observation, it can
  be verified that \cite[(11b), p.~483]{aarendall} is satisfied. It can also be verified that ${}^{0}\phi$ satisfies the relevant scalar field equation
  appearing on \cite[p.~483]{aarendall}. Since $\tr\msK^{2}+\Phia^{2}=1$, \cite[(10a), p.~483]{aarendall} holds. What remains is to demonstrate that
  \cite[(10b), p.~483]{aarendall} is fulfilled. Note that this equality can be rewritten as
  \[
  \nabla_{A}{}^{0}K^{A}_{\phantom{A}B}={}^{0}j_{B},
  \]
  where ${}^{0}K^{A}_{\phantom{A}B}={}^{0}k^{A}_{\phantom{A}B}$ (we use the capital letter $K$ to distinguish the Weingarten map from the second fundamental
  form) and $\nabla$ is the Levi-Civita connection induced by ${}^{0}g$. By arguments similar to the proof of Lemma~\ref{lemma:momcondivmsKz}, it can
  be calculated that
  \begin{align*}
    \nabla_{A}{}^{0}K^{A}_{\phantom{A}B} = & -t^{-1}\msX_{B}(p_{B})+\textstyle{\sum}_{A}t^{-1}(p_{A}-p_{B})\msY^{A}(\nabla_{\msX_{A}}\msX_{B}),\\
    \msY^{A}(\nabla_{\msX_{A}}\msX_{B}) = & \msY^{A}(\bD_{\msX_{A}}\msX_{B})+\msX_{B}(p_{A})\ln t,
  \end{align*}
  where there is no summation on $B$ in the first equality and no summation on any index in the second equality. Moreover, $\bD$ here denotes the
  Levi-Civita connection induced by $\bh$; note that $\{\msX_{A}\}$ is an orthonormal frame with respect to $\bh$. Due to these observations, it
  follows that
  \begin{equation}\label{eq:nablazKAB}
    \begin{split}
      \nabla_{A}{}^{0}K^{A}_{\phantom{A}B} = & -t^{-1}\bD_{A}\msK^{A}_{\phantom{A}B}+\textstyle{\sum}_{A}t^{-1}(p_{A}-p_{B})\msX_{B}(p_{A})\ln t.
    \end{split}
  \end{equation}
  The second term on the right hand side can be calculated to be
  \[
  \textstyle{\sum}_{A}t^{-1}(p_{A}-p_{B})\msX_{B}(p_{A})\ln t=\frac{1}{2}t^{-1}\msX_{B}\left(\textstyle{\sum}_{A}p_{A}^{2}\right)\ln t
  -t^{-1}p_{B}\msX_{B}\left(\textstyle{\sum}_{A}p_{A}\right)\ln t.
  \]
  However, $\sum_{A}p_{A}=1$, so that the second term on the right hand side vanishes. Combining this observation with $\sum_{A}p_{A}^{2}+\Phia^{2}=1$
  yields
  \begin{equation}\label{eq:secondtermnablazKid}
    \textstyle{\sum}_{A}t^{-1}(p_{A}-p_{B})\msX_{B}(p_{A})\ln t=-t^{-1}\ln t \cdot \Phia\msX_{B}(\Phia).
  \end{equation}
  Combining this equality with (\ref{eq:nablazKAB}) and the fact that $\rodiv_{\bh}\msK=\Phia d\Phib$ yields
  \[
  \nabla_{A}{}^{0}K^{A}_{\phantom{A}B} = -t^{-1}\Phia\msX_{B}(\Phib)-t^{-1}\ln t \cdot \Phia\msX_{B}(\Phia).
  \]
  On the other hand,
  \begin{equation}\label{eq:zjBeq}
    {}^{0}j_{B}=-t^{-1}\ln t \cdot \Phia\msX_{B}(\Phia)-t^{-1}\Phia\msX_{B}(\Phib).
  \end{equation}
  To conclude, \cite[(10b), p.~483]{aarendall} holds. In short, $({}^{0}g,{}^{0}k,{}^{0}\phi)$ is a solution to the velocity dominated
  Einstein-scalar field equations on $\bM\times (0,\infty)$ in the sense of \cite{aarendall}.
\end{proof}

Next, we prove Proposition~\ref{prop:vdstoid}.

\begin{proof}[Proposition~\ref{prop:vdstoid}]
  That $\tr{}^{0}k$ can be assumed to equal $-t^{-1}$ and that $-t\cdot {}^{0}K$ is independent of time follows from the discussion in
  \cite[Subsection~2.2, pp.~482--484]{aarendall}. Next, since $\msK$ is symmetric with respect to ${}^{0}g$, the eigenvalues of $\msK$ are
  real and the local frame $\{\bmsX_{A}\}$ is orthogonal with respect to ${}^{0}g$. Moreover, it can be verified that $t^{-2p_{A}}{}^{0}g_{AB}$
  (no summation) is independent of $t$. We can therefore define $\alpha_{A}$ to be the strictly positive functions such that
  $\alpha_{A}^{2}:={}^{0}g_{AA}|_{t=1}$ (no summation). This means that ${}^{0}g$ takes the form (\ref{eq:gzformulavd}). Define $\msX_{A}$
  by $\msX_{A}:=\alpha_{A}^{-1}\msX_{A}$ and let $\{\msY^{A}\}$ be the frame dual to $\{\msX_{A°}\}$. Then ${}^{0}g$ and ${}^{0}k$ take the
  form (\ref{eq:zgdef}) and (\ref{eq:zkdef}) respectively. Finally, since $\d_{t}(t{}^{0}\phi)=0$, cf. \cite[(15), p.~484]{aarendall},
  ${}^{0}\phi$ takes the form (\ref{eq:zphidef}).

  Next, let us verify that the conditions of Definition~\ref{def:ndsfidonbbs} are satisfied. That $\tr\msK=1$ follows from
  the fact that $\tr{}^{0}k=-t^{-1}$. That $\msK$ is symmetric with respect to $\bh$ follows from the fact that both are
  diagonal with respect to the frame $\{\msX_{A}\}$. That the eigenvalues of $\msK$ are distinct is an assumption. Since
  $1+p_{1}-p_{n-1}-p_{n}>0$, the final condition of Definition~\ref{def:ndsfidonbbs} is satisfied. Combining
  \cite[(10a) and (12a), p.~483]{aarendall}, keeping Remark~\ref{remark:eightpi} in mind, it follows that
  \[
  -t^{-2}\tr\msK^{2}+t^{-2}=t^{-2}\Phia^{2}.
  \]
  In particular, $\tr\msK^{2}+\Phia^{2}=1$. What remains is to verify that $\rodiv_{\bh}\msK=\Phia d\Phib$. Note, to this end, that
  (\ref{eq:nablazKAB}) and (\ref{eq:secondtermnablazKid}) are satisfied in the current situation. This means that
  \[
  \nabla_{A}{}^{0}K^{A}_{\phantom{A}B} = -t^{-1}\bD_{A}\msK^{A}_{\phantom{A}B}-t^{-1}\ln t \cdot \Phia\msX_{B}(\Phia).
  \]
  On the other hand, the equality (\ref{eq:zjBeq}) is also satisfied in the current situation. Since
  $\nabla_{A}{}^{0}K^{A}_{\phantom{A}B}={}^{0}j_{B}$, cf. \cite[(10b), p.~483]{aarendall}, it follows that $\rodiv_{\bh}\msK=\Phia d\Phib$.
  The proposition follows. 
\end{proof}

Finally, we prove Theorem~\ref{thm:aardetal}.

\begin{proof}[Theorem~\ref{thm:aardetal}]
  Due to Proposition~\ref{prop:idtovds}, $({}^{0}g,{}^{0}k,{}^{0}\phi)$ defined by (\ref{eq:AVTDsolution}) is a solution to the velocity dominated
  Einstein-scalar field equations on $\bM\times (0,\infty)$ in the sense of \cite{aarendall}, keeping Remark~\ref{remark:eightpi} in mind (as in
  the proof of Proposition~\ref{prop:idtovds}, we focus on the case $n=3$ and therefore refer to \cite{aarendall}; the general case follows by
  appealing to \cite{damouretal}, in particular \cite[Theorem~10.2, p.~1104]{damouretal} and \cite[Remark~1, p.~1105]{damouretal}). Assuming that
  the data $(\bM,\bh,\msK,\Phia,\Phib)$ are real analytic and that the eigenvalues
  satisfy $1+p_{1}-p_{n-1}-p_{n}>0$ (i.e., $p_{A}>0$, $A=1,2,3$, in case $n=3$), we can apply \cite[Theorem~1, p.~484]{aarendall}.
  This yields a real analytic solution to the Einstein-scalar field equations on a set of the form $M$ introduced in the statement of the
  theorem. Moreover, on a fixed compact subset of $\bM$, say $K$, this solution has asymptotics of the form
  \cite[1-6, Theorem~1, p.~484]{aarendall}, where the estimates are uniform on $K$. Finally, the solution is unique with these properties.

  In order to relate the results of \cite{aarendall} with the conclusions we wish to draw here, note that \cite[Theorem~1, p.~484]{aarendall}
  yields
  \[
  {}^{0}g^{AC}g_{CB}=\delta^{A}_{B}+o(t^{\alpha^{A}_{\phantom{A}B}}),
  \]
  where $\alpha^{A}_{\phantom{A}B}$ are strictly positive numbers and the indices refer to the frame $\{\msX_{A}\}$ and co-frame $\{\msY^{A}\}$. On the
  other hand, representing the metric as in (\ref{eq:rinansatzintroPsi}) yields
  \[
  {}^{0}g^{AC}g_{CB}=t^{-2p_{A}}g_{AB}=t^{2(p_{\max\{A,B\}}-p_{A})}b_{AB}.
  \]
  Combining these two equalities, it follows that the diagonal components of $b_{AB}$ converge to $1$ at a rate. Moreover, if $A>B$, we conclude that
  $b_{AB}$ converges to $0$ at a rate. Since $b$ is symmetric, the same holds for $b_{BA}$. This means that $\chh$, defined by (\ref{eq:chhdef}),
  converges to $\bh$ at a rate and uniformly on compact subsets; i.e., (\ref{eq:limitbABtobhchhsfaar}) holds. Combining this observation with
  arguments similar to the ones presented in connection
  with (\ref{eq:varrhothetadeterminant}) below yields the conclusion that $\varrho$ and $\ln t$ are interchangable; cf. (\ref{eq:varrholntCl}) below.
  Next, note that \cite[Theorem~1, p.~484]{aarendall} yields the conclusion that $t\theta$ converges to $1$ at a rate, and that $\msK-\mK$ converges
  to zero at a rate; i.e., (\ref{eq:mKlimsfaar}) holds. The conclusions concerning the scalar field, i.e. (\ref{eq:Phialimaar}) and
  (\ref{eq:Phiblimaar}), follow from \cite[Theorem~1, p.~484]{aarendall}, the fact that $\varrho$ and $\ln t$ are interchangable, the fact that
  $t\theta$ converges to $1$ at a rate and the fact that $\hU=\theta^{-1}\d_{t}$.
\end{proof}

\section{Improving the asymptotics in the generalised setting}

The purpose of the present section is to provide a proof of Theorem~\ref{thm:improvingasymptoticsgeneral}. We proceed gradually, starting with
estimates of the relative spatial variation of the mean curvature. We continue with estimates of the eigenvector fields of the expansion normalised
Weingarten map as well as of the components of the expansion normalised Weingarten map with respect to the eigenvector fields of $\msK$. Given
this information, we are in a position to estimate $\theta^{-2}\bS$, where $\bS$ is the spatial scalar curvature, and $\theta^{-2}\bR$, where
$\bR$ is the spatial Ricci tensor, viewed as a $(1,1)$-tensor field. We also obtain estimates for similar quantities involving the lapse function.
Combining this information, we are in a position to estimate the deceleration parameter $q$ as well as the expansion normalised normal derivative
of $\mK$, $\hml_{U}\mK$. We also obtain more detailed estimates of the components of $\hml_{U}\mK$ with respect to the eigenvector frame associated
with $\msK$. Given the above control of the geometry, it is possible to appeal to the results of \cite{RinWave} in order to control the asymptotics
of the scalar field. In the final step, we then combine all of the above in order to conclude that the equations corresponding to the limits of
the Hamiltonian and momentum constraints hold. 

\subsection{Estimates for the relative spatial variation of the mean curvature}\label{ssection:estrelspvariationmeancurvature}
Consider a development which is locally (close to the singularity) either of the form (\ref{eq:rinansatzintroPsi}) (in general dimensions) or
of the form (\ref{eq:rinansatzintrocrushing}). For such a development, define the logarithmic volume density, $\varrho$, by the relation
$e^{\varrho}\mu_{\bh}=\mu_{\bg}$, where $\bh$ is the metric appearing in the definition of the initial data and $\bg$ is the metric induced on
hypersurfaces of the form $\bM_{t}:=\bM\times\{t\}$. This means that
\[
e^{\varrho}=e^{\varrho}|\mu_{\bh}(\msX_{1},\dots,\msX_{n})|=|\mu_{\bg}(\msX_{1},\dots,\msX_{n})|,
\]
assuming $|\msX_{A}|_{\bh}=1$, $A=1,\dots,n$. In the case of (\ref{eq:rinansatzintrocrushing}), taking the square of this equality yields
\begin{equation}\label{eq:varrhothetadeterminant}
  e^{2\varrho}=\theta^{-2}\left|\begin{array}{cccc} b_{11} & b_{12} & \cdots & b_{1n} \\ \theta^{-2(p_{2}-p_{1})}b_{21} & b_{22} & \cdots & b_{2n} \\
  \vdots & \vdots & \cdots & \vdots \\
  \theta^{-2(p_{n}-p_{1})}b_{n1} & \theta^{-2(p_{n}-p_{2})}b_{n2} & \cdots & b_{nn}\end{array}\right|;
\end{equation}
note that the $p_{i}$ should here be thought of as given smooth functions satisfying $p_{1}<\cdots<p_{n}$. 
In the case of (\ref{eq:rinansatzintroPsi}), this equality holds with $\theta$ replaced by $t^{-1}$. In what follows, we assume $\chh$ to converge
to $\bh$ in $C^{l}$ for any $0\leq l\in\nn{}$ at a rate $\theta^{-\vare}$ in the case of (\ref{eq:rinansatzintrocrushing}) and at a rate $t^{\vare}$
in the case of (\ref{eq:rinansatzintroPsi}). This means that $b_{AB}$ converges to $\delta_{AB}$ at this rate. In particular, in the case of
(\ref{eq:rinansatzintrocrushing}),
\begin{equation}\label{eq:varrholnthetaCz}
  |\varrho+\ln\theta|\leq C_{\vare}\theta^{-\vare}
\end{equation}
on $M_{0}:=\bM\times (0,t_{0}]$, where $t_{0}\in (0,t_{+})$. Similarly, in the case of (\ref{eq:rinansatzintroPsi}), there is, for
each $0\leq l\in\nn{}$, a constant $C_{l}$ such that 
\begin{equation}\label{eq:varrholntCl}
  \|\varrho(\cdot,t)-\ln t\|_{C^{l}(\bM)}\leq C_{l}t^{\vare}
\end{equation}
for all $t\leq t_{0}$. The reason we obtain $C^{l}$ bounds in (\ref{eq:varrholntCl}) and only $C^{0}$ bounds in (\ref{eq:varrholnthetaCz}) is that we,
of course, control the spatial derivatives of $t$, but not, a priori, the spatial derivatives of $\theta$. In order to remedy this situation, we first
derive bounds on the spatial derivatives of $\varrho$, given suitable assumptions concerning $\hN$. These bounds can then be used to derive estimates
for $\ln\theta$. 

\begin{lemma}\label{lemma:bDkvarrholnthetaest}
  Consider a metric of the form (\ref{eq:asmetricgeneralcase}) and assume that there, for every $1\leq l\in\nn{}$, are constants $C_{l}$ and $a_{l}$
  such that for $1\leq k\leq l$,
  \begin{equation}\label{eq:EbfIlnhN}
    |\bD^{k}\ln\hN|_{\bh}\leq C_{l}\ldr{\ln\theta}^{a_{l}}
  \end{equation}
  on $M_{0}$, where $\hN=\theta N$. Assume, moreover, that (\ref{eq:limitbABtobhchhsfN}) holds and that $\theta$ diverges uniformly to $\infty$ as
  $t\rightarrow 0+$. Then there are, for every $1\leq l\in\nn{}$, constants $K_{l}$ and $b_{l}$ such that
  \begin{align}
    |\bD^{k}\varrho|_{\bh} \leq & K_{l}\ldr{\ln\theta}^{b_{l}},\ \ \
    |\bD^{k}\varrho|_{\bh}\leq K_{l}\ldr{\varrho}^{b_{l}},\label{eq:EbfIvarrhoest}\\
    |\bD^{k}\ln\theta|_{\bh} \leq & K_{l}\ldr{\ln\theta}^{b_{l}},\ \ \
    |\bD^{k}\ln\theta|_{\bh}\leq K_{l}\ldr{\varrho}^{b_{l}}\label{eq:EbfIlnthetaest}
  \end{align}
  on $M_{0}$ for $k\leq l$.
\end{lemma}
\begin{remark}
  Given (\ref{eq:EbfIlnthetaest}), one can return to (\ref{eq:varrhothetadeterminant}) in order to deduce that, for every $0\leq l\in\nn{}$,
  there is a constant $C_{l}$ such that
  \begin{equation}\label{eq:varrholnthetaCl}
    |\bD^{k}(\varrho+\ln\theta)|_{\bh}\leq C_{l}\theta^{-\vare}
  \end{equation}
  on $M_{0}$ for all $k\leq l$. 
\end{remark}
\begin{remark}\label{remark:hNcontrolyieldslnthetacontrol}
  If the conditions of Theorem~\ref{thm:improvingasymptoticsgeneral} are satisfied, then the conditions of the lemma are satisfied. 
\end{remark}
\begin{proof}
  Let $\{E_{i}\}$ be a local orthonormal frame with respect to $\bh$. In what follows, we then use the notation $\bfI$,
  $E_{\bfI}=E_{I_{1}}\cdots E_{I_{m}}$ and $|\bfI|=m$ for vectors of the form $\bfI=(I_{1},\dots,I_{m})$, where $I_{j}\in\{1,\dots,n\}$. 
  
  Due to \cite[(7.9), p.~74]{RinWave}, $\d_{t}\varrho=\hN$, so that $\d_{t}E_{\bfI}\varrho=E_{\bfI}\hN$. Due to this equality, it follows
  that $\hU(E_{\bfI}\varrho)$ can be written as a linear combination of terms of the form $E_{\bfI_{1}}\ln\hN\cdots E_{\bfI_{k}}\ln\hN$, where
  $|\bfI_{1}|+\dots+|\bfI_{k}|=|\bfI|$. Letting $\g$ be an integral curve of $\hU$, note that
  \begin{equation}\label{eq:varrhocges}
    \frac{d}{ds}\varrho\circ\g=(\hU\varrho)\circ\g=1;
  \end{equation}
  cf. \cite[(7.9), p.~74]{RinWave}. We can thus choose the parametrisation of $\g$ in such a way that $\varrho\circ\g(s)=s$. Then
  \[
  \left|\frac{d}{ds}[(E_{\bfI}\varrho)\circ\g](s)\right|\leq C_{l}\ldr{\ln\theta\circ\g(s)}^{c_{l}}\leq C_{l}\ldr{s}^{c_{l}},
  \]
  where we appealed to (\ref{eq:varrholnthetaCz}), (\ref{eq:EbfIlnhN}) and the fact that $\varrho\circ\g(s)=s$. This estimate can be integrated in
  order to deduce that (\ref{eq:EbfIvarrhoest}) holds. In order to obtain these estimates, we, again, appealed to (\ref{eq:varrholnthetaCz}) and
  the fact that $\varrho\circ\g(s)=s$.

  In order to estimate the derivatives of $\ln\theta$, note that (\ref{eq:varrhothetadeterminant}) yields
  \[
  \exp(2\varrho+2\ln\theta)=b_{11}\cdots b_{nn}+R,  
  \]
  where $R$, up to signs, consists of terms containing $k\geq 1$ factors of the form $\theta^{-2(p_{A}-p_{B})}b_{AB}$, $A>B$, and $n-k$ factors of the
  form $b_{AB}$, $A\leq B$. Differentiating this equality once yields
  \begin{equation}\label{eq:Eilnthetabaeq}
    2\exp(2\varrho+2\ln\theta)E_{i}\varrho+2\exp(2\varrho+2\ln\theta)E_{i}\ln\theta=E_{i}(b_{11}\cdots b_{nn})+E_{i}R.  
  \end{equation}
  Note that the factors $\exp(2\varrho+2\ln\theta)$ converge to $1$ due to (\ref{eq:varrholnthetaCz}). Moreover, the first term on the right hand
  side converges to zero as $\theta^{-\vare}$. Consider $E_{i}R$. It contains terms with two types of factors. Either $E_{i}$ hits
  $\theta^{-2(p_{A}-p_{B})}b_{AB}$, $A>B$, in which case the result is
  \[
  -2E_{i}(p_{A}-p_{B})\ln\theta\cdot \theta^{-2(p_{A}-p_{B})}b_{AB}-2(p_{A}-p_{B})E_{i}(\ln\theta)\theta^{-2(p_{A}-p_{B})}b_{AB}+\theta^{-2(p_{A}-p_{B})}E_{i}b_{AB}.
  \]
  The first term converges to zero faster than $\theta^{-\vare}$, the second term can be written as $E_{i}(\ln\theta)$ times a factor that
  converges to zero faster than $\theta^{-\vare}$, and the last term converges to zero faster than $\theta^{-\vare}$. Second, if $E_{i}$ hits
  a factor of the form $b_{AB}$, $A\leq B$, then the result converges to zero at least at the rate $\theta^{-\vare}$. To conclude, the right hand
  side of the equality (\ref{eq:Eilnthetabaeq}) contains a term including the factor $E_{i}\ln\theta$. However, in this term, the factor multiplying
  $E_{i}\ln\theta$ converges to zero faster than $\theta^{-\vare}$. The term on the right hand side including the factor $E_{i}\ln\theta$ can thus be
  combined with the second term on the left hand side. This yields a bound of $E_{i}\ln\theta$ in terms of $E_{i}\varrho$ and decaying terms.
  In order to estimate higher order derivatives, one can proceed similarly, differentiating (\ref{eq:Eilnthetabaeq}) further. By an inductive
  argument, all but the highest order derivatives appearing on the right hand side can be estimated. The highest order derivative of $\ln\theta$
  appearing on the right hand side will, in analogy with the above, be multiplied by a factor converging to zero. This can then be moved to the
  left hand side, as in the first step. At the end, we then obtain the estimate (\ref{eq:EbfIlnthetaest}).
\end{proof}

\subsection{Components of the eigenvector fields}

Next, introduce the matrices $X_{A}^{B}$ and $Y_{A}^{B}$ by the requirement that $X_{A}=X_{A}^{B}\msX_{B}$ and $\msX_{A}=Y_{A}^{B}X_{B}$, where
$\{\msX_{A}\}$ and $\{X_{A}\}$ are defined in the statement of Theorem~\ref{thm:improvingasymptoticsgeneral}. Then
\begin{equation}\label{eq:XABdefetc}
  X_{A}^{B}=\msY^{B}(X_{A}),\ \ \
  Y_{A}^{B}=Y^{B}(\msX_{A}),\ \ \
  X_{A}^{B}Y_{B}^{C}=Y_{A}^{B}X_{B}^{C}=\delta^{C}_{A}.
\end{equation}
These matrices satisfy the following relations. 

\begin{lemma}\label{lemma:XABest}
  Given the assumptions and notation of Theorem~\ref{thm:improvingasymptoticsgeneral}, there is an $\eta>0$ and, for all $k\in\nn{}$, a constant $C_{k}$,
  such that
  \begin{align}
    |\bD^{k}(X^{A}_{B}-\delta^{A}_{B})|_{\bh} \leq & C_{k}\theta^{-2\eta}\min\{1,\theta^{-2(p_{B}-p_{A})}\},\label{eq:bmsYXB}\\
    |\bD^{k}(Y^{A}_{B}-\delta^{A}_{B})|_{\bh} \leq & C_{k}\theta^{-2\eta}\min\{1,\theta^{-2(p_{B}-p_{A})}\},\label{eq:YAbmsXB}\\
    |\bD^{k}(\bmu_{A}+p_{A}\ln\theta)|_{\bh} \leq & C_{k}\theta^{-2\eta}\label{eq:bmuApAlntheta}
  \end{align}
  on $M_{0}$ for all $A$ and $B$. 
\end{lemma}
\begin{proof}
  Note, to begin with, that (\ref{eq:mKlimsfN}), (\ref{eq:limitbABtobhchhsfN}) and the definition of $X_{A}$, $Y^{A}$, $\msX_{A}$ and
  $\msY^{A}$ yield the conclusion that for all $k\in\nn{}$, there are constants $C_{k}$ and $\eta>0$ such that 
  \begin{equation}\label{eq:bmsYAXBetcbasicest}
    |\bD^{k}[\msY^{A}(X_{B})-\delta^{A}_{B}]|_{\bh}+|\bD^{k}[Y^{A}(\msX_{B})-\delta^{A}_{B}]|_{\bh}
    +|\bD^{k}(b_{AB}-\delta_{AB})|_{\bh}\leq C_{k}\theta^{-2\eta}
  \end{equation}
  on $M_{0}$ for all $A$, $B$. In particular, (\ref{eq:bmsYXB}) and (\ref{eq:YAbmsXB}) hold in case $B\leq A$. On the other hand, 
  \begin{equation}\label{eq:tworepresentations}
    \textstyle{\sum}_{A,B}b_{AB}\theta^{-2p_{\max\{A,B\}}}\msY^{A}\otimes \msY^{B}=\textstyle{\sum}_{C}e^{2\bmu_{C}}Y^{C}\otimes Y^{C}. 
  \end{equation}
  Evaluating this equality on $(\msX_{A},X_{A})$ yields
  \[
  \textstyle{\sum}_{B}b_{AB}\theta^{-2p_{\max\{A,B\}}}\msY^{B}(X_{A})=e^{2\bmu_{A}}Y^{A}(\msX_{A})
  \]
  (no summation on $A$). Combining this equality with (\ref{eq:EbfIlnthetaest}) and (\ref{eq:bmsYAXBetcbasicest}) yields
  (\ref{eq:bmuApAlntheta}). Next, evaluating (\ref{eq:tworepresentations}) on $(\msX_{A},X_{C})$ with $C<A$ yields
  \[
  \textstyle{\sum}_{B}b_{AB}\theta^{-2p_{\max\{A,B\}}}\msY^{B}(X_{C})=e^{2\bmu_{C}}Y^{C}(\msX_{A}).
  \]
  Combining this equality with (\ref{eq:bmuApAlntheta}) and (\ref{eq:bmsYAXBetcbasicest}) yields (\ref{eq:YAbmsXB}). 

  In order to prove (\ref{eq:bmsYXB}), note that $YX=\mathrm{Id}$, where $X$ and $Y$ are the matrices with components
  $X^{A}_{B}$ and $Y^{A}_{B}$ respectively. Next, note that $Y=U+L$, where $U$ is upper triangular, $L$ is lower triangular and
  the diagonal components of $L$ vanish. Moreover, for $A>B$, $L_{A}^{B}=O(\theta^{-2(p_{A}-p_{B}+\eta)})$ in any $C^{m}$-norm. Moreover,
  $U-\mathrm{Id}=O(\theta^{-2\eta})$ in any $C^{m}$-norm. In particular,
  \[
  X=Y^{-1}=(\mathrm{Id}+U^{-1}L)^{-1}U^{-1};
  \]
  note that $\|U^{-1}L\|\rightarrow 0$ in the direction of the singularity, so that we can assume $\|U^{-1}L\|\leq 1/2$ in the arguments
  to follow. Let $R:=-U^{-1}L$. Then
  \[
  X=\textstyle{\sum}_{k=0}^{\infty}R^{k}U^{-1}.
  \]
  Note that $R^{k}=O(\theta^{-2k\eta})$ in any $C^{m}$ norm. Moreover, for $k$ large enough, say $k\geq K$, $2k\eta\geq 2(p_{n}-p_{1}+\eta)$
  on $\bM$. It is therefore straightforward to prove the desired estimate for the terms with $k\geq K$, and we only need to focus on a
  finite number of terms of the form $R^{k}U^{-1}$. However, it can be verified that the components of $R$ satisfy $R_{A}^{n}=0$, $A=1,\dots,n$,
  and, for $B<n$,
  \[
  R_{A}^{B}=\left\{\begin{array}{ll} O(\theta^{-2(p_{B+1}-p_{B}+\eta)}) & B\geq A, \\ O(\theta^{-2(p_{A}-p_{B}+\eta)}) & B<A.\end{array}\right.
  \]
  Moreover, this holds in any $C^{m}$ norm. It can be verified that this structure is preserved when taking positive integer powers. Finally,
  when multiplying a matrix with this structure on the right with $U^{-1}$ (which is upper triangular), the lower triangular part of the resulting
  matrix satisfies the same conditions as the lower triangular part of $R$ (i.e., the estimates satisfied by $R_{A}^{B}$, $B<A$, are also satisfied
  by $(R^{m}U^{-1})_{A}^{B}$ for $m\geq 1$ and $B<A$). It follows that (\ref{eq:bmsYXB}) holds for $B>A$. The lemma follows.   
\end{proof}

Using this information, we are in a position to prove that (\ref{eq:mKestrefined}) holds. 

\begin{lemma}\label{lemma:mKmsYAmsXBestimate}
  Given the assumptions and notation of Theorem~\ref{thm:improvingasymptoticsgeneral}, there is an $\eta>0$ and, for all $k\in\nn{}$,
  a constant $C_{k}$, such that
  \begin{equation}\label{eq:mKmsYAmsXBestimate}
    |D^{k}[\mK(\msY^{A},\msX_{B})-p_{B}\delta^{A}_{B}]|_{\bh} \leq C_{k}e^{-2\eta}\min\{1,\theta^{-2(p_{B}-p_{A})}\}
  \end{equation}
  (no summation on $B$) on $M_{0}$.
\end{lemma}
\begin{proof}
  Note, to begin with, that
  \[
  \mK^{A}_{B}:=\mK(\msY^{A},\msX_{B})=\textstyle{\sum}_{C}\ell_{C}\msY^{A}(X_{C})Y^{C}(\msX_{B})
  =\textstyle{\sum}_{C}\ell_{C}X^{A}_{C}Y^{C}_{B},
  \]
  where $\ell_{A}$ is the eigenvalue of $\mK$ corresponding to $X_{A}$. 
  Note also that $\ell_{A}-p_{A}=O(\theta^{-2\eta})$ in any $C^{k}$-norm due to (\ref{eq:mKlimsfN}). Combining these observations with
  Lemma~\ref{lemma:XABest} yields $\mK^{A}_{A}=p_{A}+O(\theta^{-2\eta})$ (no summation on $A$) in any $C^{k}$-norm. Similarly,
  $\mK_{B}^{A}=O(\theta^{-2\eta})$ for $B<A$ in any $C^{k}$-norm. Finally, for $A<B$, the desired conclusion again follows by appealing to
  Lemma~\ref{lemma:XABest}. The desired estimate follows. 
\end{proof}

Next, define $\lambda^{A}_{BC}$ by $[X_{B},X_{C}]=\lambda_{BC}^{A}X_{A}$. We wish to estimate $\lambda^{A}_{BC}$ in case $1+p_{A}-p_{B}-p_{C}\leq 0$; cf. the
definition of quiescent initial data on the singularity. 
\begin{lemma}\label{lemma:gammaonetwothree}
  Given the assumptions and notation of Theorem~\ref{thm:improvingasymptoticsgeneral}, fix $\bx\in\bM$ and assume that
  $1+p_{A}(\bx)-p_{B}(\bx)-p_{C}(\bx)\leq 0$ for some $A$ and $B\neq C$. Then there is a neighbourhood $U$ of $\bx$ in $\bM$, an $\eta>0$ and,
  for every $k\in\nn{}$, a constant $C_{k}$ such that
  \begin{equation}\label{eq:bgammaABCdec}
    |D^{k}(\theta^{-2(1+p_{A}-p_{B}-p_{C})}\lambda^{A}_{BC})|_{\bh} \leq C_{k}\theta^{-2\eta}
  \end{equation}
  on $U\times (0,t_{0}]$. 
\end{lemma}
\begin{remark}\label{remark:gammaonetwothree}
  Assuming, in addition to the conditions of the lemma, that $n=3$ and $\msO^{1}_{23}=0$ on $\bM$, there is an $\eta>0$ and, for all $k\in\nn{}$, a
  constant $C_{k}$, such that
  \begin{equation}\label{eq:gammaonetwothree}
    |D^{k}(\lambda^{1}_{23})|_{\bh} \leq C_{k}\theta^{-2(p_{2}-p_{1}+\eta)}
  \end{equation}
  on $M_{0}$.
\end{remark}
\begin{proof}
  Compute
  \begin{equation}\label{eq:YAXBXCexpanded}
    Y^{A}([X_{B},X_{C}])=X_{B}^{D}\msX_{D}(X^{E}_{C})Y^{A}_{E}-X_{C}^{E}\msX_{E}(X^{D}_{B})Y^{A}_{D}
    +X^{D}_{B}X^{E}_{C}Y^{A}_{F}\g^{F}_{DE},
  \end{equation}
  where $\g^{F}_{DE}=\msY^{F}([\msX_{D},\msX_{E}])$. Assume that $1+p_{A}(\bx)-p_{B}(\bx)-p_{C}(\bx)\leq 0$. In order for this to hold, we have to
  have $B>A$ and $C>A$; cf. Remark~\ref{remark:pAsltone}. In what follows, we also assume, without loss of generality, that $B<C$. Combining these
  observations with
  Lemma~\ref{lemma:XABest}, it is clear that the first two terms on the right hand side of (\ref{eq:YAXBXCexpanded}) are
  $O(\theta^{-2(p_{B}-p_{A}+\eta)})$ in any $C^{k}$ norm. In particular, multiplying the first two terms on the right hand side of (\ref{eq:YAXBXCexpanded})
  with $\theta^{-2(1+p_{A}-p_{B}-p_{C})}$, it is clear that the resulting expression is $O(\theta^{-2(1-p_{C}+\eta)})$ in any $C^{k}$-norm. Next, consider
  \begin{equation*}
    \begin{split}
      & \theta^{-2(1+p_{A}-p_{B}-p_{C})}X^{D}_{B}X^{E}_{C}Y^{A}_{F}\g^{F}_{DE}\\
      = &
      (\theta^{2(p_{B}-p_{D})}X^{D}_{B})\cdot (\theta^{2(p_{C}-p_{E})}X^{E}_{C})\cdot (\theta^{2(p_{F}-p_{A})}Y^{A}_{F})\cdot
      (\theta^{-2(1+p_{F}-p_{D}-p_{E})}\g^{F}_{DE}).
    \end{split}
  \end{equation*}
  Due to Lemma~\ref{lemma:XABest}, the first three factors on the right hand side are bounded in any $C^{k}$ norm. In order for the last term to
  be non-zero, $1+p_{F}-p_{D}-p_{E}$ has to be strictly positive in a neighbourhood of $\bx$. The estimate (\ref{eq:bgammaABCdec}) follows. The
  proof of (\ref{eq:gammaonetwothree}) is similar. 
\end{proof}

\begin{cor}\label{cor:decayofproblematicterms}
  Given the assumptions and notation of Theorem~\ref{thm:improvingasymptoticsgeneral}, there is an $\eta>0$ and, for every $k\in\nn{}$, a constant
  $C_{k}$ such that
  \begin{equation}\label{eq:bgammaABCdecgeneralcase}
    |D^{k}(\theta^{-2(1+p_{A}-p_{B}-p_{C})}\lambda^{A}_{BC})|_{\bh} \leq C_{k}\theta^{-2\eta}
  \end{equation}
  on $M_{0}$ for all $A$, $B$ and $C$.
\end{cor}
\begin{remark}
  Since $\lambda_{BC}^{A}$ is well defined up to a sign, the estimate (\ref{eq:bgammaABCdecgeneralcase}) is globally meaningful.
\end{remark}
\begin{proof}
  If $B=C$, there is nothing to prove. In what follows, we therefore assume $B\neq C$. Fix an $\bx\in\bM$.  Given $A$, $B$, $C$ such that $B\neq C$,
  there are then two possibilities. Either $1+p_{A}(\bx)-p_{B}(\bx)-p_{C}(\bx)\leq 0$
  or $1+p_{A}(\bx)-p_{B}(\bx)-p_{C}(\bx)>0$. In the first case, we can appeal to Lemma~\ref{lemma:gammaonetwothree} in order to conclude
  that there is an open neighbourhood $U_{ABC}$ and an $\eta_{ABC}>0$ such that (\ref{eq:bgammaABCdec}), with $\eta$ replaced by $\eta_{ABC}$,
  holds in $U_{ABC}\times (0,t_{0}]$. In case $1+p_{A}(\bx)-p_{B}(\bx)-p_{C}(\bx)>0$, there is trivially such an open neighbourhood
  $U_{ABC}$ and $\eta_{ABC}>0$. Taking the intersection of all these $U_{ABC}$ yields an open neighbourhood $U_{\bx}$ and taking the minimum
  of all the $\eta_{ABC}$ yields an $\eta_{\bx}>0$. Since $\bM$ is compact, there is a finite cover by sets of the form $U_{\bx}$. Taking the
  minimum of the associated $\eta_{\bx}$ yields an $\eta>0$ such that the statement of the corollary holds. 
\end{proof}

Using this information, we estimate the conformally rescaled spatial Ricci curvature.

\begin{prop}\label{prop:bmRestimate}
  Assume that the conditions of Theorem~\ref{thm:improvingasymptoticsgeneral} are satisfied. Let, using the notation of
  Theorem~\ref{thm:improvingasymptoticsgeneral}, $\bR$ denote the spatial Ricci
  tensor, viewed as a family of $(1,1)$-tensor fields, $\bS$ denote the spatial scalar curvature, $\bmR:=\theta^{-2}\bR$ and
  $\bmS:=\theta^{-2}\bS$. Then there is an $\eta>0$ and, for each $k\in\nn{}$, a $C_{k}$ such that 
  \begin{equation}\label{eq:bmSbmRexpdecay}
    |\bD^{k}\bmS|_{\bh}+|\bD^{k}\bmR|_{\bh}\leq C_{k}\theta^{-2\eta}
  \end{equation}
  on $M_{0}$.
\end{prop}
\begin{remark}
  If, in addition to the assumptions of the proposition, $n=3$ and $\msO^{1}_{23}=0$ on $\bM$, then there are, for each $k\in\nn{}$, constants
  $\alpha_{k}$ and $C_{k}$ such that
  \begin{equation}\label{eq:bmSbmRexpdecayneqthree}
    |\bD^{k}\bmS|_{\bh}+|\bD^{k}\bmR|_{\bh}\leq C_{k}\ldr{\ln\theta}^{\alpha_{k}}\theta^{-2(1-p_{3})}
  \end{equation}
  on $M_{0}$.
\end{remark}
\begin{proof}
  To estimate $\bmR$, it is convenient to appeal to \cite[Corollary~125, p.~73]{RinGeo}. The $\mu_{A}$ appearing in this corollary
  are given by $\mu_{A}=\bmu_{A}+\ln\theta$, so that $\bD^{k}[\mu_{A}-(1-p_{A})\ln\theta]=O(\theta^{-2\eta})$
  due to (\ref{eq:bmuApAlntheta}).  Note that $\lambda^{A}_{BC}$ (the $\g^{A}_{BC}$ appearing in \cite[Corollary~125, p.~73]{RinGeo} should be
  replaced by $\lambda^{A}_{BC}$ here) and $a_{A}$ are bounded in any $C^{k}$-norm and that $\check{a}_{A}$, $\bmu_{A}$
  and $\mu_{A}$ do not grow faster than polynomially in $\ln\theta$ in any $C^{k}$-norm. In particular, any number of $X_{A}$ derivatives of these
  objects grows at worst polynomially. Note also that $p_{A}<1$ due to Remark~\ref{remark:pAsltone}. Combining these observations
  with \cite[(301)--(306), p.~73]{RinGeo}, it is clear that for
  $k\in\nn{}$, there is a constant $\alpha_{k}$ such that most of the terms on the right hand side of \cite[(301), p.~73]{RinGeo}
  decay to zero as $O(\ldr{\ln\theta}^{\alpha_{k}}\theta^{-2(1-p_{n})})$ in $C^{k}$. However, there are some terms containing factors of
  the form $e^{2\mu_{A}-2\mu_{B}-2\mu_{C}}$. On the other hand, terms that contain a factor $e^{2\mu_{A}-2\mu_{B}-2\mu_{C}}$ also contain a factor
  $\lambda^{A}_{BC}$ (or a spatial derivative thereof). In order to estimate the corresponding terms, we can thus appeal to
  Corollary~\ref{cor:decayofproblematicterms}. In order to illustrate the arguments required, consider, e.g., 
  \begin{equation*}
    \begin{split}
      & e^{2\mu_{A}-2\mu_{B}-2\mu_{C}}X_{D}(\lambda^{A}_{BC})\\
      = & e^{2\mu_{A}-2\mu_{B}-2\mu_{C}}\theta^{2(1+p_{A}-p_{B}-p_{C})}X_{D}(\theta^{-2(1+p_{A}-p_{B}-p_{C})}\lambda^{A}_{BC})\\
      & +e^{2\mu_{A}-2\mu_{B}-2\mu_{C}}\theta^{2(1+p_{A}-p_{B}-p_{C})}X_{D}[2(1+p_{A}-p_{B}-p_{C})\ln\theta](\theta^{-2(1+p_{A}-p_{B}-p_{C})}\lambda^{A}_{BC}).
    \end{split}
  \end{equation*}
  Due to Corollary~\ref{cor:decayofproblematicterms}, terms of this form decay as desired in any $C^{k}$-norm; since
  $\mu_{A}-(1-p_{A})\ln\theta=O(\theta^{-2\eta})$, the product of the first two factors is, in each term, bounded in any $C^{k}$-norm;
  the third factor in the second term does not grow faster than polynomially in any $C^{k}$-norm; and the last factor decays as
  $\theta^{-2\eta}$ in any $C^{k}$-norm in each of the terms.

  In case $n=3$, a more detailed argument can be carried out. In fact, then, if $1\in \{B,C\}$, $e^{2\mu_{A}-2\mu_{B}-2\mu_{C}}$ decays to zero as
  $O(\ldr{\ln\theta}^{\alpha_{k}}\theta^{-2(1-p_{3})})$ in $C^{k}$. Assume, therefore,
  that $\{B,C\}=\{2,3\}$. If $A\in\{2,3\}$, we again obtain the same decay. The only interesting terms are thus the ones containing one
  factor of the form $e^{2\mu_{1}-2\mu_{2}-2\mu_{3}}$ and a factor of the form $\lambda^{1}_{23}$ or a spatial derivative thereof. However, we know
  that $\lambda^{1}_{23}=O(\theta^{-2(p_{2}-p_{1}+\eta)})$ (and similarly for every $C^{k}$-norm); cf. Remark~\ref{remark:gammaonetwothree}.
  Thus the potentially problematic terms are $O(\theta^{-2(1-p_{3})})$ (and similarly for every $C^{k}$-norm).

  Writing
  \[
  \bmR=\bmR^{A}_{B}X_{A}\otimes Y^{B},
  \]
  the desired estimate for $\bmR$ follows. Due to \cite[Lemma~122, p.~69]{RinGeo}, the proof of the estimate for $\bmS$ is essentially
  identical.
\end{proof}

Next, we wish to estimate the contribution from the lapse function to $\hml_{U}\mK$. Introduce, to this end, the $(1,1)$-tensor field
$\bmN$, given by
\begin{equation}\label{eq:bmNdef}
  \bmN^{A}_{B}:=\theta^{-2}N^{-1}\bnabla^{A}\bnabla_{B}N.
\end{equation}
Here $\bnabla$ is the Levi-Civita connection induced on the hypersurfaces of constant $t$.

\begin{lemma}
  Assume that the conditions of Theorem~\ref{thm:improvingasymptoticsgeneral} are satisfied and let $\bmN$ be defined by
  (\ref{eq:bmNdef}). Then, using the notation of Theorem~\ref{thm:improvingasymptoticsgeneral}, there is an $\eta>0$ and,
  for each $k\in\nn{}$, a $C_{k}$ such that 
  \begin{equation}\label{eq:bmNexpdecay}
    |\bD^{k}\bmN|_{\bh}\leq C_{k}\theta^{-2\eta}
  \end{equation}
  on $M_{0}$.
\end{lemma}
\begin{remark}
  If, in addition to the assumptions of the proposition, $n=3$ and $\msO^{1}_{23}=0$ on $\bM$, then there are, for each $k\in\nn{}$, constants
  $\alpha_{k}$ and $C_{k}$ such that
  \[
  |\bD^{k}\bmN|_{\bh}\leq C_{k}\ldr{\ln\theta}^{\alpha_{k}}\theta^{-2(1-p_{3})}
  \]
  on $M_{0}$.
\end{remark}
\begin{proof}
  Note, first of all, that the derivatives of $\ln N$ can be bounded by appealing to (\ref{eq:lnthetalnNphipolbd}) and
  (\ref{eq:EbfIlnthetaest}). Combining this observation with \cite[Lemma~127, p.~74]{RinGeo} and the assumptions, the proof
  of the statement is similar to, but simpler than, the proof of Proposition~\ref{prop:bmRestimate}. 
\end{proof}

Next, we estimate $\hml_{U}\mK$.

\begin{prop}\label{prop:hmlUmKest}
  Assume that the conditions of Theorem~\ref{thm:improvingasymptoticsgeneral} are satisfied. Then, using the notation of
  Theorem~\ref{thm:improvingasymptoticsgeneral}, there is an $\eta>0$ and, for each $k\in\nn{}$, a $C_{k}$ such that
  \begin{equation}\label{eq:hmlUmKandqest}
    |\bD^{k}[q-(n-1)]|_{\bh}+|\bD^{k}\hml_{U}\mK|_{\bh}\leq C_{k}\theta^{-2\eta}
  \end{equation}
  on $M_{0}$.
\end{prop}
\begin{remark}
  If, in addition to the assumptions of the proposition, $n=3$ and $\msO^{1}_{23}=0$ on $\bM$, then there are, for each $k\in\nn{}$, constants
  $\alpha_{k}$ and $C_{k}$ such that
  \begin{equation}\label{eq:hmlUmKandqestnequalsthree}
    |\bD^{k}(q-2)|_{\bh}+|\bD^{k}\hml_{U}\mK|_{\bh}\leq C_{k}\ldr{\ln\theta}^{\alpha_{k}}\theta^{-2(1-p_{3})}
  \end{equation}
  on $M_{0}$.
\end{remark}
\begin{remark}\label{remark:chKsilentupperbound}
  When appealing to the results of \cite{RinWave}, we need upper bounds on the eigenvalues of $\chK$, the Weingarten map of the conformally
  rescaled metric $\hg:=\theta^{2}g$. Combining \cite[(3.3), p.~26]{RinWave} with \cite[(3.4), p.~26]{RinWave}, we conclude that
  $\chK=\mK-(q+1)\mathrm{Id}/n$. Since $q$ converges to $n-1$, cf. (\ref{eq:hmlUmKandqest}), we conclude that the eigenvalues of $\chK$
  converge to $p_{A}-1$. In particular, they are all negative, cf. Remark~\ref{remark:pAsltone}, so that $\chK$ has a silent upper bound
  asymptotically; cf. \cite[Definition~3.10, p.~27]{RinWave}.
\end{remark}
\begin{proof}
  The basis for the estimate of $\hml_{U}\mK$ is the formula for $\hml_{U}\mK$ given by \cite[(270), p.~65]{RinGeo}:
  \begin{equation}\label{eq:mlUmKwithEinstein}
    \begin{split}
      (\hml_{U}\mK)^{i}_{\phantom{i}j} =  & -\left(\frac{\Delta_{\bg}N}{\theta^{2}N}+\frac{n}{n-1}\frac{\rho-\bp+2\Lambda}{\theta^{2}}
      -\frac{\bS}{\theta^{2}}\right)\mK^{i}_{\phantom{i}j}\\
      & +\frac{1}{(n-1)\theta^{2}}\left(\rho-\bp+2\Lambda\right)\delta^{i}_{j}
      +\frac{1}{N\theta^{2}}\bnabla^{i}\bnabla_{j}N+\theta^{-2}\mP^{i}_{\phantom{i}j}-\frac{\bR^{i}_{\phantom{i}j}}{\theta^{2}}.
    \end{split}
  \end{equation}  
  Several of the terms can be estimated by appealing to (\ref{eq:bmSbmRexpdecay}), (\ref{eq:bmNexpdecay}) and the convergence of $\mK$. However, there
  are three types of exceptions: terms including a factor of
  the form $\theta^{-2}(\rho-\bp)$, a factor of the form $\Lambda/\theta^{2}$ or a factor of the form $\theta^{-2}\mP$. Estimates for terms including
  a factor of the form $\Lambda/\theta^{2}$ follow immediately from (\ref{eq:EbfIlnthetaest}). To estimate the remaining terms, it is sufficient
  to appeal to \cite[Remark~81, p.~49]{RinGeo}; given the corresponding formulae and (\ref{eq:lnthetalnNphipolbd}), it is clear that we can proceed
  as in the previous proofs in order to derive the desired estimate for $\hml_{U}\mK$.

  The estimate for $q-(n-1)$ can be derived by a similar argument due to \cite[(266), p.~65]{RinGeo}.
\end{proof}

In order to refine this conclusion, it is of interest to consider the components of $\hml_{U}\mK$ with respect to the frame $\{\msX_{A}\}$.
\begin{lemma}
  Assume that the conditions of Theorem~\ref{thm:improvingasymptoticsgeneral} are satisfied. Then, using the notation of
  Theorem~\ref{thm:improvingasymptoticsgeneral}, there is an $\eta>0$ and, for each $k\in\nn{}$, a constant $C_{k}$ such that
  \begin{equation}\label{eq:hmlUcompbettest}
    |\bD^{k}[(\hml_{U}\mK)(\msY^{A},\msX_{B})]|_{\bh}\leq C_{k}\theta^{-2\eta}\min\{1,\theta^{-2(p_{B}-p_{A})}\}
  \end{equation}
  on $M_{0}$. 
\end{lemma}
\begin{remark}
  If, in addition to the assumptions of the lemma, $n=3$ and $\msO^{1}_{23}=0$ on $\bM$, then there are, for each $k\in\nn{}$, constants
  $\alpha_{k}$ and $C_{k}$ such that
  \[
  |\bD^{k}[(\hml_{U}\mK)(\msY^{A},\msX_{B})]|_{\bh}\leq C_{k}\ldr{\ln\theta}^{\alpha_{k}}\theta^{-2(1-p_{3})}\min\{1,\theta^{-2(p_{B}-p_{A})}\}
  \]
  on $M_{0}$.
\end{remark}
\begin{proof}
  For $A\geq B$, the estimates are immediate consequences of (\ref{eq:hmlUmKandqest}) and (\ref{eq:hmlUmKandqestnequalsthree}). We therefore need
  to focus on the case $B>A$. Define, for $A\neq B$,
  \begin{equation}\label{eq:mWABdef}
    \mW^{A}_{B}:=\frac{1}{\ell_{B}-\ell_{A}}(\hml_{U}\mK)(Y^{A},X_{B});
  \end{equation}
  cf. \cite[(6.8), p.~68]{RinWave}. Due to Proposition~\ref{prop:hmlUmKest}, we know that the $C^{k}$-norm of $\mW^{A}_{B}$ is bounded by the
  right hand side of (\ref{eq:hmlUmKandqest}). On the other hand, \cite[(6.19), p.~70]{RinWave} implies that $e^{\bmu_{A}-\bmu_{B}}\mW^{A}_{B}$ is
  antisymmetric for $A\neq B$. In particular, if $A<B$, then
  \begin{equation}\label{eq:crucialantisymmetry}
    \mW^{A}_{B}=-e^{2(\bmu_{B}-\bmu_{A})}\mW^{B}_{A}.
  \end{equation}
  There is thus an $\eta>0$ and, for each $k\in\nn{}$, a constant $C_{k}$ such that
  \begin{equation}\label{eq:mWABhmlUmKXAYBprelest}
    |\bD^{k}\mW^{A}_{B}|_{\bh}+|\bD^{k}[(\hml_{U}\mK)(Y^{A},X_{B})]|_{\bh}\leq C_{k}\theta^{-2(p_{B}-p_{A})}\theta^{-2\eta}
  \end{equation}
  on $M_{0}$ if $A<B$, where we appealed to (\ref{eq:bmuApAlntheta}) and to (\ref{eq:mWABdef}) in order to estimate the second term on the
  left hand side of (\ref{eq:mWABhmlUmKXAYBprelest}). In case $n=3$ and $\msO^{1}_{23}=0$ on $\bM$, this estimate can be improved to
  \begin{equation}\label{eq:mWABhmlUmKXAYBprelestneqthree}
    |\bD^{k}\mW^{A}_{B}|_{\bh}+|\bD^{k}[(\hml_{U}\mK)(Y^{A},X_{B})]|_{\bh}\leq C_{k}\ldr{\ln\theta}^{\alpha_{k}}\theta^{-2(p_{B}-p_{A})}\theta^{-2(1-p_{3})}
  \end{equation}
  on $M_{0}$ if $A<B$, where $\alpha_{k}$ and $C_{k}$ are constants. 

  Next, note that
  \[
  (\hml_{U}\mK)(\msY^{A},\msX_{B})=X_{C}^{A}Y^{D}_{B}(\hml_{U}\mK)(Y^{C},X_{D}).
  \]
  Combining this observation with Lemma~\ref{lemma:XABest}, (\ref{eq:mWABhmlUmKXAYBprelest}) and (\ref{eq:mWABhmlUmKXAYBprelestneqthree}) yields
  the conclusion of the lemma. 
\end{proof}

\subsection{Scalar field}

Next, we wish to analyse the asymptotic behaviour of the scalar field. We do so by appealing to the results of \cite{RinWave}.

\begin{lemma}\label{lemma:scalarfieldfromgeometry}
  Assume that the conditions of Theorem~\ref{thm:improvingasymptoticsgeneral} are satisfied. Then, using the notation of
  Theorem~\ref{thm:improvingasymptoticsgeneral}, there are $\Phia,\Phib\in C^{\infty}(\bM)$,
  an $\eta>0$ and, for each $k\in\nn{}$, a constant $C_{k}$ such that
  \begin{align}
    |\bD^{k}(\hU\phi-\Phia)|_{\bh} \leq & C_{k}\theta^{-2\eta},\label{eq:hUphiPhiaCkest}\\
    |\bD^{k}(\phi+\Phia\ln\theta-\Phib)|_{\bh} \leq & C_{k}\theta^{-2\eta}\label{eq:phiPhialnthetaPhibCkest}
  \end{align}
  on $M_{0}$. 
\end{lemma}
\begin{remark}
  Due to (\ref{eq:varrholnthetaCl}), (\ref{eq:phiPhialnthetaPhibCkest}) can be written
  \begin{equation}\label{eq:phiPhiavarrhoPhibCkest}
    |\bD^{k}(\phi-\Phia\varrho-\Phib)|_{\bh} \leq  C_{k}\theta^{-2\eta}
  \end{equation}
  on $M_{0}$.
\end{remark}
\begin{proof}
  The statement essentially follows by appealing to the results of \cite{RinWave}. For this reason, the proof mainly consists of a verification
  that the corresponding conditions are satisfied. Let us begin by verifying that the basic assumptions, cf. \cite[Definition~3.27, p.~34]{RinWave}
  are satisfied.

  \textit{The basic assumptions.} It is clear that by restricting $(M,g)$ close enough to the big bang, it is time oriented, it has an expanding
  partial pointed foliation and $\mK$ is non-degenerate. In what follows, it is, however, important to keep in mind that in \cite{RinWave},
  $\varrho$, $\bD$ etc. are defined with respect to the reference metric $\bge_{\refer}$ introduced in \cite{RinWave} (and $\bge_{\refer}$ is typically
  different from $\bh$). We therefore here
  write $\varrho_{\refer}$ to indicate the logarithmic volume density defined with respect to $\bge_{\refer}$. On the other hand, the difference
  $\varrho-\varrho_{\refer}$ is independent of $t$. Due to \cite[Lemma~A.1, p.~201]{RinWave}, we can assume $\mK$ to have a global frame;
  this is achieved by considering a finite covering space, if necessary, an operation which does not affect the local conclusions concerning the
  asymptotics. That $\chK$ has a silent upper bound follows from Remark~\ref{remark:chKsilentupperbound}. Due to (\ref{eq:varrholnthetaCl}) and
  Proposition~\ref{prop:hmlUmKest},
  it is clear that $\hml_{U}\mK$ satisfies a weak off diagonal exponential bound; cf. \cite[Definition~3.19, p.~29]{RinWave} (in this definition,
  the requirement that the constants $C_{\mK,\rood}$, $G_{\mK,\rood}$ and $M_{\mK,\rood}$ be strictly positive is not necessary; it is sufficient if they
  are $\geq 0$). Since $\mK$ converges exponentially in any $C^{k}$-norm and
  since $\chi=0$ in our setting, it is clear that \cite[(3.28), p.~34]{RinWave} and \cite[(3.29), p.~34]{RinWave} hold. Moreover, by assumption,
  \cite[(3.18), p.~32]{RinWave} holds. Thus the basic assumptions are satisfied; cf. \cite[Definition~3.27, p.~34]{RinWave}. 

  \textit{Higher order Sobolev assumptions.} Next, let us verify that the higher order Sobolev assumptions are satisfied; cf.
  \cite[Definition~3.28, p.~34]{RinWave}. Due to the assumptions concerning the derivatives of $\ln\hN$ and the fact that $\ldr{\varrho_{\refer}}$
  and $\ldr{\ln\theta}$ are equivalent, cf. (\ref{eq:varrholnthetaCl}), it is clear that $\ln\hN$ satisfies the required bound (in what follows,
  the $\cweight$ appearing in the definition of the higher order Sobolev assumptions will be allowed to depend on the number of derivatives,
  say $l$). Since $\hU\ln N$ does not grow faster than polynomially in any $C^{k}$-norm by assumption,
  and the same is true of $\hU(\ln\theta)=-(1+q)/n$ (due to (\ref{eq:hmlUmKandqest})), it is clear that $\hU\ln\hN$ does not grow faster than
  polynomially in any $C^{k}$-norm. Combining these observations with the fact that $\chi=0$ and the fact that (\ref{eq:mKlimsfN}),
  (\ref{eq:EbfIlnthetaest}) and (\ref{eq:hmlUmKandqest}) hold, it follows that all the conditions of
  \cite[Definition~3.28, p.~34]{RinWave} are satisfied. In other words, for any choice of $l$, there is a $0\leq \cweight_{l}\in\rn{}$ such that
  the $(\cweight_{l},l)$-Sobolev assumptions are satisfied.

  \textit{Higher order $C^{k}$-assumptions.} Due to the arguments of the previous paragraph, there is, for any choice of $l$, a
  $0\leq \cweight_{l}\in\rn{}$ such that the $(\cweight_{l},l)$-supremum assumptions are satisfied; cf. \cite[Definition~3.31, p.~35]{RinWave}.

  \textit{The standard assumptions.} Since $\chi=0$, it is clear that the conditions of \cite[Lemma~3.33, p.~36]{RinWave} are satisfied.
  In particular, the standard assumptions are thus satisfied; cf. \cite[Definition~3.36, p.~36]{RinWave}. 

  \textit{Basic energy estimate.} Next, we wish to appeal to \cite[Proposition~14.24, p.~158]{RinWave}. Note, to this end, that, due to the
  above observations and the assumptions, the conditions of \cite[Lemma~7.5, p.~76]{RinWave}, \cite[Lemma~7.12, p.~81]{RinWave} and
  \cite[Lemma~7.13, p.~83]{RinWave}
  are all satisfied. Since we are here interested in the wave equation, the estimates \cite[(3.32), p.~35]{RinWave} and \cite[(3.34), p.~35]{RinWave}
  are trivially satisfied. Next, note that if $\kappa_{1}$ is the smallest integer strictly larger than $n/2+1$, then, for a suitable choice of
  $\cweight$, the $(\cweight,\kappa_{1})$-supremum assumptions are satisfied. Given $1\leq l\in\zn{}$ and a suitable choice of $\cweight_{l}$, the
  $(\cweight_{l},l)$-Sobolev assumptions are also satisfied. Since $\phi$ satisfies the wave equation, it satisfies \cite[(12.32), p.~124]{RinWave}
  with vanishing right hand side, $\hat{\mathcal{X}}^{0}=0$, $\hat{\mathcal{X}}^{A}=0$ and $\hat{\alpha}=0$.
  Due to Proposition~\ref{prop:hmlUmKest}, it is also clear that \cite[(7.78), p.~84]{RinWave} holds. Since
  $\alpha$ (using the notation of \cite[(1.1), p.~7]{RinWave}) vanishes, it is clear that \cite[(11.7), p.~110]{RinWave} holds. That
  \cite[(11.39), p.~115]{RinWave} holds is trivial due to the fact that the left hand side of this estimate vanishes identically. Summing up,
  the conditions of \cite[Proposition~14.24, p.~158]{RinWave} are satisfied. In particular, for any $1\leq l\in\zn{}$, there are constants
  $a_{l}$ and $C_{l}$ such that
  \begin{equation}\label{eq:hGlbd}
    \hG_{l}(\tau)\leq C_{l}\ldr{\tau}^{a_{l}}
  \end{equation}
  for $\tau\leq 0$. Here $\tau(t):=\varrho_{\refer}(\bx_{0},t)$ for some reference point $\bx_{0}\in\bM$. Moreover, 
  \begin{align}
    \me_{k} := & \frac{1}{2}\sum_{|\bfI|\leq k}\left(|\hU(E_{\bfI}\phi)|^{2}+\textstyle{\sum}_{A}e^{-2\mu_{A}}|X_{A}(E_{\bfI}\phi)|^{2}
    +\ldr{\tau}^{-3}|E_{\bfI}\phi|^{2}\right)\label{eq:mekdef}\\
    \hG_{k}(\tau) := & \int_{\bM_{\tau}}\me_{k}\mu_{\bge_{\refer}}.\label{eq:hGkdef}
  \end{align}
  Here $\{E_{i}\}$ is a global frame on the tangent space of $\bM$ and if $\bfI=(i_{1},\dots,i_{k})$ is such that $i_{j}\in\{1,\dots,n\}$, then
  $E_{\bfI}:=E_{i_{1}}\cdots E_{i_{k}}$ and $|\bfI|=k$. One immediate consequence of (\ref{eq:hGlbd}) is that $\phi$ does not grow faster than
  polynomially in any $C^{k}$-norm. However, this is already part of the assumptions. More interestingly, combining this estimate with
  \cite[Corollary~13.8, p.~130]{RinWave} and \cite[Remark~13.9, p.~130]{RinWave} yields the conclusion that $\hU\phi$ does not grow faster than
  polynomially in any $C^{k}$-norm. Combining this estimate with \cite[(165), p.~51]{RinGeo} and arguments similar to those of previous
  lemmas, it follows that $\hU^{2}\phi$ decays as $\theta^{-2\eta}$ in any $C^{k}$-norm. 

  Next, note that
  \[
  \d_{\tau}[\hU(\phi)]=\xi\hU^{2}(\phi),
  \]
  where $\xi:=\hN/\d_{t}\tau$. Since $\xi$ is bounded due to \cite[(7.86), p.~85]{RinWave}, and since the spatial derivatives of $\ln\hN$ do not grow
  faster than polynomially, it is clear that $\xi$ does not grow faster than polynomially in any $C^{k}$-norm. Thus $\d_{\tau}[\hU(\phi)]$
  decays exponentially in any $C^{k}$-norm. In particular, $\hU(\phi)$ converges exponentially in any $C^{k}$-norm. This guarantees the existence
  of a $\Phia\in C^{\infty}(\bM)$ such that (\ref{eq:hUphiPhiaCkest}) holds. Next, consider
  \begin{equation}\label{eq:hUphiminusPhiavarrho}
    \hU(\phi-\Phia\varrho)=\hU(\phi)-\Phia,
  \end{equation}
  where we appealed to \cite[(7.9), p.~74]{RinWave}. Since the right hand side converges to zero exponentially in every $C^{k}$-norm, we conclude
  the existence of a $\Phib\in C^{\infty}(\bM)$ such that (\ref{eq:phiPhiavarrhoPhibCkest}), and thereby (\ref{eq:phiPhialnthetaPhibCkest}), holds. 
\end{proof}

\subsection{The limits of the constraint equations in the case of a scalar field}

Next, we calculate the limits of the constraint equations in the presence of a scalar field.

\begin{lemma}\label{lemma:thelimitsoftheconstrainequations}
  Assume that the conditions of Theorem~\ref{thm:improvingasymptoticsgeneral} are satisfied. Then, using the notation of
  Theorem~\ref{thm:improvingasymptoticsgeneral}, the limits of the renormalised Hamiltonian and momentum constraints can be written
  \begin{equation}\label{eq:limitofconstraints}
    \Phia^{2}+\textstyle{\sum}_{A}p_{A}^{2}=1,\ \ \
    \mathrm{div}_{\bh}\msK=\Phia d\Phib
  \end{equation}
  respectively, where $\Phia$ and $\Phib$ are the functions whose existence is guaranteed by Lemma~\ref{lemma:scalarfieldfromgeometry}. 
\end{lemma}
\begin{proof}
  Due to \cite[(8), p.~7]{RinGeo}, the Hamiltonian constraint implies that
  \begin{equation}\label{eq:Hamiltonianconstraint}
  1=2\Omega+2\Omega_{\Lambda}+\mathrm{tr}\mK^{2}-\theta^{-2}\bS,
  \end{equation}
  where $\Omega=\rho/\theta^{2}$, $\rho$ denotes the energy density associated with the scalar field, $\Omega_{\Lambda}=\Lambda/\theta^{2}$ and $\bS$
  denotes the spatial scalar curvature. It is clear that $\Omega_{\Lambda}$ converges to zero. Moreover, due to
  \cite[(151), p.~49]{RinGeo} and the conclusions of Lemma~\ref{lemma:scalarfieldfromgeometry}, it is clear that  $2\Omega$ converges
  to $\Phia^{2}$. Due to (\ref{eq:bmSbmRexpdecay}), it is also clear that $\theta^{-2}\bS$ converges to zero. Finally, $\tr \mK^{2}$ equals
  the sum of the squares of the $\ell_{A}$. This converges to the sum of the squares of the $p_{A}$. To conclude, the limit of the expansion
  normalised Hamiltonian constraint is the first equality in (\ref{eq:limitofconstraints}). 
    
  Next, note that the momentum constraint can be written
  \[
  \bnabla_{m}\mK^{m}_{\phantom{m}j}+\mK^{m}_{\phantom{m}j}\bnabla_{m}\ln \theta-\bnabla_{j}\ln\theta=\hU(\phi)\bnabla_{j}\phi.
  \]
  Compute
  \[
  \bnabla_{A}\mK^{B}_{\phantom{B}C}=X_{A}(\mK^{B}_{\phantom{B}C})-\mK(\bnabla_{X_{A}}Y^{B},X_{C})-\mK(Y^{B},\bnabla_{X_{A}}X_{C}).
  \]
  Choosing $B=A$ and summing over $A$, the first term on the right hand side becomes $X_{C}(\ell_{C})$ (no summation). Next,
  \begin{equation*}
    \begin{split}
      -\mK(\bnabla_{X_{A}}Y^{A},X_{C}) = & -\ell_{C}(\bnabla_{X_{A}}Y^{A})(X_{C})=\ell_{C}Y^{A}(\bnabla_{X_{A}}X_{C})\\
      = & \textstyle{\sum}_{A}\ell_{C}e^{-2\bmu_{A}}\ldr{\bnabla_{X_{A}}X_{C},X_{A}}
    \end{split}
  \end{equation*}
  (no summation on $C$). However,
  \[
  \textstyle{\sum}_{A}e^{-2\bmu_{A}}\ldr{\bnabla_{X_{A}}X_{C},X_{A}}=\textstyle{\sum}_{A}e^{-2\bmu_{A}}\ldr{\lambda_{AC}^{D}X_{D}
    +\bnabla_{X_{C}}X_{A},X_{A}}=\textstyle{\sum}_{A}X_{C}(\bmu_{A})+\lambda^{A}_{AC},
  \]
  where $[X_{B},X_{C}]=\lambda_{BC}^{A}X_{A}$. To conclude
  \[
  -\mK(\bnabla_{X_{A}}Y^{A},X_{C})=\ell_{C}X_{C}\left(\textstyle{\sum}_{A}\bmu_{A}\right)+\ell_{C}\lambda^{A}_{AC}
  \]
  (no summation on $C$). Next,
  \begin{equation*}
    \begin{split}
      -\mK(Y^{A},\bnabla_{X_{A}}X_{C}) = & -\textstyle{\sum}_{A}\ell_{A}Y^{A}(\bnabla_{X_{A}}X_{C})
      = -\textstyle{\sum}_{A}\ell_{A}e^{-2\bmu_{A}}\ldr{\bnabla_{X_{A}}X_{C},X_{A}}\\
      = & -\textstyle{\sum}_{A}\ell_{A}\lambda^{A}_{AC}-\textstyle{\sum}_{A}\ell_{A}X_{C}(\bmu_{A}).
    \end{split}
  \end{equation*}
  The momentum constraint can thus be rewritten
  \begin{equation*}
    \begin{split}
      & X_{C}(\ell_{C})+\ell_{C}\lambda^{A}_{AC}-\textstyle{\sum}_{A}\ell_{A}\lambda^{A}_{AC}+(\ell_{C}-1)X_{C}(\ln\theta)\\
      & +\ell_{C}X_{C}\left(\textstyle{\sum}_{A}\bmu_{A}\right)-\textstyle{\sum}_{A}\ell_{A}X_{C}(\bmu_{A})=\hU(\phi)X_{C}(\phi)
    \end{split}
  \end{equation*}
  (no summation on $C$). The first three terms on the left hand side have finite limits. However, the same is not in general true of the last
  three terms on the left hand side, nor is it in general true of the right hand side. However, it is of interest to consider
  \begin{equation}\label{eq:FirstSingularTerm}
    \ell_{C}X_{C}\left(\textstyle{\sum}_{A}\bmu_{A}+\ln\theta\right).
  \end{equation}
  Due to (\ref{eq:bmuApAlntheta}) and the fact that the sum of the $p_{A}$ equals $1$, it is clear that the expression (\ref{eq:FirstSingularTerm})
  converges to zero exponentially. Next, consider
  \[
  \textstyle{\sum}_{A}\ell_{A}X_{C}(\bmu_{A})+\hU(\phi)X_{C}(\phi)+X_{C}(\ln\theta).
  \]
  Since $\ell_{A}$ converges to $p_{A}$ exponentially, since (\ref{eq:bmuApAlntheta}) holds, since $\hU(\phi)$ converges to $\Phia$ exponentially,
  and since $\phi+\Phia\ln\theta-\Phib$ converges to zero exponentially in any $C^{k}$-norm, this expression equals
  \begin{equation*}
    \begin{split}
      & \textstyle{\sum}_{A}p_{A}X_{C}(-p_{A}\ln\theta)+\Phia X_{C}(-\Phia\ln\theta+\Phib)+X_{C}(\ln\theta)\\
      = & \left[-\textstyle{\sum}_{A}p_{A}X_{C}(p_{A})-\Phia X_{C}(\Phia)\right]\ln\theta
      -\left[\textstyle{\sum}_{A}p_{A}^{2}+\Phia^{2}\right]X_{C}(\ln\theta)\\
      & +\Phia X_{C}(\Phib)+X_{C}(\ln\theta)=\Phia X_{C}(\Phib)
    \end{split}
  \end{equation*}
  up to terms that decay exponentially in any $C^{k}$-norm. Here we used the fact that $\textstyle{\sum}_{A}p_{A}^{2}+\Phia^{2}=1$, so that
  \[
  \textstyle{\sum}_{A}p_{A}X_{C}(p_{A})+\Phia X_{C}(\Phia)=0.
  \]
  Summing up, the expansion normalised limit of the momentum constraint is
  \[
  \msX_{C}(p_{C})+p_{C}\g^{A}_{AC}-\textstyle{\sum}_{A}p_{A}\g^{A}_{AC}=\Phia\msX_{C}(\Phib)
  \]
  (no summation on $C$), where $\g^{A}_{BC}$ are the structure constants associated with the frame $\{\msX_{A}\}$. On the other hand, an argument
  similar to the above yields
  \[
  (\mathrm{div}_{\bh}\msK)(\msX_{C})=\msX_{C}(p_{C})+p_{C}\g^{A}_{AC}-\textstyle{\sum}_{A}p_{A}\g^{A}_{AC}
  \]
  (no summation on $C$); cf. the beginning of the proof of Lemma~\ref{lemma:momcondivmsKz}. 
  The limit of the momentum constraint can thus be written $\mathrm{div}_{\bh}\msK=\Phia d\Phib$. 
\end{proof}

\subsection{Proof of Theorem~\ref{thm:improvingasymptoticsgeneral}}\label{ssection:proofpropimprovegen}
Finally, we are in a position to prove Theorem~\ref{thm:improvingasymptoticsgeneral}.

\begin{proof}[Theorem~\ref{thm:improvingasymptoticsgeneral}]
  The estimate (\ref{eq:lnthetavarrhobmuApAlntheta}) follows from (\ref{eq:varrholnthetaCl}) and (\ref{eq:bmuApAlntheta}). The estimate
  (\ref{eq:mKestrefined}) follows from (\ref{eq:mKmsYAmsXBestimate}). The estimates (\ref{eq:hmlUmKrough}) and (\ref{eq:qminusnminusoneimprprop})
  follow from (\ref{eq:hmlUmKandqest}). The estimate (\ref{eq:hmlUmKestrefined}) follows from (\ref{eq:hmlUcompbettest}). The existence of
  $\Phia$ and $\Phib$ as well as the conclusion that the asymptotic versions of the constraint equations hold is an immediate consequence
  of Lemmas~\ref{lemma:scalarfieldfromgeometry} and \ref{lemma:thelimitsoftheconstrainequations}.
\end{proof}

\section{Proof of the improvements in the case of Gaussian foliations}\label{section:impGaussianfol}

The first step in the proof of Theorem~\ref{thm:improvingasymptoticsgeneralG} is to relate $\theta$ and $t$. 

\begin{lemma}\label{lemma:lntthetatthetaminusoneestimate}
  Assume the conditions of Theorem~\ref{thm:improvingasymptoticsgeneralG} to be fulfilled. Then, using the notation of
  Theorem~\ref{thm:improvingasymptoticsgeneralG}, there is, for each $l\in\nn{}$, a constant $C_{l}$ such that 
  \begin{equation}\label{eq:tthetalimhd}
    \|\ln\theta(\cdot,t)+\ln t\|_{C^{l}(\bM)}+\|t\theta(\cdot,t)-1\|_{C^{l}(\bM)}\leq C_{l}t^{\vare}
  \end{equation}
  for all $t\leq t_{0}$. In particular, (\ref{eq:asmetricgeneralGaussiancase}), (\ref{eq:chhdef}) and (\ref{eq:limitbABtobhchhsf}) can be
  reformulated in terms of the mean curvature: there are smooth functions $\hat{b}_{AB}$, $A,B=1,\dots,n$, and, for each $l\in\nn{}$, a constant
  $C_{l}$ such that 
  \begin{align}
    g = & -dt\otimes dt+\textstyle{\sum}_{A,B}\hat{b}_{AB}\theta^{-2p_{\max\{A,B\}}}\msY^{A}\otimes \msY^{B},\label{eq:Psistargmeancurvhd}\\
    \hat{h} := & \textstyle{\sum}_{A,B}\hat{b}_{AB}\msY^{A}\otimes \msY^{B},\label{eq:hathmeancurvhd}\\
    |\bD^{l}(\hat{h}-\bh)|_{\bh} \leq & C_{l}\theta^{-\vare},\label{eq:hathlimithd}
  \end{align}
  where the last estimate holds on $M_{0}$.
\end{lemma}
\begin{remark}
  Given (\ref{eq:tthetalimhd}), the last statement of the lemma can be reversed. In other words, given (\ref{eq:tthetalimhd}),
  (\ref{eq:Psistargmeancurvhd}), (\ref{eq:hathmeancurvhd}) and (\ref{eq:hathlimithd}), it is possible to deduce that
  (\ref{eq:asmetricgeneralGaussiancase}), (\ref{eq:chhdef}) and (\ref{eq:limitbABtobhchhsf}) hold. 
\end{remark}
\begin{proof}
  Due to the Hamiltonian constraint, \cite[(8), p.~7]{RinGeo}, the equality (\ref{eq:Hamiltonianconstraint}) holds. On the other hand, the
  deceleration parameter $q$ is given by
  \begin{equation}\label{eq:qGaussianfol}
    q=n-1-\frac{n^{2}}{n-1}\frac{\rho-\bp}{\theta^{2}}-\frac{2n^{2}}{n-1}\frac{\Lambda}{\theta^{2}}+n\frac{\bS}{\theta^{2}},
  \end{equation}
  where we appealed to \cite[(266), p.~65]{RinGeo} and $\bp:=\bge^{ij}T_{ij}/n$, where $T$ is the stress energy tensor associated with the scalar
  field matter. Combining (\ref{eq:Hamiltonianconstraint}), (\ref{eq:qGaussianfol}), \cite[(151), p.~49]{RinGeo} and \cite[Remark~81, p.~49]{RinGeo}
  yields
  \begin{equation}\label{eq:qmnmone}
    q-(n-1)=-\frac{2n}{n-1}\frac{\Lambda}{\theta^{2}}+n[(\hU\phi)^{2}+\tr\mK^{2}-1].
  \end{equation}
  Next, note that since $\Phia^{2}+\tr\msK^{2}=1$, 
  \begin{equation}\label{eq:assumptoHamcon}
    (\hU\phi)^{2}+\tr\mK^{2}-1=(\hU\phi-\Phia)^{2}+2\Phia(\hU\phi-\Phia)+\tr(\mK-\msK)^{2}+2\tr[\msK(\mK-\msK)].
  \end{equation}
  In particular, for each $l\in\nn{}$, there is a constant $C_{l}$ such that 
  \[
  \|[(\hU\phi)^{2}+\tr\mK^{2}-1](\cdot,t)\|_{C^{l}(\bM)}\leq C_{l}t^{\vare}
  \]
  for $t\leq t_{0}$, due to the assumptions. Next, note that
  \begin{equation}\label{eq:hUvarrhopluslntheta}
    \hU(\varrho+\ln\theta)=-\frac{1}{n}[q-(n-1)]=\frac{2}{n-1}\frac{\Lambda}{\theta^{2}}+1-(\hU\phi)^{2}-\tr\mK^{2},
  \end{equation}
  where we appealed to \cite[(3.4), p.~26]{RinWave}, \cite[(7.9), p.~74]{RinWave} and (\ref{eq:qmnmone}). Let $\g$ be an integral curve of $\hU$
  such that $\varrho\circ\g(s)=s$; cf. (\ref{eq:varrhocges}) and the adjacent text. Then, if $\Lambda\leq 0$,
  \[
  \frac{d}{ds}[(\varrho+\ln\theta)\circ\g](s)\leq Ce^{\vare s},
  \]
  where we appealed to (\ref{eq:varrholntCl}) and the fact that $\varrho\circ\g(s)=s$. Integrating this estimate from $s_{a}$ to $s_{b}\leq 0$,
  where $s_{a}\leq s_{b}$ and $t[\g(s_{b})]\leq t_{0}$, yields
  \[
  s_{b}+\ln\theta\circ\g(s_{b})-s_{a}-\ln\theta\circ\g(s_{a})\leq Ce^{\vare s_{b}}.
  \]
  Assume now that $\g(s_{b})=(\bx,t_{0})$ and $\g(s_{a})=(\bx,t)$ with $0<t\leq t_{0}$. Then this inequality yields
  \[
  \ln\theta(\bx,t)\geq -\varrho(\bx,t)+\varrho(\bx,t_{0})+\ln\theta(\bx,t_{0})-Ce^{\vare\varrho(\bx,t_{0})}.
  \]
  Since this is true for any $\bx\in\bM$ and since (\ref{eq:varrholntCl}) holds,
  \[
  \ln\theta(\bx,t)\geq -\ln t+\ln\theta_{-}-C,
  \]
  where $C$ is allowed to depend on $t_{0}$, but not on $\bx$. Moreover, $\theta_{-}:=\inf_{\bx\in\bM}\theta(\bx,t_{0})$. Since $\theta_{-}>0$ by
  assumption, we conclude that $\ln\theta\geq -\ln t-C$ on $M_{0}$. Combining this observation with (\ref{eq:hUvarrhopluslntheta}) yields the
  conclusion that $\varrho+\ln\theta$ is bounded, and thereby the conclusion that $\ln\theta+\ln t$ is bounded. In other words, there is a
  constant $C$ such that
  \begin{equation}\label{eq:lnthpvarrholnthplnt}
    \|[\ln\theta+\varrho](\cdot,t)\|_{C^{0}(\bM)}+\|\ln\theta(\cdot,t)+\ln t\|_{C^{0}(\bM)}\leq C
  \end{equation}
  for $t\leq t_{0}$. 

  If $\Lambda>0$, we assume that $\theta\geq \theta_{0}$ on $M_{0}$, where $\theta_{0}\in\rn{}$ satisfies $\theta_{0}>[2\Lambda/(n-1)]^{1/2}$.
  Then (\ref{eq:hUvarrhopluslntheta}) can be used to deduce that $\ln\theta$ tends to infinity
  as $-\sigma\ln t$ for some $\sigma>0$. Inserting this information into (\ref{eq:hUvarrhopluslntheta}) again leads to the conclusion
  that (\ref{eq:lnthpvarrholnthplnt}) holds.

  Next, we wish to estimate $E_{\bfI}\ln\theta$; cf. the notation introduced at the beginning of the proof of Lemma~\ref{lemma:bDkvarrholnthetaest}.
  Apply, to this end, $E_{i}$ to (\ref{eq:hUvarrhopluslntheta}). This yields, recalling that $\hU=\theta^{-1}\d_{t}$ in the present setting, 
  \begin{equation}\label{eq:Eilnthetafirststep}
    \begin{split}
      & -E_{i}(\ln\theta+\varrho)\hU(\varrho+\ln\theta)+E_{i}(\varrho)\hU(\varrho+\ln\theta)+\hU[E_{i}(\varrho+\ln\theta)]\\
      = & -\frac{4}{n-1}\frac{\Lambda}{\theta^{2}}E_{i}(\varrho+\ln\theta)+\frac{4}{n-1}\frac{\Lambda}{\theta^{2}}E_{i}(\varrho)
      +E_{i}[1-(\hU\phi)^{2}-\tr\mK^{2}].
    \end{split}
  \end{equation}
  Introducing the notation $f:=E_{i}(\ln\theta+\varrho)$, this equality can schematically be written
  \[
  \hU(f)=g_{1}f+g_{2}.
  \]
  Here $\|g_{i}(\cdot,t)\|_{C^{0}(\bM)}\leq Ct^{\vare}$, where we appealed to (\ref{eq:varrholntCl}), (\ref{eq:assumptoHamcon}) and the fact that the
  right hand side of (\ref{eq:hUvarrhopluslntheta}) is $O(t^{\vare})$. Evaluating this equality on an integral curve $\g$ as above yields
  \[
  \frac{d}{ds}(f^{2}\circ\g+1)\geq -2(|g_{1}\circ\g|+|g_{2}\circ\g|)(f^{2}\circ\g+1).
  \]
  Since $g_{i}\circ\g(s)=O(e^{\vare s})$, this estimate can be integrated to yield the conclusion that $f$ is bounded on $M_{0}$. This means that
  (\ref{eq:lnthpvarrholnthplnt}) can be improved in that the $C^{0}$-norm can be replaced by the $C^{1}$-norm. Proceeding inductively, and applying
  successively higher order derivatives to (\ref{eq:Eilnthetafirststep}), similar arguments yield the conclusion that for every $l\in\nn{}$, there
  is a constant $C_{l}$ such that 
  \begin{equation}\label{eq:lnthpvarrholnthplntho}
    \|[\ln\theta+\varrho](\cdot,t)\|_{C^{l}(\bM)}+\|\ln\theta(\cdot,t)+\ln t\|_{C^{l}(\bM)}\leq C
  \end{equation}
  for all $t\leq t_{0}$.

  Next, consider $\d_{t}\theta^{-1}=-\theta^{-2}\d_{t}\theta=(q+1)/n$. Integrating this equality from $0$ to $t$, keeping in mind that $\theta^{-1}$
  vanishes at $t=0$, yields
  \[
  [\theta(\bx,t)]^{-1}=t+\int_{0}^{t}[q(\bx,s)-(n-1)]/n\ ds.
  \]
  Due to (\ref{eq:qmnmone}), (\ref{eq:lnthpvarrholnthplntho}) and the assumptions, this equality can be differentiated $l$ times in order
  to deduce that  
  \[
  \|\theta^{-1}(\cdot,t)-t\|_{C^{l}(\bM)}\leq Ct^{1+\vare}
  \]
  for $t\leq t_{0}$. The lemma follows. 
\end{proof}

\subsection{Proof of Theorem~\ref{thm:improvingasymptoticsgeneralG}}\label{ssection:proofpropimprovegenG}
Finally, we are in a position to prove Theorem~\ref{thm:improvingasymptoticsgeneralG}.

\begin{proof}[Theorem~\ref{thm:improvingasymptoticsgeneralG}]
  The idea is to verify that the conditions of Theorem~\ref{thm:improvingasymptoticsgeneral} are satisfied and to combine the
  conclusions of this theorem with an estimate relating $t$ and $\theta$ in order to obtain the desired conclusions. Due to
  Lemma~\ref{lemma:lntthetatthetaminusoneestimate}, we know that (\ref{eq:tthetalimhd}) holds and that the metric can be represented by
  (\ref{eq:Psistargmeancurvhd}). This means, in particular, that the metric can be represented as in (\ref{eq:asmetricgeneralcase})
  with $N=1$. Due to (\ref{eq:tthetalimhd}), $\theta$ diverges to $\infty$ uniformly. Moreover, due to (\ref{eq:mKlimsf}), (\ref{eq:hathlimithd})
  and (\ref{eq:tthetalimhd}), it follows that (\ref{eq:mKlimsfN}) and (\ref{eq:limitbABtobhchhsfN}) hold. Next, due to
  (\ref{eq:tthetalimhd}) and the fact that $\hN=\theta$, it follows that (\ref{eq:Crorelbd}) holds and that the first term on the
  left hand side of (\ref{eq:lnthetalnNphipolbd}) is bounded by a third times the right hand side. Since the third term on the left hand
  side vanishes identically (since $N=1$), it remains to estimate the middle term. On the other hand,
  \[
  t\d_{t}(\phi-\Phia\ln t)=t\theta\hU(\phi)-\Phia=t\theta[\hU(\phi)-\Phia]+(t\theta-1)\Phia.
  \]
  The first term on the right hand side is bounded by $C_{k}t^{\vare}$ in $C^{k}$ and the second term on the right hand side is bounded by
  $C_{k}t^{\vare}$ in $C^{k}$ due to (\ref{eq:tthetalimhd}). Integrating the corresponding estimate results in the existence of a
  $\Phib\in C^{\infty}(\bM)$ such that
  \[
  \|\phi(\cdot,t)-\Phia\ln t-\Phib\|_{C^{k}(\bM)}\leq C_{k}t^{\vare}
  \]
  for all $t\leq t_{0}$. 
  This estimate implies that the second term on the right hand side of (\ref{eq:lnthetalnNphipolbd}) satisfies the desired bound. To conclude,
  all the conditions of Theorem~\ref{thm:improvingasymptoticsgeneral} are satisfied. This means that all the conclusions hold, and due to
  (\ref{eq:tthetalimhd}), the mean curvature $\theta$ can be replaced by $t^{-1}$ in the conclusions. This yields the conclusions of
  Theorem~\ref{thm:improvingasymptoticsgeneralG}. 
\end{proof}

\section{Obtaining data on the singularity from convergent solutions}

For the remainder of the article, we assume that we have a solution which converges in the sense described at the beginning of
Subsection~\ref{ssection:contodataonsing}. The goal is then to conclude that such solutions yield data on the singularity. We
begin by deducing information concerning the asymptotics of the expansion normalised Weingarten map in the $C^{0}$-setting. 

\subsection{Asymptotics in $C^{0}$}\label{ssection:Czlim}
Let us begin by deriving limits for $\mK$ and $\hU\phi$. 

\begin{lemma}\label{lemma:mKconvtomsK}
  Let $(M,g,\phi)$ be a solution to the Einstein-scalar field equations with a cosmological constant $\Lambda$ and a uniform volume singularity
  at $t=0$; cf. Definition~\ref{def:volsing}. If (\ref{eq:CkexpdechmlUmK}) holds with $k=0$, there is a continuous $(1,1)$-tensor field $\msK$
  on $\bM$ and a constant $C$, depending only on $C_{0}$, $a_{0}$, $n$ and $\e_{0}$, such that
  \begin{equation}\label{eq:msKmKexpunifconv}
    |\mK-\msK|_{\bg_{\refer}}\leq C\ldr{\varrho}^{a_{0}}e^{2\e\varrho}
  \end{equation}
  on $M$. If, in addition, (\ref{eq:bDlhUsqphiestimate}) holds, there is, analogously, a continuous function $\Phia$ on $\bM$ such that
  \begin{equation}\label{eq:hUphiPhiaunifconv}
    |\hU\phi-\Phia|\leq C\ldr{\varrho}^{a_{0}}e^{2\e\varrho}
  \end{equation}
  on $M$.  
\end{lemma}
\begin{proof}
  Let, to begin with, $\{E_{i}\}$ be a local orthonormal frame on $T\bM$ with respect to $\bg_{\refer}$, and let $\{\omega^{i}\}$ be the dual frame. Let,
  moreover, $\mK^{i}_{j}:=\mK(\omega^{i},E_{j})$. Due to \cite[(A.4), p.~204]{RinWave} and (\ref{eq:CkexpdechmlUmK}), 
  \begin{equation}\label{eq:hUmKij}
    |\hU(\mK^{i}_{j})|=|(\hml_{U}\mK)(\omega^{i},E_{j})|\leq C_{0}\ldr{\varrho}^{a_{0}}e^{2\e\varrho}
  \end{equation}
  on $M$. Let $\g$ be an integral curve with respect to $\hU$. Then, due to \cite[(7.9), p.~74]{RinWave},
  \[
  \frac{d}{ds}(\varrho\circ\g)(s)=[\hU(\varrho)]\circ\g(s)=1.
  \]
  In particular, we can parametrise $\g$ so that $\varrho\circ\g(s)=s$. Combining this observation with (\ref{eq:hUmKij}) yields
  \[
  \left|\frac{d}{ds}(\mK^{i}_{j}\circ\g)(s)\right|\leq C_{0}\ldr{s}^{a_{0}}e^{2\e s}
  \]
  for all $s$ in the domain of definition of $\g$; note that the $\bM$-component of $\g$ is independent of $s$, so that $\e$ is evaluated at a fixed
  point of $\bM$. In particular, $\mK\circ\g$ converges to a limit and there is a $\msK$ and a constant
  $C$, depending only on $C_{0}$, $a_{0}$, $n$ and $\e_{0}$, such that (\ref{eq:msKmKexpunifconv}) holds. Since $\varrho$ converges to $-\infty$ uniformly,
  (\ref{eq:msKmKexpunifconv}) implies that $\mK(\cdot,t)$ converges uniformly to $\msK$, so that $\msK$ has to be continuous.

  The proof of (\ref{eq:hUphiPhiaunifconv}) is similar. 
\end{proof}

Let $p_{A}$ denote the eigenvalues of $\msK$. Since the eigenvalues of $\mK$ are real, and the $p_{A}$ are limits of the eigenvalues of $\mK$,
the $p_{A}$ are real. In what follows, we assume the $p_{A}$ to be distinct and order them so that $p_{1}<\cdots<p_{n}$. Then
\begin{equation}\label{eq:ellAminuspAestimateCz}
  |\ell_{A}-p_{A}|\leq C\ldr{\varrho}^{a_{0}}e^{2\e\varrho}
\end{equation}
on $M$, where $\ell_{A}$ denotes the eigenvalues of $\mK$. Since the convergence is uniform, we can restrict the interval $I$ in such a way
that the $\ell_{A}$ are distinct on $M$. In fact, the following holds. 
\begin{lemma}\label{lemma:nondegenerate}
  Let $(M,g,\phi)$ be a solution to the Einstein-scalar field equations with a cosmological constant $\Lambda$ and a
  uniform volume singularity at $t=0$. Assume that (\ref{eq:CkexpdechmlUmK}) holds with $k=0$. Assume, moreover,
  that the eigenvalues of $\msK$ (whose existence is ensured by Lemma~\ref{lemma:mKconvtomsK}) are distinct. Then, by restricting the
  interval $I$, if necessary, it can be assumed that the eigenvalues $\ell_{A}$ of $\mK$ are distinct on $M$ and that there is a constant 
  $\e_{\rond}>0$ such that the difference between the eigenvalues of $\mK$ is bounded from below by $\e_{\rond}$ on $M$. Moreover, by taking
  a finite covering space of $\bM$, if necessary, there is, for each $\ell_{A}$, a corresponding global eigenvector field $X_{A}$ of $\mK$
  with $|X_{A}|_{\bg_{\refer}}=1$. Let $\{Y^{A}\}$ be the frame dual to $\{X_{A}\}$. Then there is a frame $\{\msX_{A}\}$ of continuous eigenvector
  fields of $\msK$, with dual frame $\{\msY^{A}\}$ and a constant $C$ such that
  \begin{equation}\label{eq:XAasYAas}
    |X_{A}-\msX_{A}|_{\bg_{\refer}}+|Y^{A}-\msY^{A}|_{\bg_{\refer}}\leq C\ldr{\varrho}^{a_{0}}e^{2\e\varrho}
  \end{equation}
  on $M$.  
\end{lemma}
\begin{remark}
  In what follows, we assume the conditions of this lemma to be satisfied. Moreover, we also restrict $I$ as in the lemma, and replace $\bM$
  with an appropriate finite cover. We also order the $\ell_{A}$ so that $\ell_{1}<\cdots<\ell_{n}$. Note also that the $X_{A}$ are uniquely defined
  up to a sign (we here assume $\bM$ to be connected).
\end{remark}
\begin{proof}
  Since $\bM$ is compact and the $p_{A}$ are continuous functions on $\bM$, there is an $\e_{\rond}>0$ such that
  $\min_{A\neq B}|p_{A}-p_{B}|\geq 2\e_{\rond}$ on $\bM$. Since (\ref{eq:ellAminuspAestimateCz}) holds, where the convergence is uniform, we can
  ensure that $|\ell_{A}(\cdot,t)-p_{A}|<\e_{\rond}/2$ for $t\in I$ by restricting $I$, if necessary. Then $\min_{A\neq B}|\ell_{A}-\ell_{B}|\geq \e_{\rond}$
  on $M$. Appealing to \cite[Lemma~A.1, p.~201]{RinWave}, we can, by taking a finite covering space of $\bM$, if necessary, assume that there is a
  global frame, say $\{X_{A}\}$, of eigenvector fields of $\mK$, with $\mK X_{A}=\ell_{A}X_{A}$ (no summation), and $|X_{A}|_{\bg_{\refer}}=1$. The remaining
  statement of the lemma is an immediate consequence of the previous observations and the fact that (\ref{eq:msKmKexpunifconv}) holds. 
\end{proof}

This lemma justifies the introduction of the standard assumptions; cf. Definition~\ref{def:standardassumptions}. Next, define $\mW^{A}_{B}$ by the
condition that
\begin{equation}\label{eq:mWdef}
  \hml_{U}X_{A}=\mW^{B}_{A}X_{B}+\overline{\mW}^{0}_{A}U;
\end{equation}
cf. \cite[(6.5), p.~68]{RinWave}. Here $\hml_{U}=\theta^{-1}\ml_{U}$.

\begin{lemma}\label{lemma:rAexist}
  Assume the standard assumptions to hold, cf. Definition~\ref{def:standardassumptions}. Define $\bmu_{A}$ by the condition that
  $|X_{A}|_{\bg}=e^{\bmu_{A}}$. Then there is a constant $C$ and a continuous function $r_{A}$ on $\bM$ such that
 \begin{equation}\label{eq:rAasymptotics}
    |\bmu_{A}-p_{A}\varrho-r_{A}|\leq C\ldr{\varrho}^{a_{0}}e^{2\e\varrho}
  \end{equation}
  on $M$.  
\end{lemma}
\begin{proof}
  Combining (\ref{eq:CkexpdechmlUmK}) for $k=0$ with \cite[(6.7)--(6.8), p.~68]{RinWave}, the fact that $\min_{A\neq B}|\ell_{A}-\ell_{B}|\geq\e_{\rond}$
  on $M$, the fact that $|X_{A}|_{\bg_{\refer}}=1$ and the fact that $|Y^{A}|_{\bg_{\refer}}$ is bounded (this follows from \cite[Lemma~5.5, p.~60]{RinWave})
  implies that $|\mW^{A}_{B}|\leq C\ldr{\varrho}^{a_{0}}e^{2\e\varrho}$ on $M$. Combining this observation with
  \cite[(7.8)--(7.9), p.~74]{RinWave} yields
  \begin{equation}\label{eq:bmuAminuspAvarrho}
    \hU(\bmu_{A}-p_{A}\varrho)=\ell_{A}+\mW^{A}_{A}-p_{A}=O(\ldr{\varrho}^{a_{0}}e^{2\e\varrho})
  \end{equation}
  on $M$, where we appealed to (\ref{eq:ellAminuspAestimateCz}). At this point, we can argue as in the proof of Lemma~\ref{lemma:mKconvtomsK}
  in order to conclude that $\bmu_{A}-p_{A}\varrho$ converges exponentially to a continuous function $r_{A}$ on $\bM$. The lemma follows. 
\end{proof}

Next, we estimate $\mW^{A}_{B}$.

\begin{lemma}\label{lemma:mWABestoptimal}
  Assume the standard assumptions to be satisfied and let $\mW^{A}_{B}$ be defined by (\ref{eq:mWdef}). Then there is a constant $C$ such that
  \begin{equation}\label{eq:mWABestoptimal}
    |\mW^{A}_{B}|\leq C\ldr{\varrho}^{a_{0}}e^{2\e\varrho}\min\{1,e^{2(p_{B}-p_{A})\varrho}\}
  \end{equation}
  on $M$.
\end{lemma}
\begin{proof}
  By the argument presented at the beginning of the proof of Lemma~\ref{lemma:rAexist}, $|\mW^{A}_{B}|\leq C\ldr{\varrho}^{a_{0}}e^{2\e\varrho}$ on $M$.
  Next, for $A\neq B$, the right hand side of \cite[(6.19), p.~70]{RinWave} is antisymmetric. Thus
  \begin{equation}\label{eq:hmWABimprovedestimate}
    |\mW^{A}_{B}|\leq e^{-2(\bmu_{A}-\bmu_{B})}|\mW^{B}_{A}|\leq C\ldr{\varrho}^{a_{0}}e^{2(p_{B}-p_{A}+\e)\varrho}
  \end{equation}
  on $M$, where we appealed to (\ref{eq:rAasymptotics}). The lemma follows. 
\end{proof}

Next, let $X_{A}^{B}$ and $Y_{A}^{B}$ be defined as in Definition~\ref{def:XABYABbmuAmuA}. It is of interest to estimate the asymptotic behaviour of
$X_{A}^{B}$ and $Y^{A}_{B}$. Due to (\ref{eq:XAasYAas}), we already know that there is a constant $C$ such that
\begin{equation}\label{eq:XABmdeABYABmdeAB}
  |X_{A}^{B}-\delta_{A}^{B}|+|Y_{A}^{B}-\delta_{A}^{B}|\leq C\ldr{\varrho}^{a_{0}}e^{2\e\varrho}
\end{equation}
on $M$. However, these estimates can be improved.

\begin{lemma}\label{eq:optimalestimatesforXABYAB}
  Assume the standard assumptions to be satisfied. Then there are constants $a$ and $C$ such that
  \begin{equation}\label{eq:XABYABoptimalestimates}
    |X_{B}^{A}-\delta_{B}^{A}|+|Y_{B}^{A}-\delta_{B}^{A}|\leq C\ldr{\varrho}^{a}e^{2\e\varrho}\min\{1,e^{2(p_{B}-p_{A})\varrho}\}
  \end{equation}
  on $M$. 
\end{lemma}
\begin{proof}
  Note, to begin with, that (\ref{eq:mWdef}) yields
  \begin{equation}\label{eq:evolutionXAB}
    \hU(X_{A}^{B})\msX_{B}+X_{A}(\ln N)\hU=\mW_{A}^{B}X_{B}^{C}\msX_{C}+\overline{\mW}_{A}^{0}U.
  \end{equation}
  This equality requires a comment, since the vector fields $\msX_{A}$ can, at this stage, only be assumed to be continuous. This means that
  $X_{A}^{B}$ can only be assumed to be continuous. On the other hand, the $X_{A}^{B}$ are smooth considered as functions of $t$. For this reason,
  the time derivative of $X_{A}^{B}$ is well defined. 
  
  One particular consequence of (\ref{eq:evolutionXAB}) is that $\hU(X_{A}^{C})=\mW_{A}^{B}X_{B}^{C}$. The goal is to use this equation in order to
  improve our knowledge concerning the asymptotic behaviour of the $X_{A}$. As a first step, note that (\ref{eq:XABmdeABYABmdeAB}) holds. Let $\g$ be
  an integral curve of $\hU$ with $\varrho\circ\g(s)=s$, cf. the proof of Lemma~\ref{lemma:mKconvtomsK}. Then, if $A\geq 2$, 
  \begin{equation}\label{eq:hXoneAequation}
    \frac{d}{ds}X_{A}^{1}\circ\g = \mW_{A}^{1}\circ\g\cdot X_{1}^{1}\circ\g+\textstyle{\sum}_{B\geq 2}\mW_{A}^{B}\circ\g\cdot X_{B}^{1}\circ\g.
  \end{equation}
  The corresponding $n-1$ equations can be written
  \[
  \dot{\xi}=M\xi+\zeta,
  \]
  where $\xi$ is the vector consisting of the $n-1$ elements $X_{A}^{1}\circ\g$, $A=2,\dots,n$; $\zeta$ is the vector consisting of the $n-1$
  elements $\mW_{A}^{1}\circ\g\cdot X_{1}^{1}\circ\g$, $A=2,\dots,n$; $\|M(s)\|\leq C\ldr{s}^{a_{0}}e^{2\e s}$;
  $|\xi(s)|\leq C\ldr{s}^{a_{0}}e^{2\e s}$; and $|\zeta(s)|\leq C\ldr{s}^{a_{0}}e^{2(p_{2}-p_{1}+\e)s}$; these estimates are immediate consequences of
  (\ref{eq:mWABestoptimal}) and (\ref{eq:XABmdeABYABmdeAB}), and the constants (in these estimates and the ones below) are independent of $\g$.
  Then, for $-\infty<s_{a}\leq s\leq 0$ in the domain of definition of $\g$,
  \[
  \xi(s)=\Phi(s;s_{a})\xi(s_{a})+\int_{s_{a}}^{s}\Phi(s;\tau)\zeta(\tau)d\tau,
  \]
  where $\Phi(s;s_{a})$ is defined by $\dot{\Phi}(s;s_{a})=M(s)\Phi(s;s_{a})$ and $\Phi(s_{a};s_{a})=\mathrm{Id}$. Due to the bounds on $\|M(s)\|$,
  there is a constant $C$ such that $\|\Phi(s;s_{a})\|\leq C$ for all $-\infty<s_{a}\leq s\leq 0$ in the domain of definition of $\g$. This implies,
  in particular, that
  \[
  |\xi(s)|\leq C|\xi(s_{a})|+C\int_{s_{a}}^{s}|\zeta(\tau)|d\tau.
  \]
  Letting $s_{a}\rightarrow -\infty$ in this estimate, keeping in mind that $|\zeta(s)|\leq C\ldr{s}^{a_{0}}e^{2(p_{2}-p_{1}+\e)s}$, yields
  \begin{equation}\label{eq:xiestimate}
    |\xi(s)|\leq C\ldr{s}^{a_{0}}e^{2(p_{2}-p_{1}+\e)s}.
  \end{equation}
  Next, assume inductively that for some $2\leq m\leq n-1$, there are constants $a$ and $C$ such that 
  \begin{equation}\label{eq:XoneAindass}
    |X^{1}_{A}|\leq C\ldr{\varrho}^{a}e^{2\e\varrho}\max\{e^{2(p_{A}-p_{1})\varrho},e^{2(p_{m}-p_{1})\varrho}\}
  \end{equation}
  for all $A\geq 2$. Note that we know this inductive assumption to hold with $m=2$. Then, if $A>m$,
  \begin{equation}\label{eq:firstindstepmtomplusone}
    \frac{d}{ds}X_{A}^{1}\circ\g = \textstyle{\sum}_{B=1}^{m}\mW_{A}^{B}\circ\g\cdot X_{B}^{1}\circ\g
    +\textstyle{\sum}_{B>m}\mW_{A}^{B}\circ\g\cdot X_{B}^{1}\circ\g.
  \end{equation}
  Due to (\ref{eq:mWABestoptimal}) and (\ref{eq:XoneAindass}), the first sum on the right hand side is $O(\ldr{s}^{a}e^{2(p_{A}-p_{1}+\e)s})$.
  Viewing (\ref{eq:firstindstepmtomplusone}) as an equation for $X_{A}^{1}\circ\g$ for $A>m$, we can argue as above in order to conclude that
  the inductive hypothesis holds with $m$ replaced by $m+1$. To conclude, there are constants $a$ and $C$ such that 
  \begin{equation}\label{eq:XoneAfinaliso}
    |X^{1}_{A}|\leq C\ldr{\varrho}^{a}e^{2\e\varrho}e^{2(p_{A}-p_{1})\varrho}
  \end{equation}
  for all $A\geq 2$.

  As a next inductive step, assume that for some $1\leq m\leq n-2$, there are constants $a$ and $C$ such that 
  \begin{equation}\label{eq:XABinductiveassumption}
    |X^{A}_{B}|\leq C\ldr{\varrho}^{a}e^{2\e\varrho}e^{2(p_{B}-p_{A})\varrho}
  \end{equation}
  for all $B>A$ such that $A\leq m$. We know this to be true for $m=1$. Next, for $A>m+1$,
  \[
  \frac{d}{ds}X^{m+1}_{A}\circ\g=\textstyle{\sum}_{B\leq m+1}\mW^{B}_{A}\circ\g\cdot X^{m+1}_{B}\circ\g
  +\textstyle{\sum}_{B>m+1}\mW^{B}_{A}\circ\g\cdot X^{m+1}_{B}\circ\g.
  \]
  The first sum on the right hand side is $O(\ldr{s}^{a_{0}}e^{2(p_{A}-p_{m+1})s})$. We can thus argue as before in order to conclude that
  \[
  |X^{m+1}_{A}|\leq C\ldr{\varrho}^{a_{0}}e^{2(p_{m+2}-p_{m+1}+\e)\varrho}
  \]
  on $M$ for all $A>m+1$. Assume now, inductively, that there is a $k$, satisfying $m+1<k\leq n-1$, and constants $a$ and $C$ such that 
  \begin{equation}\label{eq:indhypmkstep}
    |X^{m+1}_{A}|\leq C\ldr{\varrho}^{a}\max\{e^{2(p_{k}-p_{m+1}+\e)\varrho},e^{2(p_{A}-p_{m+1}+\e)\varrho}\}
  \end{equation}
  on $M$ for all $A>m+1$. We know this to be true for $k=m+2$. If $A>k$,
  \[
  \frac{d}{ds}X^{m+1}_{A}\circ\g=\textstyle{\sum}_{B\leq k}\mW^{B}_{A}\circ\g\cdot X^{m+1}_{B}\circ\g
  +\textstyle{\sum}_{B>k}\mW^{B}_{A}\circ\g\cdot X^{m+1}_{B}\circ\g.
  \]
  Moreover, the first term on the right hand side is $O(\ldr{s}^{a}e^{2(p_{A}-p_{m+1}+\e)})$ and we can argue as before in order to prove that
  $k$ can be replaced by $k+1$ in the inductive hypothesis (\ref{eq:indhypmkstep}). This means that we can replace $m$ by $m+1$ in
  (\ref{eq:XABinductiveassumption}). Thus (\ref{eq:XABinductiveassumption}) holds for all $B>A$. In particular, the first term on the
  left hand side of (\ref{eq:XABYABoptimalestimates}) is bounded by the right hand side. 

  Finally, note that $Y^{A}_{B}$ are the components of the inverse of the matrix with components $X^{A}_{B}$. Combining this observation with
  the estimates for $X^{A}_{B}$ and an argument similar to the one presented at the end of the proof of Lemma~\ref{lemma:XABest} yields the
  conclusion that the second term on the left hand side of (\ref{eq:XABYABoptimalestimates}) satisfies the desired bound.
\end{proof}

At this stage, we can derive asymptotics for the metric, $\mK$ and $\hml_{U}\mK$. In other words, we can prove Theorem~\ref{thm:asymptoticformofmetric}. 

\begin{proof}[Theorem~\ref{thm:asymptoticformofmetric}]
  Note, to begin with, that
  \begin{equation}\label{eq:metricintermsofYAandbmuA}
    g=-N^{2}dt\otimes dt+\textstyle{\sum}_{A}e^{2\bmu_{A}}Y^{A}\otimes Y^{A}.
  \end{equation}
  In particular,
  \[
  g(\msX_{B},\msX_{C})=\textstyle{\sum}_{A}e^{2\bmu_{A}}Y^{A}(\msX_{B})Y^{A}(\msX_{C})
  =\textstyle{\sum}_{A}e^{2\bmu_{A}}Y^{A}_{B}Y^{A}_{C}.
  \]
  Combining this observation with (\ref{eq:rAasymptotics}) and (\ref{eq:XABYABoptimalestimates}) yields, if $A\neq B$,
  \[
  e^{-2p_{\max\{A,B\}}\varrho}|g(\msX_{A},\msX_{B})|\leq C\ldr{\varrho}^{a}e^{2\e\varrho}
  \]
  on $M$ (no summation). Moreover, if $r_{A}$ are the functions whose existence is guaranteed by Lemma~\ref{lemma:rAexist}, 
  \[
  |e^{-2p_{A}\varrho-2r_{A}}g(\msX_{A},\msX_{A})-1|\leq C\ldr{\varrho}^{a}e^{2\e\varrho}
  \]
  on $M$ (no summation). Due to these observations, it is clear that (\ref{eq:gbABgenasform}) and (\ref{eq:bABasymptotics}) hold. That
  (\ref{eq:bmuACzasymptoticsintro}) and (\ref{eq:XABYABoptimalestimatesintro}) hold follows immediately from (\ref{eq:rAasymptotics})
  and (\ref{eq:XABYABoptimalestimates}) respectively. 

  Next, note that
  \begin{equation}\label{eq:mKellAXAYAdecomp}
    \mK=\textstyle{\sum}_{A}\ell_{A}X_{A}\otimes Y^{A}.
  \end{equation}
  Thus
  \begin{equation}\label{eq:mKmsYAmsXBcomp}
    \mK(\msY^{A},\msX_{B})=\textstyle{\sum}_{C}\ell_{C}\msY^{A}(X_{C})Y^{C}(\msX_{B})
    =\textstyle{\sum}_{C}\ell_{C}X^{A}_{C}Y^{C}_{B}.
  \end{equation}
  Combining this equality with (\ref{eq:ellAminuspAestimateCz}) and (\ref{eq:XABYABoptimalestimates}) yields the conclusion that
  (\ref{eq:mKmsYAmsXBgenest}) holds. Finally, note that
  \begin{equation}\label{eq:hmlUfixedframintermediate}
    (\hml_{U}\mK)(\msY^{A},\msX_{B})=(\hml_{U}\mK)(X^{A}_{C}Y^{C},Y^{D}_{B}X_{D})=X^{A}_{C}(\hml_{U}\mK)(Y^{C},X_{D})Y^{D}_{B}.
  \end{equation}
  If $A\geq B$, then (\ref{eq:hmlUfixedframeopt}) is an immediate consequence of the fact that (\ref{eq:CkexpdechmlUmK}) holds with $k=0$.
  The interesting case is thus $A<B$. If $C=D$, then the estimate of the far right hand side of (\ref{eq:hmlUfixedframintermediate}) is
  similar to the proof of (\ref{eq:mKmsYAmsXBgenest}) for $A<B$. If $C\neq D$, the term that needs to be estimated is
  \[
  (\ell_{D}-\ell_{C})X^{A}_{C}\mW^{C}_{D}Y^{D}_{B},
  \]
  where we appealed to \cite[(6.8), p.~68]{RinWave}. Combining this observation with (\ref{eq:mWABestoptimal}) and (\ref{eq:XABYABoptimalestimates})
  yields (\ref{eq:hmlUfixedframeopt}). The lemma follows. 
\end{proof}

Next, we derive conclusions concerning the mean curvature, assuming (\ref{eq:bDlhUsqphiestimate}) to hold for $k=0$ and the limit, say $\Phia$,
of $\hU\phi$ to satisfy $\Phia^{2}+\tr\msK^{2}=1$. 

\begin{lemma}\label{lemma:lnthetavarrhoas}
  Assume the standard assumptions to be satisfied. Assume, moreover, that (\ref{eq:bDlhUsqphiestimate}) holds for $k=0$ and that the limit, say
  $\Phia$, of $\hU\phi$ satisfies $\Phia^{2}+\tr\msK^{2}=1$. Then there are constants $a$ and $C$ and a continuous function $r_{\theta}$ such that
  \[
  |\ln\theta+\varrho-r_{\theta}|\leq C\ldr{\varrho}^{a}e^{2\e\varrho}
  \]
  on $M$. 
\end{lemma}
\begin{proof}
  Since the proof is very similar to the beginning of the proof of Lemma~\ref{lemma:lntthetatthetaminusoneestimate}, combined with arguments similar
  to the proof of Lemma~\ref{lemma:mKconvtomsK}, we leave the details to the reader. 
\end{proof}

\subsection{Higher order derivatives}\label{ssection:higherorderderivativesconv}

In the previous subsection, we derive conclusions in $C^{0}$. In the present subsection, we derive estimates in $C^{k}$ for some $k\geq 1$. In order to
be able to do so, we need to impose additional assumptions; cf. Definition~\ref{def:kstandass}.

\begin{lemma}\label{lemma:kmKconvtomsK}
  Let $1\leq k\in\nn{}$ and assume the $k$-standard assumptions to hold; cf. Definition~\ref{def:kstandass}. Then the tensor field $\msK$ whose
  existence is guaranteed by Lemma~\ref{lemma:mKconvtomsK} is $C^{k}$. Moreover, there are constants $a_{k}$ and $C_{k}$ such that   
  \begin{equation}\label{eq:kmsKmKexpunifconv}
    \textstyle{\sum}_{l=0}^{k}|\bD^{l}(\mK-\msK)|_{\bg_{\refer}}\leq C_{k}\ldr{\varrho}^{a_{k}}e^{2\e\varrho}
  \end{equation}
  on $M$. In particular, if $\ell_{A}$, $X_{A}$, $Y^{A}$, $\msX_{A}$ and $\msY^{A}$ are given by Definition~\ref{lemma:nondegenerate} and $p_{A}$
  are the eigenvalues of $\msK$, ordered by size, then $\msX_{A}$, $\msY^{A}$ and $p_{A}$ are $C^{k}$ and there are constants $a_{k}$ and $C_{k}$ such that   
  \begin{equation}\label{eq:ellAXAYACkest}
    |\bD^{l}(\ell_{A}-p_{A})|_{\bg_{\refer}}+|\bD^{l}(X_{A}-\msX_{A})|_{\bg_{\refer}}
    +|\bD^{l}(Y^{A}-\msY^{A})|_{\bg_{\refer}}\leq C_{k}\ldr{\varrho}^{a_{k}}e^{2\e\varrho}
  \end{equation}
  on $M$ for all $0\leq l\leq k$ and all $A\in\{1,\dots,n\}$. 
\end{lemma}
\begin{proof}
  Let $\{E_{i}\}$ be a global orthonormal frame with respect to $\bg_{\refer}$ and let $\{\omega^{i}\}$ be the dual frame. If $\bfI=(i_{1},\dots,i_{m})$,
  where $i_{j}\in\{1,\dots,n\}$, then we use the notation $\bD_{\bfI}:=\bD_{E_{i_{1}}}\dots\bD_{E_{i_{m}}}$, $E_{\bfI}:=E_{i_{1}}\cdots E_{i_{m}}$ and $|\bfI|=m$.
  One can either check directly or verify by appealing to, e.g., \cite[Lemma~5.7, p.~62]{RinWave}, that if $|\bfI|\leq k$, then
  \begin{equation}\label{eq:bDbfIhmlUmKestimate}
    |\bD_{\bfI}\hml_{U}\mK|_{\bg_{\refer}}\leq C_{k}\ldr{\varrho}^{a_{k}}e^{2\e\varrho}
  \end{equation}
  on $M$, where we appealed to (\ref{eq:CkexpdechmlUmK}). Next, we wish to prove that for $|\bfI|\leq k$,
  \begin{equation}\label{eq:EbfIhUmKijestimate}
    |E_{\bfI}[\hU(\mK^{i}_{j})]|\leq C_{k}\ldr{\varrho}^{a_{k}}e^{2\e\varrho}
  \end{equation}
  on $M$, where $\mK^{i}_{j}=\mK(\omega^{i},E_{j})$. Note, to this end, that, due to \cite[(A.4), p.~204]{RinWave},
  \[
  \hU(\mK^{i}_{j})=(\hml_{U}\mK)(\omega^{i},E_{j}).
  \]
  It can therefore be demonstrated inductively that $E_{\bfI}[\hU(\mK^{i}_{j})]$ consists of a linear combination of terms of the form
  \[
  (\bD_{\bfJ}\hml_{U}\mK)(\bD_{\bfK}\omega^{i},\bD_{\bfL}E_{j}),
  \]
  where $|\bfJ|+|\bfK|+|\bfL|=|\bfI|$. Combining this observation with (\ref{eq:bDbfIhmlUmKestimate}) yields the conclusion that
  (\ref{eq:EbfIhUmKijestimate}) holds. Next, we wish to commute $E_{\bfI}$ and $\hU$. Note, to this end, that $[E_{\bfI},\hU]$ can be written
  as a linear combination of terms of the form
  \begin{equation}\label{eq:termsarisingfromhUEbfIcomm}
    E_{\bfI_{1}}(\ln\hN)\cdots E_{\bfI_{m}}(\ln\hN)E_{\bfJ}\hU,
  \end{equation}
  where $|\bfI_{1}|+\dots+|\bfI_{m}|+|\bfJ|=|\bfI|$, $|\bfI_{j}|\neq 0$ and $|\bfJ|<|\bfI|$. Combining this observation with (\ref{eq:EbfIhUmKijestimate})
  and the assumptions yields the conclusion that if $|\bfI|\leq k$, then
  \[
  |\hU[E_{\bfI}(\mK^{i}_{j})]|\leq C_{k}\ldr{\varrho}^{a_{k}}e^{2\e\varrho}
  \]
  on $M$. At this stage, we can argue as in the proof of Lemma~\ref{lemma:mKconvtomsK} in order to conclude that $E_{\bfI}(\mK^{i}_{j})$ converges
  at the rate $\ldr{\varrho}^{a_{k}}e^{2\e\varrho}$. Since the convergence is uniform, we conclude that the $\msK$ whose existence is guaranteed by
  Lemma~\ref{lemma:mKconvtomsK} is in fact $C^{k}$. Moreover, we obtain (\ref{eq:kmsKmKexpunifconv}). The estimate (\ref{eq:ellAXAYACkest}) is
  an immediate consequence of this. 
\end{proof}

Next, we derive information concerning the asymptotics of the scalar field, given estimates of the form (\ref{eq:bDlhUsqphiestimate}).  

\begin{lemma}\label{lemma:PhiaPhibasgivhUsqphidec}
  Assume that the standard assumptions hold; cf. Definition~\ref{def:standardassumptions}. Assume, moreover, that
  (\ref{eq:bDlhUsqphiestimate}) holds for some $k\in\nn{}$. If $k\geq 1$, assume, in addition, that (\ref{eq:bDllnhNconditions}) holds. Then
  there are $C^{k}$-functions $\Phia$ and $\Phib$ on $\bM$ and constants $a_{k}$ and $C_{k}$ such that
  \begin{equation}\label{eq:bDlhUphimPhiaPhibestimate}
    |\bD^{l}[\hU(\phi)-\Phia]|_{\bg_{\refer}}
    +|\bD^{l}(\phi-\Phia\varrho-\Phib)|_{\bg_{\refer}}\leq C_{k}\ldr{\varrho}^{a_{k}}e^{2\e\varrho}
  \end{equation}
  on $M$ for all $0\leq l\leq k$. 
\end{lemma}
\begin{remark}
  In order to derive conclusions concerning the scalar field, another possibility is to impose assumptions on the geometry and to deduce conclusions
  concerning the scalar field; cf. the proof of Lemma~\ref{lemma:scalarfieldfromgeometry}. We encourage the interested reader to provide the
  details. 
\end{remark}
\begin{proof}
  The proof of the fact that $\hU\phi$ converges to a $C^{k}$-function $\Phia$ and the fact that the first term on the left hand side of
  (\ref{eq:bDlhUphimPhiaPhibestimate}) is bounded by the right hand side is very similar to the proof of Lemma~\ref{lemma:kmKconvtomsK}.
  We therefore omit the details. On the other hand, since (\ref{eq:hUphiminusPhiavarrho}) holds, the proof of the existence of $\Phib$ and
  the fact that the second term on the left hand side of (\ref{eq:bDlhUphimPhiaPhibestimate}) is bounded by the right hand side is identical.
  The lemma follows. 
\end{proof}

Next, we prove that the functions $r_{A}$ whose existence is guaranteed by Lemma~\ref{lemma:rAexist} are $C^{k}$.

\begin{lemma}\label{lemma:rAfunctionsareCk}
  Let $1\leq k\in\nn{}$ and assume the $k$-standard assumptions to hold; cf. Definition~\ref{def:kstandass}. Then the functions $r_{A}$ whose
  existence is guaranteed by Lemma~\ref{lemma:rAexist} are $C^{k}$. Moreover, there are constants $a_{k}$ and $C_{k}$ such that
 \begin{equation}\label{eq:CkrAasymptotics}
    |\bD^{l}(\bmu_{A}-p_{A}\varrho-r_{A})|_{\bg_{\refer}}\leq C_{k}\ldr{\varrho}^{a_{k}}e^{2\e\varrho}
  \end{equation}
  on $M$ for all $A\in\{1,\dots,n\}$ and $0\leq l\leq k$. 
\end{lemma}
\begin{proof}
  Consider, to begin with, \cite[(6.7)--(6.8), p.~68]{RinWave}. Combining these formulae with the conclusions of Lemma~\ref{lemma:kmKconvtomsK} and
  the assumptions yields the conclusion that
  \[
  |\bD^{l}\mW_{A}^{B}|_{\bg_{\refer}}\leq C_{k}\ldr{\varrho}^{a_{k}}e^{2\e\varrho}
  \]
  on $M$ for $0\leq l\leq k$ and $A,B\in\{1,\dots,n\}$. Combining this with the first equality in (\ref{eq:bmuAminuspAvarrho})
  and (\ref{eq:ellAXAYACkest}) yields the conclusion that
  \[
  |E_{\bfI}\hU(\bmu_{A}-p_{A}\varrho)|\leq C_{k}\ldr{\varrho}^{a_{k}}e^{2\e\varrho}
  \]
  on $M$ for $0\leq l\leq k$ and $A\in\{1,\dots,n\}$. Due to the observations made in the proof of Lemma~\ref{lemma:kmKconvtomsK}, we can commute
  $E_{\bfI}$ and $\hU$. Integrating the resulting estimate as in the proof of Lemma~\ref{lemma:rAexist} yields the conclusion of the lemma. 
\end{proof}

In what follows, it is of interest to estimate the spatial derivatives of $\bmu_{A}$ and $\mu_{A}$.

\begin{lemma}\label{lemma:lnthetaCkasymptotics}
  Let $1\leq k\in\nn{}$ and assume the $k$-standard assumptions to hold; cf. Definition~\ref{def:kstandass}. Then there are constants
  $C_{k}$ and $a_{k}$ such that
  \begin{equation}\label{eq:bDlvarrhobDlbmuApolest}
    |\bD^{l}\varrho|_{\bg_{\refer}}+|\bD^{l}\bmu_{A}|_{\bg_{\refer}}\leq C_{k}\ldr{\varrho}^{a_{k}}
  \end{equation}
  on $M$ for all $0\leq l\leq k$ and $A\in\{1,\dots,n\}$.

  Assume, in addition to the above, that (\ref{eq:bDlhUsqphiestimate}) holds and that the limit $\Phia$, whose existence is guaranteed by
  Lemma~\ref{lemma:PhiaPhibasgivhUsqphidec}, satisfies $\Phia^{2}+\tr\msK^{2}=1$. Then the function $r_{\theta}$, whose existence is guaranteed
  by Lemma~\ref{lemma:lnthetavarrhoas}, satisfies $r_{\theta}\in C^{k}(\bM)$. Moreover, there are constants $L_{k}$ and $b_{k}$ such that
  \begin{equation}\label{eq:lnthetapvarrhoasymptotics}
    \textstyle{\sum}_{l=0}^{k}|\bD^{l}(\ln\theta+\varrho-r_{\theta})|_{\bg_{\refer}}\leq L_{k}\ldr{\varrho}^{b_{k}}e^{2\e\varrho}
  \end{equation}
  on $M$. In particular, there are constants $M_{k}$ and $c_{k}$ such that
  \begin{equation}\label{eq:muAasymptoticsprel}
    \textstyle{\sum}_{l=0}^{k}|\bD^{l}(\mu_{A}-(p_{A}-1)\varrho-r_{A}-r_{\theta})|_{\bg_{\refer}}\leq M_{k}\ldr{\varrho}^{c_{k}}e^{2\e\varrho}
  \end{equation}
  on $M$ for all $A\in\{1,\dots,n\}$.
\end{lemma}
\begin{proof}
  Due to \cite[(7.9), p.~74]{RinWave}, $\hU(\varrho)=1$. Thus, for $\bfI\neq 0$, 
  \[
  \hU(E_{\bfI}\varrho)=[\hU,E_{\bfI}]\varrho+E_{\bfI}\hU\varrho=[\hU,E_{\bfI}]\varrho.
  \]
  However, the right hand side consists of a sum of terms of the form (\ref{eq:termsarisingfromhUEbfIcomm}) applied to $\varrho$, where
  $|\bfI_{1}|+\dots+|\bfI_{m}|+|\bfJ|=|\bfI|$ and $|\bfI_{j}|\neq 0$; cf. the proof of Lemma~\ref{lemma:kmKconvtomsK}. Since $E_{\bfJ}\hU\varrho$
  equals zero if $\bfJ\neq 0$ and equals $1$ if $\bfJ=0$, $\hU(E_{\bfI}\varrho)$ does not grow faster than polynomially in $\varrho$. In particular, if
  $\g$ is an integral curve of $\hU$ with $\varrho\circ\g(s)=s$, then
  \[
  \frac{d}{ds}(E_{\bfI}\varrho)\circ\g=f_{\bfI}(s),
  \]
  where $|f_{\bfI}(s)|\leq C_{k}\ldr{s}^{a_{k}}$ on the domain of definition of $\g$, assuming $|\bfI|\leq k$, and $C_{k}$ is a constant independent of $\g$.
  In particular, it is clear that the first term on the left hand side of (\ref{eq:bDlvarrhobDlbmuApolest}) is bounded by the right hand side. Combining
  this estimate with (\ref{eq:CkrAasymptotics}) yields the conclusion that (\ref{eq:bDlvarrhobDlbmuApolest}) holds.

  The proof of (\ref{eq:lnthetapvarrhoasymptotics}) is similar to the proof of Lemma~\ref{lemma:lntthetatthetaminusoneestimate}, and we omit the
  details. Finally, (\ref{eq:muAasymptoticsprel}) is an immediate consequence of (\ref{eq:CkrAasymptotics}) and (\ref{eq:lnthetapvarrhoasymptotics}),
  keeping the fact that $\mu_{A}=\bmu_{A}+\ln\theta$ in mind; cf. the comments made in connection with \cite[(3.10)--(3.11), p.~29]{RinWave}. 
\end{proof}

At this stage, we can derive a higher order analogue of Lemma~\ref{lemma:mWABestoptimal}. 

\begin{lemma}
  Let $1\leq k\in\nn{}$ and assume the $k$-standard assumptions to hold; cf. Definition~\ref{def:kstandass}. Then there are constants
  $C_{k}$ and $a_{k}$ such that
  \begin{equation}\label{eq:bDlmWABoptimal}
    |\bD^{l}\mW^{A}_{B}|_{\bg_{\refer}}\leq C_{k}\ldr{\varrho}^{a_{k}}e^{2\e\varrho}\min\{1,e^{2(p_{B}-p_{A})\varrho}\}
  \end{equation}
  on $M$ for all $A,B\in\{1,\dots,n\}$ and all $0\leq l\leq k$. 
\end{lemma}
\begin{proof}
  Due to \cite[(6.7)--(6.8), p.~68]{RinWave}, (\ref{eq:ellAXAYACkest}) and (\ref{eq:CkexpdechmlUmK}), it is clear that
  \begin{equation}\label{eq:bDlmWABintermediateestimate}
    |\bD^{l}\mW^{A}_{B}|_{\bg_{\refer}}\leq C_{k}\ldr{\varrho}^{a_{k}}e^{2\e\varrho}
  \end{equation}
  on $M$ for all $A,B\in\{1,\dots,n\}$ and all $0\leq l\leq k$. Thus (\ref{eq:bDlmWABoptimal}) holds for $A\geq B$. In case
  $A<B$, (\ref{eq:crucialantisymmetry}) holds. Differentiating this equality, keeping (\ref{eq:rAasymptotics}), (\ref{eq:bDlvarrhobDlbmuApolest})
  and (\ref{eq:bDlmWABintermediateestimate}) in mind, yields (\ref{eq:bDlmWABoptimal}) for $A<B$. 
\end{proof}

The next goal is to derive a higher order version of Lemma~\ref{eq:optimalestimatesforXABYAB}.

\begin{lemma}\label{eq:optimalestimatesforXABYABhigherorder}
  Let $1\leq k\in\nn{}$ and assume the $k$-standard assumptions to hold; cf. Definition~\ref{def:kstandass}. Then there are constants $C_{k}$ and
  $a_{k}$ such that
  \begin{equation}\label{eq:XABYABoptimalestimateshigherorder}
    |\bD^{l}(X_{B}^{A}-\delta_{B}^{A})|_{\bge_{\refer}}+|\bD^{l}(Y_{B}^{A}-\delta_{B}^{A})|_{\bge_{\refer}}
    \leq C_{k}\ldr{\varrho}^{a_{k}}e^{2\e\varrho}\min\{1,e^{2(p_{B}-p_{A})\varrho}\}
  \end{equation}
  on $M$ for all $0\leq l\leq k$ and all $A,B\in\{1,\dots,n\}$. 
\end{lemma}
\begin{proof}
  Note, to begin with, that Lemma~\ref{lemma:kmKconvtomsK} implies that $X_{A}^{B}$ is $C^{k}$ with respect to the spatial variables and infinitely
  differentiable with respect to $t$. Moreover, due to the proof of Lemma~\ref{eq:optimalestimatesforXABYAB}, $X^{A}_{B}$ satisfies the relation
  $\hU(X_{A}^{C})=\mW_{A}^{B}X_{B}^{C}$. Introducing $Z^{A}_{B}=X^{A}_{B}-\delta^{A}_{B}$, this relation can be written
  \begin{equation}\label{eq:ZCAevolution}
    \hU(Z^{C}_{A})=\mW_{A}^{B}Z_{B}^{C}+\mW^{C}_{A}.
  \end{equation}
  Let $0\leq m<k$ and assume, inductively, that
  \begin{equation}\label{eq:hotXABinductiveassumption}
    |E_{\bfI}Z^{A}_{B}|\leq C_{m}\ldr{\varrho}^{a_{m}}e^{2\e\varrho}\min\{1,e^{2(p_{B}-p_{A})\varrho}\}
  \end{equation}
  on $M$ for all $A,B\in\{1,\dots,n\}$ and $|\bfI|\leq m$. Note that if this estimate holds, then 
  \begin{equation}\label{eq:indassconsequencehUZAB}
    |E_{\bfI}\hU Z^{A}_{B}|\leq C_{m}\ldr{\varrho}^{a_{m}}e^{2\e\varrho}\min\{1,e^{2(p_{B}-p_{A})\varrho}\}
  \end{equation}
  on $M$ for all $A,B\in\{1,\dots,n\}$ and $|\bfI|\leq m$. This is a consequence of combining (\ref{eq:hotXABinductiveassumption}) with
  (\ref{eq:bDlmWABoptimal}) and (\ref{eq:ZCAevolution}).

  Due to Lemma~\ref{eq:optimalestimatesforXABYAB}, we know that for $m=0$, the inductive hypothesis (\ref{eq:hotXABinductiveassumption}) holds.
  Next, let $\bfI$ be such that $|\bfI|=m+1$ and commute the equation (\ref{eq:ZCAevolution}) with $E_{\bfI}$. This yields
  \begin{equation}\label{eq:hUEbfIZACevolution}
    \hU(E_{\bfI}Z^{C}_{A})=[\hU,E_{\bfI}]Z^{C}_{A}+E_{\bfI}(\mW_{A}^{B}Z_{B}^{C})+E_{\bfI}(\mW^{C}_{A}).
  \end{equation}
  On the other hand, $[\hU,E_{\bfI}]$ can be written as a linear combination of terms of the form (\ref{eq:termsarisingfromhUEbfIcomm}) where
  $|\bfI_{1}|+\dots+|\bfI_{m}|+|\bfJ|=|\bfI|$, $|\bfI_{j}|\neq 0$ and $|\bfJ|<|\bfI|$. Due to this fact, the assumptions, (\ref{eq:bDlmWABoptimal}),
  (\ref{eq:hotXABinductiveassumption}) and (\ref{eq:indassconsequencehUZAB}), the equation (\ref{eq:hUEbfIZACevolution}) can be rewritten
  \[
  \hU(E_{\bfI}Z^{C}_{A})=\mW_{A}^{B}E_{\bfI}Z_{B}^{C}+F_{A}^{C},
  \]
  where
  \[
  |F_{B}^{A}|\leq C_{m}\ldr{\varrho}^{a_{m}}e^{2\e\varrho}\min\{1,e^{2(p_{B}-p_{A})\varrho}\}
  \]
  on $M$. At this stage, an argument similar to the proof of Lemma~\ref{eq:optimalestimatesforXABYAB} yields the conclusion that
  (\ref{eq:hotXABinductiveassumption}) holds with $m$ replaced by $m+1$. Thus (\ref{eq:hotXABinductiveassumption}) holds with $m=k$.
  Thus the first term on the left hand side of (\ref{eq:XABYABoptimalestimateshigherorder}) is bounded by the right hand side.

  Finally, since the $Y^{A}_{B}$ are the components of the inverse of the matrix with components $X^{A}_{B}$, an argument similar to the one
  presented at the end of the proof of Lemma~\ref{lemma:XABest} yields the conclusion that the second term on the left hand side of
  (\ref{eq:XABYABoptimalestimateshigherorder}) satisfies the desired bound.
\end{proof}

Finally, we are in a position to prove Theorem~\ref{thm:asymptoticformofmetrichod}.

\begin{proof}[Theorem~\ref{thm:asymptoticformofmetrichod}]
  Since $\{\msY^{A}\}$ is a frame for the cotangent space for the leaves of the foliation, it is clear that the metric can be written in the
  form (\ref{eq:gbABgenasformhod}). Moreover,
  \[
  b_{AB}=e^{-2p_{\max\{A,B\}}\varrho}g(\msX_{A},\msX_{B})=\textstyle{\sum}_{C}e^{2\bmu_{C}-2p_{\max\{A,B\}}\varrho}Y^{C}_{A}Y^{C}_{B},
  \]
  where we appealed to (\ref{eq:metricintermsofYAandbmuA}). Combining this equality with Lemma~\ref{lemma:rAfunctionsareCk} and
  (\ref{eq:XABYABoptimalestimateshigherorder}) yields the conclusion that (\ref{eq:bABasymptoticshod}) holds. 

  The estimates (\ref{eq:CkrAasymptoticsintro}) and (\ref{eq:XABYABoptimalestimateshigherorderintro}) follow immediately from
  (\ref{eq:CkrAasymptotics}) and (\ref{eq:XABYABoptimalestimateshigherorder}) respectively. Next, note that (\ref{eq:mKellAXAYAdecomp}) and
  (\ref{eq:mKmsYAmsXBcomp}) hold. Combining these equalities with (\ref{eq:ellAXAYACkest}) and
  (\ref{eq:XABYABoptimalestimateshigherorder}) yields (\ref{eq:mKmsYAmsXBgenesthod}). The proof of
  (\ref{eq:hmlUfixedframeopthod}) is similar to the end of the proof of Theorem~\ref{thm:asymptoticformofmetric}, given the assumptions,
  (\ref{eq:ellAXAYACkest}), (\ref{eq:bDlmWABoptimal}) and (\ref{eq:XABYABoptimalestimateshigherorder}).

  The estimate (\ref{eq:bDlhUphimPhiaPhibestimateintro}) is an immediate consequence of (\ref{eq:bDlhUphimPhiaPhibestimate}) and
  (\ref{eq:lnthetapvarrhoasymptoticsintro}) is an immediate consequence of (\ref{eq:lnthetapvarrhoasymptotics}). 
  The equality (\ref{eq:gbABgenasformtheta}) and the estimate (\ref{eq:bABasymptoticstheta}) follow by combining
  (\ref{eq:lnthetapvarrhoasymptoticsintro})
  with (\ref{eq:gbABgenasformhod}) and (\ref{eq:bABasymptoticshod}). Next, the estimates (\ref{eq:CkrAasymptoticsintrolnth}) and
  (\ref{eq:XABYABoptimalestimateshigherorderintrolnth}) follow from (\ref{eq:CkrAasymptoticsintro}) and
  (\ref{eq:XABYABoptimalestimateshigherorderintro}) respectively. The estimate (\ref{eq:muAasymptoticsprelintro}) follows from
  (\ref{eq:muAasymptoticsprel}). Next, (\ref{eq:mKmsYAmsXBgenesttheta}) and (\ref{eq:hmlUfixedframeopttheta})
  follow from (\ref{eq:lnthetapvarrhoasymptotics}), (\ref{eq:mKmsYAmsXBgenesthod}) and (\ref{eq:hmlUfixedframeopthod}). The estimate
  (\ref{eq:hUlnthetaasymptoticsintermsoftheta}) follows from (\ref{eq:bDlhUphimPhiaPhibestimateintro}), (\ref{eq:qmnmone}), (\ref{eq:assumptoHamcon}),
  (\ref{eq:kmsKmKexpunifconv}), (\ref{eq:bDlvarrhobDlbmuApolest}) and (\ref{eq:lnthetapvarrhoasymptotics}).

  Next, assume $N=1$ and compute
  \begin{equation}\label{eq:dtthetainv}
    \d_{t}\theta^{-1}=-\theta^{-1}\d_{t}\ln\theta=-n^{-1}\hU(n\ln\theta)=n^{-1}(q+1)=1+n^{-1}[q-(n-1)].
  \end{equation}
  Since the right hand side converges uniformly to $1$ as $t\rightarrow 0+$, there is a $T>0$ such that on $(0,T]$, $1/2\leq  \d_{t}\theta\leq 2$
  on $\bM\times (0,T]$. Since $\theta\rightarrow\infty$ uniformly as $t\rightarrow 0+$, we conclude that $t/2\leq \theta^{-1}\leq 2t$ on $(0,T]$.
  Combining this observation with (\ref{eq:dtthetainv}) and (\ref{eq:hUlnthetaasymptoticsintermsoftheta}) yields the
  conclusion that there are constants $C_{k}$ and $a_{k}$ such that 
  \[
  |\bD^{l}(\d_{t}\theta^{-1}-1)|_{\bg_{\refer}}\leq C_{k}\ldr{\ln t}^{a_{k}}t^{2\e}
  \]
  on $M$ for $0\leq l\leq k$. Integrating this estimate yields (\ref{eq:thetaasymptoticsGaussianfoliation}). Combining
  (\ref{eq:thetaasymptoticsGaussianfoliation}) with (\ref{eq:gbABgenasformtheta}) and (\ref{eq:bABasymptoticstheta}) yields
  (\ref{eq:gbABgenasformGauss}) and (\ref{eq:bABasymptoticsGauss}). 
\end{proof}

\subsection{Reproducing the conditions on data on the singularity}\label{ssection:reproducingcondonsing}
Next, we prove Theorem~\ref{thm:reprod}.

\begin{proof}[Theorem~\ref{thm:reprod}]
  Recall, to begin with, the Hamiltonian constraint (\ref{eq:Hamiltonianconstraint}). Next, 
  \[
  2\Omega+\tr\mK^{2}=(\hU\phi)^{2}+\tr\mK^{2}+\textstyle{\sum}_{A}e^{-2\mu_{A}}(X_{A}\phi)^{2};
  \]
  see \cite[(151), p.~49]{RinGeo}. Combining these equalities yields
  \begin{equation}\label{eq:provingthetamtwobSexpdec}
    \theta^{-2}\bS=(\hU\phi)^{2}+\tr\mK^{2}-1+\textstyle{\sum}_{A}e^{-2\mu_{A}}(X_{A}\phi)^{2}+2\Omega_{\Lambda}.
  \end{equation}
  Our first goal is to prove that the right hand side of (\ref{eq:provingthetamtwobSexpdec}) converges to zero. Due to
  Lemma~\ref{lemma:lnthetavarrhoas}, it is clear that $\Omega_{\Lambda}$ decays as $e^{2\varrho}$. 
  Due to (\ref{eq:XAasYAas}), (\ref{eq:bDlhUphimPhiaPhibestimate}) and (\ref{eq:bDlvarrhobDlbmuApolest}), we know that $X_{A}\phi$ does not grow
  faster than polynomially in $\varrho$. Next, note that since the sum of the $p_{A}^{2}$ is bounded from above by $1$, since there is an
  $\e_{\rond}>0$ such that $|p_{A}-p_{B}|\geq\e_{\rond}$, and since $n\geq 3$, there is an $\e_{\Spe}>0$ such that $p_{A}\leq 1-\e_{\Spe}$. Combining this
  observation with (\ref{eq:muAasymptoticsprel}) yields the conclusion that there is a constant $C$ such that
  \begin{equation}\label{eq:muAlowbdrepr}
    \mu_{A}\geq -\e_{\Spe}\varrho-C.
  \end{equation}
  Due to (\ref{eq:msKmKexpunifconv}), (\ref{eq:bDlhUphimPhiaPhibestimate}), the fact that $\bPhi_{a}^{2}+\tr\msK^{2}=1$ and the above observations,
  there is an $\eta>0$ such that the right hand side of (\ref{eq:provingthetamtwobSexpdec}) decays as $e^{2\eta\varrho}$. In other words, there is a
  constant $C$ such that 
  \begin{equation}\label{eq:thetamtwobSexpdecay}
    |\theta^{-2}\bS|\leq Ce^{2\eta\varrho}.
  \end{equation}  
  Next, define $\lambda^{A}_{BC}$ by $[X_{B},X_{C}]=\lambda_{BC}^{A}X_{A}$. Due to (\ref{eq:ellAXAYACkest}), it is clear that $\lambda^{A}_{BC}$ and
  $X_{D}(\lambda^{A}_{BC})$ converge in $C^{0}$. Due to (\ref{eq:ellAXAYACkest}) and (\ref{eq:bDlvarrhobDlbmuApolest}), it is clear that $X_{A}(\bmu_{B})$
  and $X_{B}X_{B}(\bmu_{C})$ do not grow faster than polynomially. On the other hand, $e^{-\mu_{A}}$ decays to zero exponentially due to
  (\ref{eq:muAlowbdrepr}). Combining these observations with the formula \cite[(279), p.~69]{RinGeo} (with $\g^{A}_{BC}$ replaced by $\lambda^{A}_{BC}$)
  yields
  \begin{equation}\label{eq:bSexpnormasformula}
    \left|\theta^{-2}\bS+\textstyle{\frac{1}{4}\sum}_{A,B,C}e^{2\mu_{B}-2\mu_{A}-2\mu_{C}}(\lambda^{B}_{AC})^{2}\right|\leq C\ldr{\varrho}^{a}e^{2\e_{\Spe}\varrho}
  \end{equation}
  on $M$. Combining this estimate with (\ref{eq:thetamtwobSexpdecay}) yields the conclusion that there is an $\eta>0$ and a $C$ such that
  \begin{equation}\label{eq:sumstructureconstsq}
    \textstyle{\sum}_{A,B,C}e^{2\mu_{A}-2\mu_{B}-2\mu_{C}}(\lambda^{A}_{BC})^{2}\leq Ce^{2\eta\varrho}
  \end{equation}
  on $M$. Assume now that $\bx\in\bM$ is such that $1+p_{A}(\bx)-p_{B}(\bx)-p_{C}(\bx)\leq 0$. Then, due to the continuity of the $p_{A}$, there is
  an open neighbourhood $V$ of $\bx$ such that $1+p_{A}-p_{B}-p_{C}\leq \eta/2$ in $V$. On the other hand, due to (\ref{eq:muAasymptoticsprel}),
  \[
  2\mu_{A}-2\mu_{B}-2\mu_{C}=2(1+p_{A}-p_{B}-p_{C})\varrho+O(1)\geq \eta\varrho-C
  \]
  on $V$, for some constant $C$. Combining this estimate with (\ref{eq:sumstructureconstsq}) yields
  \[
  (\lambda^{A}_{BC})^{2}\leq Ce^{\eta\varrho}
  \]
  on $V$. This means that $\lambda^{A}_{BC}$ converges to zero on $V$. On the other hand, $\lambda^{A}_{BC}$ converges to $\g^{A}_{BC}$ due to
  (\ref{eq:ellAXAYACkest}). This proves the first statement of the theorem; note that $A$, $B$ and $C$ have to be distinct in order for the
  statement not to be void, in which case $\msY^{A}([\msX_{B},\msX_{C}])=0$ is equivalent to $\bmsY^{A}([\bmsX_{B},\bmsX_{C}])=0$. 
  
  Next, note that $\tr\msK=1$ is an immediate consequence of the fact that $\mK\rightarrow\msK$ and the fact that $\tr\mK=1$. The fact that
  $\msK$ is symmetric with respect to $\bh$ is an immediate consequence of the fact that there is a frame of eigenvector fields of $\msK$ which is
  orthonormal with respect to $\bh$. The relation $\mathrm{tr}\msK^{2}+\bPhi_{a}^{2}=1$ holds by assumption. That
  (\ref{eq:mKmsKconvergencehhbhconvergence}) holds is an immediate consequence of (\ref{eq:kmsKmKexpunifconv}), (\ref{eq:lnthetapvarrhoasymptotics}),
  (\ref{eq:bABasymptoticstheta}) and the definition of $\chh$; cf. (\ref{eq:chhdeffinal}). That (\ref{eq:bDlhUphimPhiaPhibestimatereprodver}) holds
  is an immediate consequence of (\ref{eq:bDlhUphimPhiaPhibestimatereprodverprel}). What remains to be demonstrated is that
  $\mathrm{div}_{\bh}\msK=\bPhi_{a}d\bPhi_{b}$. The template for doing so is the proof of Lemma~\ref{lemma:thelimitsoftheconstrainequations}. However,
  the definition of $X_{A}$ etc. used in this proof differs from the definition of the present section. Let $\bX_{A}=e^{-\br_{A}}X_{A}$ (no summation)
  and define $e^{\chmu_{A}}$ by $|\bX_{A}|_{\bge}=e^{\chmu_{A}}$. Then $\{\bX_{A}\}$ is a $C^{k}$-frame with dual $C^{k}$ frame $\{\bY^{A}\}$. Note also that
  $\{\bmsX_{A}\}$ satisfies the conditions stated in Remark~\ref{remark:msOdefsf}, except for the regularity; the vector fields
  $\bmsX_{A}$ are $C^{k}$. Next, since $\br_{A}=r_{A}+p_{A}r_{\theta}$,
  \begin{equation*}
    \begin{split}
      \chmu_{A}+p_{A}\ln\theta = & \bmu_{A}-\br_{A}+p_{A}\ln\theta
      =(\bmu_{A}-p_{A}\varrho-r_{A})+p_{A}\varrho+r_{A}-r_{A}-p_{A}r_{\theta}+p_{A}\ln\theta\\
      = & (\bmu_{A}-p_{A}\varrho-r_{A})+p_{A}(\varrho+\ln\theta-r_{\theta}).
    \end{split}
  \end{equation*}
  Combining this equality with (\ref{eq:CkrAasymptotics}) and (\ref{eq:lnthetapvarrhoasymptotics}) yields
  \[
  |\bD^{l}(\chmu_{A}+p_{A}\ln\theta)|_{\bg_{\refer}}\leq C_{k}\ldr{\varrho}^{a_{k}}e^{2\e\varrho}
  \]
  on $M$ for all $A$ and all $0\leq l\leq k$. Given the above estimates, we can revisit the proof of
  Lemma~\ref{lemma:thelimitsoftheconstrainequations}, with suitable substitutions: $X_{A}$ and $Y^{A}$ are replaced by $\bX_{A}$ and $\bY^{A}$
  respectively; $\bmu_{A}$ is replaced by $\chmu_{A}$; $\Phia$ and $\Phib$ are replaced by $\bPhi_{a}$ and $\bPhi_{b}$ respectively; the
  structure constants are calculated with respect to $\bX_{A}$ etc. The conclusion is that $\mathrm{div}_{\bh}\msK=\bPhi_{a}d\bPhi_{b}$. The
  theorem follows. 
\end{proof}

Finally, we are in a position to prove Theorem~\ref{thm:finaloriginalsetting}.

\begin{proof}[Theorem~\ref{thm:finaloriginalsetting}]
  Note that the conditions of Theorems~\ref{thm:asymptoticformofmetrichod} and \ref{thm:reprod} are fulfilled. In particular,
  it thus follows that the metric can be represented as in (\ref{eq:gbABgenasformGauss}). The corresponding $\chh$ is given by
  \[
  \chh:=\textstyle{\sum}_{A,B}c_{AB}\msY^{A}\otimes\msY^{B}.
  \]
  Due to (\ref{eq:cABdef}), it follows that the limit of $\chh$ is the $\bh$ appearing in the statement of Theorem~\ref{thm:reprod}; in fact,
  \[
  \textstyle{\sum}_{l=0}^{k}|\bD^{l}(\hh-\bh)|_{\bg_{\refer}} \leq C_{k}\ldr{\ln t}^{a_{k}}t^{2\e}.
  \]
  Moreover, the limit of $\mK$, i.e. $\msK$, is of course the same. This means that $\rodiv_{\bh}\msK=0$, that $\tr\msK^{2}=1$, $\tr\msK=1$ and that
  $\msK$ is symmetric with respect to $\bh$. Moreover, $\msY^{1}[\msX_{2},\msX_{3}]=0$. Assuming $\msK$ to be $C^{\infty}$,
  Lemma~\ref{lemma:localcoordinates} 
  yields coordinates $(V,\bsfx)$ such that the conclusions of Lemma~\ref{lemma:localcoordinates} hold; denote the renormalised bases by
  $\{\hmsX_{A}\}$ and $\{\hmsY^{A}\}$. Then $\msY^{A}=\b_{A}\hmsY^{A}$ (no summation) for some positive functions $\b_{A}$. Moreover, 
  \begin{equation*}
    \begin{split}
      \textstyle{\sum}_{A,B}c_{AB}t^{2p_{\max\{A,B\}}}\msY^{A}\otimes \msY^{B} = & \textstyle{\sum}_{A,B}\bc_{AB}t^{2p_{\max\{A,B\}}}\hmsY^{A}\otimes \hmsY^{B}\\
      = & \textstyle{\sum}_{i,j}d_{ij}t^{2p_{\max\{i,j\}}}d\bsfx^{i}\otimes d\bsfx^{j},
    \end{split}
  \end{equation*}  
  where $\bc_{AB}=\beta_{A}c_{AB}\beta_{B}$ (no summation),
  \[
  \textstyle{\sum}_{l=0}^{k}|\bD^{l}(\bc_{AB}-\bc_{A}^{2}\delta_{AB})|\leq C\ldr{\ln t}^{a_{k}}t^{2\e}
  \]
  and $\bc_{A}=\b_{A}e^{\bar{r}_{A}}$ (no summation); cf. (\ref{eq:bABasymptoticsGauss}). Moreover, it can be deduced that
  (\ref{eq:dijlimits}) holds on $W$, where $\bd_{ii}=\bc_{i}^{2}>0$. Combining this observation with Lemma~\ref{lemma:localcoordinates} yields
  (\ref{eq:bdijrelations}). Next, note that (\ref{eq:mKmsYAmsXBgenesttheta}) can be written
  \begin{equation}\label{eq:mKmmsKintermediateestimate}
    \textstyle{\sum}_{l=0}^{k}|\bD^{l}[\mK(\msY^{A},\msX_{B})-\msK(\msY^{A},\msX_{B})]| \leq C_{k}\ldr{\ln t}^{a_{k}}t^{2\e}\min\{1,t^{2(p_{B}-p_{A})}\}.
  \end{equation}
  One immediate consequence of this estimate is that 
  \[
  \textstyle{\sum}_{l=0}^{k}|\bD^{l}[\mK^{i}_{j}-\msK^{i}_{j}]|_{\bge_{\refer}} \leq C_{k}\ldr{\ln t}^{a_{k}}t^{2\e}
  \]
  on $W$ for $i\geq j$, where $\mK^{i}_{j}=\mK(d\bsfx^{i},\d_{j})$. On the other hand,
  \[
  \mK(d\bsfx^{1},\d_{2})=\mK(\hmsY^{1},\hmsX_{2}-\hmsX^{3}_{2}\hmsX_{3})=O_{k}(\ldr{\ln t}^{a_{k}}t^{2(p_{2}-p_{1}+\e)})
  \]
  on $W$, where we appealed to (\ref{eq:mKmmsKintermediateestimate}) and Lemma~\ref{lemma:localcoordinates}, and the notation $O_{k}$
  signifies that the relevant estimates hold for up to $k$ derivatives. The arguments concerning the remaining components are similar and
  the estimate
  \begin{equation}\label{eq:mKijmsKijestimate}
    \textstyle{\sum}_{l=0}^{k}|\bD^{l}[\mK^{i}_{j}-\msK^{i}_{j}]|_{\bge_{\refer}} \leq C_{k}\ldr{\ln t}^{a_{k}}t^{2\e}\min\{1,t^{2(p_{j}-p_{i})}\}
  \end{equation}
  follows (on $W$). By a similar argument, appealing to (\ref{eq:hmlUfixedframeopttheta}), 
  \begin{equation}\label{eq:hmlUmKijestimate}
    \textstyle{\sum}_{l=0}^{k}|\bD^{l}[(\hml_{U}\mK)^{i}_{j}]| \leq C_{k}\ldr{\ln t}^{a_{k}}t^{2\e}\min\{1,t^{2(p_{j}-p_{i})}\}
  \end{equation}
  on $W$. Combining these estimates with (\ref{eq:thetaasymptoticsGaussianfoliation}), it can be deduced that (\ref{eq:bKijasesthmlUbKijasest}) holds. 

  Finally, combining the fact that $\rodiv_{\bh}\msK=0$ with the arguments presented in the proof of Lemma~\ref{lemma:momcondivmsKz} yields the
  final conclusion of the lemma. 
\end{proof}


\begin{thebibliography}{1}
\bibitem{ABIF} Ames, E.; Beyer, F.; Isenberg, J.; LeFloch, P. G.: Quasilinear hyperbolic Fuchsian systems and AVTD
behavior in $T^2$-symmetric vacuum spacetimes. Ann. Henri Poincar\'{e} {\bf 14}, no. 6, 1445--1523 (2013)
\bibitem{aeta} Ames, E.; Beyer, F.; Isenberg, J.; LeFloch, P. G.: A class of solutions to the Einstein equations with AVTD behavior
  in generalized wave gauges. J. Geom. Phys. {\bf 121}, 42--71 (2017)
\bibitem{aaf} Andersson, L.; Fajman, D.: Nonlinear stability of the Milne model with matter. Comm. Math. Phys. {\bf 378}, no. 1, 261--298 (2020)   
\bibitem{aammilne} Andersson, L.; Moncrief, V.: Future complete vacuum spacetimes. The Einstein equations and the large scale behavior of
  gravitational fields, 299--330, Birkh\"{a}user, Basel (2004)
\bibitem{aameflow} Andersson, L.; Moncrief, V.: Einstein spaces as attractors for the Einstein flow. J. Differential Geom.
  {\bf 89}, no. 1, 1--47 (2011)
\bibitem{aarendall} Andersson, L.; Rendall, A. D.: Quiescent cosmological singularities. Comm. Math. Phys. {\bf 218}, 479--511 (2001)
\bibitem{AAR} Andr\'{e}asson; H.; Ringstr\"{o}m, H.: Proof of the cosmic no-hair conjecture in the $\mathbb{T}^{3}$-Gowdy symmetric Einstein-Vlasov 
setting. J. Eur. Math. Soc. (JEMS) \textbf{18}, no. 7, 1565--1650 (2016)  
\bibitem{beguin} B\'{e}guin, F.: Aperiodic oscillatory asymptotic behavior for some 
  Bianchi spacetimes. Class. Quantum Grav. {\bf 27}, 185005 (2010)
\bibitem{du} Béguin, F.; Dutilleul, T.: Chaotic dynamics of spatially homogeneous spacetimes. Comm. Math. Phys. {\bf 399}, no.2, 737--927 (2023)
\bibitem{bkl1} Belinski\v{\i}, V. A.; Khalatnikov, I. M.; Lifshitz, E. M.: Oscillatory approach to a
singular point in the relativistic cosmology. Adv. Phys. {\bf 19}, 525--573 (1970)
\bibitem{bkl15} Belinskii, V. A.; Khalatnikov, I. M.: Effect of scalar and vector fields on the nature of the cosmological
singularity. Sov. Phys. JETP {\bf 36}, 591--597 (1973)
\bibitem{bkl2} Belinski\v{\i}, V. A.; Khalatnikov, I. M.; Lifshitz, E. M.: A general solution of the
Einstein equations with a time singularity. Adv. Phys. {\bf 31}, 639--667 (1982)
\bibitem{bieri} Bieri, L.: Part I: Solutions of the Einstein vacuum equations. Extensions of the stability theorem of the Minkowski space
  in general relativity, 1–295, AMS/IP Stud. Adv. Math., {\bf 45}, Amer. Math. Soc., Providence, RI, (2009)
\bibitem{baz} Bieri, L.; Zipser, N.: Extensions of the stability theorem of the Minkowski space in general relativity. AMS/IP Studies in Advanced
  Mathematics, {\bf 45}. Amer. Math. Soc., Providence, RI; International Press, Cambridge, MA, xxiv+491 pp. (2009)
\bibitem{betal} Bigorgne, L.; Fajman, D.; Joudioux, J.; Smulevici, J.; Thaller, M.: Asymptotic stability of Minkowski space-time with
  non-compactly supported massless Vlasov matter. Arch. Ration. Mech. Anal. {\bf 242}, no. 1, 1--147 (2021)
\bibitem{brehm} Brehm, B.: Bianchi VIII and IX vacuum cosmologies: Almost every solution forms particle horizons and converges 
to the Mixmaster attractor. Preprint, \href{https://arxiv.org/abs/1606.08058}{arXiv:1606.08058}
\bibitem{cah10} Calogero, S.; Heinzle, J. M.: Oscillations toward the singularity of locally rotationally symmetric Bianchi type IX cosmological
  models with Vlasov matter. SIAM J. Appl. Dyn. Syst. {\bf 9}, no. 4, 1244--1262 (2010)
\bibitem{cah11} Calogero, S.; Heinzle, J. M.: Bianchi cosmologies with anisotropic matter: locally rotationally symmetric models. Phys. D {\bf 240},
  no. 7, 636--669 (2011)
\bibitem{cak} Christodoulou, D.; Klainerman, S.: The global nonlinear stability of the Minkowski space. Princeton Mathematical Series, 41.
  Princeton University Press, Princeton, NJ. x+514 pp. (1993)
\bibitem{dafetal} Dafermos, M.; Holzegel, G.; Rodnianski, I.; Taylor, M.: The non-linear stability of the Schwarzschild family of black holes.
  Preprint, \href{https://arxiv.org/abs/2104.08222}{arXiv:2104.08222}  
\bibitem{dal} Dafermos, M.; Luk, J.: The interior of dynamical vacuum black holes I: The C0-stability of the Kerr Cauchy horizon.
  Preprint, \href{https://arxiv.org/abs/1710.01722}{arXiv:1710.01722}  
\bibitem{dadB} Damour, T.; de Buyl, S.: Describing general cosmological singularities in Iwasawa variables. Phys. Rev. D {\bf 77}, no. 4,
  043520, 26 pp. (2008)
\bibitem{damouretal} Damour, T.; Henneaux, M.; Rendall, A. D.; Weaver, M.: 
Kasner-like behaviour for subcritical Einstein-matter systems. 
Ann. Henri Poincar\'{e} {\bf 3}, 1049--1111 (2002)
\bibitem{dhn} Damour, T.; Henneaux, M.; Nicolai, H.: Cosmological billiards. Class. Quantum Grav. {\bf 20}, R145--R200 (2003)
\bibitem{dah} Damour, T.; Nicolai, H.: Higher order M theory corrections and the Kac-Moody
algebra E10. Class. Quantum Grav. {\bf 22}, 2849--2879 (2005)
\bibitem{Henneauxetal} Demaret, J.; Henneaux, M.; Spindel, P.: Nonoscillatory Behavior In Vacuum Kaluza-Klein Cosmologies. 
Phys. Lett. {\bf 164B}, 27--30 (1985)
%\bibitem{du} Dutilleul, T.: Chaotic dynamics of spatially homogeneous spacetimes. Universit\'{e} Paris 13 - Sorbonne Paris Cit\'{e},
%  2019. \href{https://tel.archives-ouvertes.fr/tel-02488655}{https://tel.archives-ouvertes.fr/tel-02488655}  
\bibitem{ELS} Eardley, D; Liang, E.; Sachs, R.: Velocity-dominated singularities in irrotational dust cosmologies. J. Math. Phys. {\bf 13},
  99--106 (1972)
\bibitem{fjs} Fajman, D.; Joudioux, J.; Smulevici, J.: The stability of the Minkowski space for the Einstein-Vlasov system.
  Anal. PDE {\bf 14}, no. 2, 425--531 (2021)
\bibitem{fow} Fajman, D.; Oliynyk, T. A.; Wyatt, Z.: Stabilizing relativistic fluids on spacetimes with non-accelerated expansion.
  Comm. Math. Phys. {\bf 383}, no. 1, 401--426 (2021)
\bibitem{faw} Fajman, D.; Wyatt, Z.: Attractors of the Einstein-Klein-Gordon system. Comm. Partial Differential Equations
  {\bf 46}, no. 1, 1--30 (2021)
\bibitem{fal} Fournodavlos, G. and Luk, J.: Asymptotically Kasner-like singularities. Amer. J. Math. {\bf 145}, no.4, 1183--1272 (2023)
\bibitem{GIJ} Fournodavlos, G.; Rodnianski, I.; Speck, J.: Stable Big Bang Formation for Einstein’s Equations: The Complete
Sub-critical Regime. J. Amer. Math. Soc. {\bf 36}, no.3, 827--916 (2023)
\bibitem{heflambda} Friedrich, H.: Existence and structure of past asymptotically simple solutions of Einstein's field equations
  with positive cosmological constant. J. Geom. Phys. {\bf 3}, no. 1, 101--117 (1986)
\bibitem{hefmink} Friedrich, H.: On the existence of n-geodesically complete or future complete solutions of Einstein's field equations
  with smooth asymptotic structure. Comm. Math. Phys. {\bf 107}, no. 4, 587--609 (1986)
\bibitem{hefYM} Friedrich, H.: On the global existence and the asymptotic behavior of solutions to the Einstein-Maxwell-Yang-Mills equations.
  J. Differential Geom. {\bf 34}, no. 2, 275--345 (1991)
\bibitem{grisales} Grisales, A. F.: Solutions to the Einstein-scalar field equations with initial data on the singularity. Preprint,
  \href{https://arxiv.org/abs/2409.17065}{arXiv:2409.17065}
\bibitem{haz} Had\v{z}i\'{c}, M.; Speck, J.: The global future stability of the FLRW solutions to the dust-Einstein system with a
  positive cosmological constant. J. Hyperbolic Differ. Equ. {\bf 12}, no. 1, 87--188 (2015)
\bibitem{HU} Heinzle, J. M.; Uggla, C.: A new proof of the Bianchi type IX attractor theorem. Class. Quantum Grav. {\bf 26}, no. 7, 075015 (2009)
\bibitem{huar} Heinzle, J. M.; Uggla, C.; Niklas Rohr, N.: The cosmological billiard attractor.
Adv. Theor. Math. Phys. {\bf 13}, 293--407 (2009)
\bibitem{HUL} Heinzle, J. M.; Uggla, C.; Lim, W. C.: Spike oscillations. Phys. Rev. D {\bf 86}, 104049 (2012)    
\bibitem{havKdS} Hintz, Peter; Vasy, András The global non-linear stability of the Kerr–de Sitter family of black holes.
  Acta Math. {\bf 220}, no. 1, 1--206 (2018)
\bibitem{havmink} Hintz, P.; Vasy, A.: Stability of Minkowski space and polyhomogeneity of the metric. Ann. PDE {\bf 6},
  no. 1, Paper No. 2, 146 pp. (2020)
\bibitem{iak} Isenberg, J.; Kichenassamy, S.: Asymptotic behavior in polarized $T^{2}$-symmetric vacuum space-times. J. Math. Phys. {\bf 40},
  no. 1, 340--352 (1999)
\bibitem{JimandVince} Isenberg, J.; Moncrief, V.: Asymptotic behavior of the gravitational field and the nature of singularities in Gowdy spacetimes.
  Ann. Physics {\bf 199}, no. 1, 84--122 (1990)
\bibitem{ial} Isenberg, J.; Luo, X.: Power law inflation with electromagnetism. Ann. Phys. {\bf 334}, 420--454 (2013)
\bibitem{iam} Isenberg, J.; Moncrief, V.: Asymptotic behaviour in polarized and half-polarized $U(1)$ symmetric vacuum spacetimes.
  Class. Quantum Grav. {\bf 19}, no. 21, 5361--5386 (2002)
\bibitem{kar} Kichenassamy, S.; Rendall, A.: Analytic description of singularities in Gowdy spacetimes. Class. Quantum Grav. {\bf 15}, no. 5,
  1339--1355 (1998)
\bibitem{kas} Klainerman, S.; Szeftel, J.: Global nonlinear stability of Schwarzschild spacetime under polarized perturbations. Annals of
  Mathematics Studies, {\bf 210}. Princeton University Press, Princeton, NJ, 2020. xi+840 pp. 
\bibitem{klinger} Klinger, P.: A new class of asymptotically non-chaotic vacuum singularities. Ann. Physics {\bf 363}, 1--35 (2015)
\bibitem{Lee} Lee, J. M.: Introduction to Smooth Manifolds, Graduate Texts in Mathematics, Springer Verlag, New York, 2013
\bibitem{lea} Liebscher, S.; H\"{a}rterich, J.; Webster, K.; Georgi, M.: Ancient dynamics in Bianchi models: approach to periodic 
cycles. Comm. Math. Phys. {\bf 305}, no. 1, 59--83 (2011)
\bibitem{lrt} Liebscher, S.; Rendall, A. D.; Tchapnda, S. B.: Oscillatory singularities in Bianchi models with magnetic fields.
  Ann. Henri Poincar\'{e} {\bf 14}, no. 5, 1043--1075 (2013)
\bibitem{LK} Lifshitz, E. M.; Khalatnikov, I. M.: Investigations in relativistic cosmology. Adv. Phys. {\bf 12}, 185--249 (1963)
\bibitem{larannals} Lindblad, H.; Rodnianski, I.: The global stability of Minkowski space-time in harmonic gauge. Ann. of Math. (2) {\bf 171},
  no. 3, 1401--1477 (2010)
\bibitem{larcmp} Lindblad, H.; Rodnianski, I.: Global existence for the Einstein vacuum equations in wave coordinates. Comm. Math. Phys.
  {\bf 256}, no. 1, 43--110 (2005)
\bibitem{lat} Lindblad, H.; Taylor, M.: Global stability of Minkowski space for the Einstein-Vlasov system in the harmonic gauge.
  Arch. Ration. Mech. Anal. {\bf 235}, no. 1, 517--633 (2020)
\bibitem{lavk} L\"{u}bbe, C.; Valiente Kroon, J. A.: A conformal approach for the analysis of the non-linear stability of radiation cosmologies.
  Ann. Phys. {\bf 328}, 1--25 (2013)
\bibitem{misner} Misner, C. W.: Mixmaster Universe. Phys. Rev. Lett. {\bf 22} 1071 (1969)
\bibitem{oliynyk} Oliynyk, T. A.: Future stability of the FLRW fluid solutions in the presence of a positive cosmological constant.
  Comm. Math. Phys. {\bf 346}, no. 1, 293--312 (2016)
\bibitem{ren} Rendall, A. D.: Fuchsian analysis of singularities in Gowdy spacetimes beyond analyticity. Class. Quantum Grav. {\bf 17},
  no. 16, 3305--3316 (2000)
\bibitem{raw} Rendall, A. D.; Weaver, M.: Manufacture of Gowdy spacetimes with spikes. Class. Quantum Grav. {\bf 18}, no. 15, 2959--2975 (2001)
\bibitem{cbu} Ringstr\"{o}m,  H.: Curvature blow up in Bianchi VIII and IX  vacuum spacetimes. Class. Quantum Grav. \textbf{17}, 
713--731 (2000)
\bibitem{BianchiIXattr} Ringstr\"{o}m, H.: The Bianchi IX attractor, Annales Henri Poincar\'{e} {\bf 2}, 405--500 (2001)
\bibitem{rininv} Ringstr\"{o}m, H.: Future stability of the Einstein non-linear scalar field system. Invent. Math. {\bf 173}, 123--208 (2008)
\bibitem{rinpl} Ringstr\"{o}m, H.: Power law inflation. Comm. Math. Phys. {\bf 290}, 155--218 (2009)
\bibitem{RinCauchy} Ringstr\"{o}m, H.: The Cauchy problem in General Relativity, European Mathematical Society, Z\"{u}rich (2009)
\bibitem{rintopstab} Ringstr\"{o}m, H.: On the Topology and Future Stability of the Universe. Oxford Univ. Press, Oxford (2013)
\bibitem{RinWave} Ringstr\"{o}m, H.: Wave equations on silent big bang backgrounds. Preprint, 
  \href{https://arxiv.org/abs/2101.04939}{arXiv:2101.04939}. Accepted for publication, Memoirs of the AMS
\bibitem{RinGeo} Ringstr\"{o}m, H.: On the geometry of silent and anisotropic big bang singularities. Preprint,
  \href{https://arxiv.org/abs/2101.04955}{arXiv:2101.04955v2}
\bibitem{RinDoS} Ringstr\"{o}m, H.: Initial data on big bang singularities in symmetric settings. Pure Appl. Math. Q. {\bf 20}, no. 4, 1505--1539 (2024)
\bibitem{rasJEMS} Rodnianski, I.; Speck, J.: The nonlinear future stability of the FLRW family of solutions to the irrotational Euler–Einstein
system with a positive cosmological constant. J. Eur. Math. Soc. {\bf 15}, 2369--2462 (2013)
\bibitem{rasql} Rodnianski, I.; Speck, J.: A regime of linear stability for the Einstein-scalar field system with applications to
  nonlinear big bang formation. Ann. Math. (2) {\bf 187}, no. 1, 65--156 (2018)
\bibitem{rasq} Rodnianski, I.; Speck, J.: Stable big bang formation in near-FLRW solutions to the Einstein-scalar field and 
Einstein-stiff fluid systems. Selecta Math. (N.S.) {\bf 24}, no. 5, 4293--4459 (2018)
\bibitem{rsh} Rodnianski, I.; Speck, J.: On the nature of Hawking's incompleteness for the Einstein-vacuum equations: the regime of 
moderately spatially anisotropic initial data. J. Eur. Math. Soc. (JEMS) {\bf 24}, no. 1, 167--263 (2022)
\bibitem{speck} Speck, J.: The nonlinear future stability of the FLRW family of solutions to the Euler–
  Einstein system with a positive cosmological constant. Selecta Math. (N.S.) {\bf 18}, 633--715 (2012)
\bibitem{speckradiation} Speck, J.: The stabilizing effect of spacetime expansion on relativistic fluids with sharp results for the
  radiation equation of state. Arch. Ration. Mech. Anal. {\bf 210}, 535–579 (2013)
\bibitem{specks3} Speck, J.: The maximal development of near--FLRW data for the Einstein--scalar field system with spatial topology 
$\sn{3}$. Comm. Math. Phys. {\bf 364}, no. 3, 879--979 (2018)  
\bibitem{sta} St\aa hl, F.: Fuchsian analysis of $S^{2}\times S^{1}$ and $S^{3}$ Gowdy spacetimes. Class. Quantum Grav. {\bf 19}, no. 17,
  4483--4504 (2002)
\bibitem{svedberg} Svedberg, C.: Future stability of the Einstein–Maxwell-scalar field system. Ann. Henri
Poincar\'{e} {\bf 12}, 849--917 (2011)  
\bibitem{wea} Weaver, M.: Dynamics of magnetic Bianchi VI${}_{0}$-cosmologies. Class. Quantum Grav. {\bf 17}, no. 2, 421--434 (2000)
\bibitem{zipser} Zipser, N.: The global nonlinear stability of the trivial solution of the Einstein-Maxwell equations.
  Thesis (Ph.D.)–Harvard University. 198 pp. (2000)
\end{thebibliography}
\end{document}